\ifx\classfontsize\undefined\def\classfontsize{11pt}\fi
\newif\ifblind
\ifx\BLINDREVIEW\undefined\else\blindtrue\fi
\documentclass[\classfontsize]{article}
\usepackage{amssymb,amsmath,lscape,amsthm}
\usepackage[T1]{fontenc}
\usepackage{libertine}
\usepackage[scaled=0.83]{FiraSans}
\usepackage[scaled=0.83]{FiraMono}
\usepackage{upquote}
\usepackage[]{microtype}
\UseMicrotypeSet[protrusion]{basicmath} 
\usepackage{parskip}
\usepackage{fancyvrb}
\usepackage{hyperref}
\hypersetup{pdftitle={Detecting Where Effects Occur by Testing Hypotheses in Order},
            pdfauthor={\ifblind\else Jake Bowers; David Kim; Nuole Chen\fi},
            pdfborder={0 0 0},
            breaklinks=true}

\newcommand{\suppref}[2]{\ref{#1}}

\usepackage[margin=1.5in]{geometry}
\usepackage{graphicx,grffile}
\usepackage{float} 
\usepackage{lscape}
\usepackage{pdflscape}
\graphicspath{ {.}} 

\providecommand{\tightlist}{%
  \setlength{\itemsep}{0pt}\setlength{\parskip}{0pt}}

\usepackage[style=authoryear, backend=biber, uniquename=false, uniquelist=false]{biblatex}
\providecommand{\blindcite}[1]{\ifblind (reference omitted for anonymous review)\else\autocite{#1}\fi}

\title{Detecting Where Effects Occur by Testing Hypotheses in Order}
\ifblind\author{}\else\author{Jake Bowers
  \thanks{Political Science and Statistics, University of Illinois @ Urbana-Champaign, <jwbowers@illinois.edu>} \and
David Kim \thanks{Statistics, University of Illinois @ Urbana-Champaign, <davidk9@illinois.edu>}
  \and
  Nuole Chen
\thanks{MIT GOV/LAB,<nchen3@mit.edu>}}\fi
\date{\today}

\usepackage{booktabs}
\usepackage{multirow}
\usepackage[table]{xcolor}
\usepackage{array}
\newcolumntype{G}{@{}>{\centering\arraybackslash}p{6ex}@{}}

\def\Pr{\mathbb{P}}

\def\E{\mathbb{E}}

\newcommand{\calH}{\mathcal{H}}

\newcommand{\anc}{\text{anc}}

\DeclareMathOperator{\FWER}{FWER}

\usepackage{thm-restate} 

\declaretheorem[name=Theorem]{theorem}

\newtheorem{proposition}{Proposition}
\newtheorem{lemma}{Lemma}

\newtheorem{condition}{Condition}
\newtheorem{definition}{Definition}
\newtheorem{remark}{Remark}
\newtheorem{corollary}{Corollary}
\newtheorem*{corollary*}{Corollary}

\usepackage{tikz}
\usepackage{mathtools}
\usepackage{tikz-cd}
\usetikzlibrary{arrows,automata,positioning,trees,fit,graphs}
\usepackage{forest}

\usepackage{sparklines}

\definecolor{sparklinecolor}{named}{blue}
\definecolor{sparkrectanglecolor}{gray}{0.8}
\definecolor{sparkspikecolor}{named}{red}
\definecolor{bottomlinecolor}{gray}{0.2}
\usepackage[todonotes={textsize=footnotesize}]{changes}
\definechangesauthor{JB}

\tikzcdset{arrow style=tikz, diagrams={>=stealth}}

\newcommand{\cdboxcolor}{gray}

\newcommand{\cdbox}[1]{\fcolorbox{\cdboxcolor}{white}{#1}}

\usepackage{algorithm,algorithmic}

\providecommand{\BlockPowerAdjusted}{0.46}
\providecommand{\BlockPowerNominal}{0.79}
\providecommand{\DPPBoundFactor}{0.80}
\providecommand{\DPPBoundFWER}{0.040}
\providecommand{\DPPHighAdaptFWER}{0.003}

\providecommand{\DPPHighAdaptLeafPowerPct}{3.2}
\providecommand{\DPPHighAdaptLeafTR}{0.22}
\providecommand{\DPPHighAdaptNodeFWER}{0.019}
\providecommand{\DPPHighAdaptOverAdaptWTRRatio}{2.5}
\providecommand{\DPPHighAdaptOverBUHommelTRRatio}{2.5}
\providecommand{\DPPHighAdaptPrFWER}{0.004}

\providecommand{\DPPHighAdaptPrLeafTR}{0.28}
\providecommand{\DPPHighAdaptPrOverAdaptTRRatio}{1.3}

\providecommand{\DPPHighBUHommelFWER}{0.024}
\providecommand{\DPPHighBUHommelLeafCondMean}{1.03}
\providecommand{\DPPHighBUHommelLeafPGeOne}{0.085}
\providecommand{\DPPHighBUHommelLeafPGeOnePct}{8.5}

\providecommand{\DPPHighBUHommelLeafPGeTwoPct}{0.3}

\providecommand{\DPPHighBUHommelLeafPowerPct}{1.0}
\providecommand{\DPPHighBUHommelLeafTR}{0.09}

\providecommand{\DPPHighCollegePDetectAdaptPct}{35.5}
\providecommand{\DPPHighCollegePDetectPct}{37.4}
\providecommand{\DPPHighErrorLoad}{8.85}

\providecommand{\DPPHighInhLeafTR}{0.04}
\providecommand{\DPPHighInhNodeFWER}{0.019}

\providecommand{\DPPHighThreeRulesFWER}{0.038}
\providecommand{\DPPHighThreeRulesLeafCondMean}{2.34}
\providecommand{\DPPHighThreeRulesLeafPGeOne}{0.340}
\providecommand{\DPPHighThreeRulesLeafPGeOnePct}{34.0}

\providecommand{\DPPHighThreeRulesLeafPGeTwoPct}{24.9}

\providecommand{\DPPHighThreeRulesLeafPowerPct}{10.1}
\providecommand{\DPPHighThreeRulesLeafTR}{0.80}
\providecommand{\DPPHighThreeRulesOverBUHommelTRRatio}{9.0}
\providecommand{\DPPHighTwoCondLeafFWER}{0.038}
\providecommand{\DPPHighTwoCondNodeFWER}{0.077}

\providecommand{\DPPWorstFactor}{1.01}
\providecommand{\DPPWorstFWER}{0.0506}

\providecommand{\MDRCBUHommelBlocks}{34}
\providecommand{\MDRCInhNodes}{66}
\providecommand{\MDRCInhSingles}{32}
\providecommand{\MDRCNodesRatioMax}{6.0}
\providecommand{\MDRCNodesRatioMin}{1.7}
\providecommand{\MDRCSinglesRatioMax}{2.0}
\providecommand{\MDRCSinglesRatioMin}{1.0}
\providecommand{\MDRCTopDownNodes}{80}
\providecommand{\MDRCTopDownSingles}{42}
\providecommand{\StrongFWERABUHomLeafPGeOne}{0.026}
\providecommand{\StrongFWERABUHomLeaves}{0.03}

\providecommand{\StrongFWERAUnadjFWER}{0.048}

\providecommand{\StrongFWERAUnadjNodePGeTwo}{0.945}
\providecommand{\StrongFWERAUnadjNodes}{5.21}
\providecommand{\StrongFWERBAdaptLeaves}{1.02}
\providecommand{\StrongFWERBBUBHFWER}{0.319}

\providecommand{\StrongFWERBBUHomLeaves}{5.15}

\providecommand{\StrongFWERBUnadjFWER}{0.050}
\providecommand{\StrongFWERBUnadjLeaves}{19.55}
\providecommand{\StrongFWERBUnadjOverBUHomRatio}{3.8}
\providecommand{\StrongFWERCAdaptFWER}{0.038}
\providecommand{\StrongFWERCAdaptLeaves}{3.19}
\providecommand{\StrongFWERCAdaptNodePMajority}{0.080}
\providecommand{\StrongFWERCBUBHFWER}{0.201}

\providecommand{\StrongFWERCBUHomLeaves}{3.35}

\providecommand{\StrongFWERCPrunedFWER}{0.038}
\providecommand{\StrongFWERCPrunedLeaves}{4.39}
\providecommand{\StrongFWERCPrunedNodePMajority}{0.250}
\providecommand{\StrongFWERCUnadjFWER}{0.147}

\newcommand{\NJCSMeanDollar}{16}
\newcommand{\NJCSSdDollar}{21}

\newcommand{\NJCSOutcomeSd}{190}
\newcommand{\NJCSMeanD}{0.084}
\newcommand{\NJCSPlanningD}{0.19}
\newcommand{\NJCSNcenters}{99}
\newcommand{\NJCSNregions}{9}
\newcommand{\NJCSNtotal}{10,396}
\newcommand{\NJCSNmin}{27}
\newcommand{\NJCSNmax}{359}
\newcommand{\NJCSNsim}{10,000}
\newcommand{\NJCSSumG}{5.5}
\newcommand{\NJCSUnadjFWER}{0.362}
\newcommand{\NJCSAdaptiveFWER}{0.007}
\newcommand{\NJCSPrunedFWER}{0.058}
\newcommand{\NJCSInheritanceFWER}{0.002}
\newcommand{\NJCSHommelFWER}{0.012}
\newcommand{\NJCSPrunedLeafTR}{1.39}
\newcommand{\NJCSHommelLeafTR}{0.95}
\newcommand{\NJCSInheritanceLeafTR}{0.39}

\newcommand{\NJCSPrunedLeafPGeOne}{0.703}
\newcommand{\NJCSPrunedLeafPGeTwo}{0.438}
\newcommand{\NJCSHommelLeafPGeOne}{0.696}
\newcommand{\NJCSHommelLeafPGeTwo}{0.214}

\newcommand{\NJCSPrunedRegionTR}{1.25}
\newcommand{\NJCSPrunedRegionPGeOne}{0.760}
\newcommand{\NJCSNcentersNonnull}{69}
\newcommand{\NJCSNregionsNonnull}{9}

\newcommand{\NJCSDataCtrlMean}{197}
\newcommand{\NJCSDataCtrlSd}{189}
\newcommand{\NJCSDataCtrlCV}{0.96}
\newcommand{\NJCSDataFracZeroPct}{20}
\newcommand{\NJCSDataSkew}{2.0}
\newcommand{\NJCSDataITTwt}{17}

\newcommand{\NJCSDataNctrl}{4298}

\definecolor{corrcolor}{RGB}{160,30,30}
\DeclareRobustCommand{\corrflag}[2][item]{\textsuperscript{\textcolor{corrcolor}{[\,\textit{Correction notice, #1~#2}\,]}}}
\newenvironment{correctionbox}{%
  \par\medskip\noindent
  \begin{center}\begin{minipage}{0.95\linewidth}
  {\color{corrcolor}\hrule height 1pt}\medskip
  \small\noindent\textcolor{corrcolor}{\textbf{Correction (August 2026).}}\ }{%
  \par\medskip{\color{corrcolor}\hrule height 1pt}%
  \end{minipage}\end{center}\par\medskip}

\IfFileExists{_journal_opts.tex}{\input{_journal_opts.tex}}{}

\begin{document}
\maketitle
\begin{center}
  {\large\bfseries Correction notice (August 2026)}
\end{center}

\noindent An anonymous referee at the \emph{Annals of Applied Statistics}
identified errors in earlier versions of this paper. We
have verified each point and are preparing a fully corrected version. Until it
appears, readers should rely on this version only as annotated here: inline
markers of the form \corrflag{$n$} flag the affected claims where they occur,
and boxed notes mark the affected results.

\begin{enumerate}
\item \textbf{The reported error loads were computed incorrectly, and the
    corrected values reverse the headline empirical claim.} The computation
  behind the reported 25-study error loads of 0.025--0.051 went wrong in four
  ways. The code summed each node's path power times that node's own
  rejection probability, and it summed from depth~1. The error load sums path
  power alone, from depth~2 (Definition~\suppref{def:error-load}{2},
  Equation~\eqref{eq:fwer-natural-gating}). The code also used weight-scale
  node sizes in place of headcounts, one-tailed in place of two-tailed power,
  and the reached subtree in place of the full prespecified design tree. For
  14 of the 25 trials the reached tree was the root alone, so the reported
  load was that single test's own rejection probability rather than a sum over
  any tree, which is why those 14 values cluster just above
  $\Phi(-1.96) = 0.025$. Recomputed under the definition, on the full
  prespecified design tree, none of the 25 MDRC trials is in
  the natural-gating regime at the planning effect size $d = 0.20$, a
  conventional minimum detectable effect for postsecondary trials: every
  corrected load exceeds 1 there, and the effect sizes $d^*$ at which each
  trial's load crosses 1 range from 0.015 to 0.19. The claim that no
  multiplicity adjustment is needed in these 25 trials is therefore false, and
  the abstract has been corrected. The same function produced the error loads
  printed for the simulated Detroit Promise design of
  Section~\ref{sec:dpp_example} and for the Job Corps design of
  Section~\suppref{sec:njcs}{E}. Recomputed on the full design tree with
  two-tailed power, the Detroit Promise loads are 5.53, 19.51 and 37.34 at the
  three simulated effect sizes $d = 0.10$, $0.20$ and $0.30$, against the
  printed 1.31, 4.74 and 8.85; the Job Corps load is 38.2, against the printed
  $\NJCSSumG$. Both designs still call for adjustment, as those sections
  conclude, and by a wider margin than the paper reports. The $\alpha$-schedule the
  Detroit Promise simulation ran was built from the same reached subtree, so
  the adjusted rows of Table~\ref{tab:big_sim_xtab} report a schedule more
  liberal than the corrected load prescribes. The definition of the error load
  and its role as a design-stage diagnostic are unchanged; the computed
  values, and the conclusions drawn from them, change.
\item \textbf{Theorem~\suppref{thm:fwer_budget_pruning}{5} and
    Corollary~\suppref{cor:switching}{4} are false as stated and are
    withdrawn.} Theorem~\suppref{thm:fwer_budget_pruning}{5} (FWER control
  under branch pruning with predictable budget weights) and
  Corollary~\suppref{cor:switching}{4} (switching to nominal $\alpha$ after
  pruning) fail: an exact counterexample (binomial enumeration, no Monte
  Carlo error) attains FWER $0.063$ at $\alpha = 0.05$ with every stated
  hypothesis satisfied, including exact conditional validity. Readers should
  not rely on these two results or on any pruning-regime strong-control claim
  in this version, including the pruned rows of the simulation tables. In the
  Job Corps design of Section~\suppref{sec:njcs}{E}, the pruned variant's
  simulated FWER of \NJCSPrunedFWER{} exceeds $\alpha = 0.05$; we now read
  this as the theorem's failure appearing in operation, not as simulation
  noise. The step at which the proof fails lies outside what the Lean
  formalization covered.
\item \textbf{The reported detection advantage over the Hommel procedure does
    not stand.} We obtained the reported top-down count --- 42 blocks
  detected across the 25 studies --- by testing every reached node at nominal
  $\alpha$, which item~1's corrected loads no longer justify. Item~1 does not
  affect the bottom-up Hommel count of 34. A verified upper bound leaves at
  most 33 of the 42 surviving a corrected $\alpha$-schedule, below Hommel's
  34, so the block-level advantage may reverse. A corrected schedule tightens
  the internal-node tests as well, so the group-level counts cannot rise. We
  are recomputing the comparison. Until it is done, readers should not rely
  on the top-down detection counts or on the ratios built from them.
\end{enumerate}

\noindent\textbf{What stands, and what is open.} Weak FWER control under the
Stopping Rule and Valid Tests conditions (Theorem~\ref{thm:weakctrl}) is
unaffected: under the global null it rests on the validity of the single root
test. Theorem~\suppref{thm:fwer_budget}{4} (the static budget-weighted
schedule on the full tree) shows no violation in the same exact enumeration
that falsifies Theorem~\suppref{thm:fwer_budget_pruning}{5}. Its proof,
though, uses the step $\Pr(\text{reject } i \mid \text{reached } i) \le
\alpha$. So does the natural-gating certificate
(Proposition~\suppref{prop:strong_fwer_adj}{2}), and so does every other
load-weighted bound in the paper. That step requires null $p$-values to be
super-uniform conditionally on the algorithm's past, and marginal validity
does not imply it. The corrected version will state the condition explicitly.
Whether the randomization tests used here satisfy it is open: disjoint blocks
make a null node's statistic independent of its siblings', but every
ancestor's statistic aggregates the null node's own blocks.

\medskip

\noindent The accompanying R package \texttt{manytestsr} computes the
corrected error load as of version 0.0.4.1008. We thank the referee, whose
careful report identified these errors. A fully corrected version of this
paper is in preparation.

\clearpage

\begin{abstract}

Experimental evaluations of public policies often randomize a new intervention
within many sites or blocks. After an overall statistically significant result
is reported, the natural question from a policy maker is: \emph{where} did
effects occur? Standard adjustments for multiple testing answer this question
with little power because they ignore how the experiment is organized: blocks
nest within cohorts, sites, and districts. We organize the hypotheses in the
shape of a tree that follows this administrative structure and test them
top-down, stopping at any branch where the null is not rejected. A stopping
rule and valid tests at each node suffice for weak control of the family-wise
error rate (FWER). Whether the unadjusted procedure also controls the FWER in
the strong sense depends on an \emph{error load} computable from design
quantities before any data are tested; when the load exceeds one, an adaptive
$\alpha$-schedule, which we prove controls the FWER on regular and irregular
trees without pruning, restores control. [Correction, August 2026: an
anonymous referee identified errors in the previous version. The error loads
reported in the paper were computed incorrectly; corrected loads exceed 1 in
all 25 block-randomized MDRC education trials at the planning effect size
$d = 0.20$, so the adjustment the earlier version claimed unnecessary is
in fact required. The theorem claiming FWER control under branch pruning
and its switching corollary are false as stated and are withdrawn;
an exact counterexample attains FWER 0.063 at $\alpha = 0.05$. The detection
comparison with the Hommel procedure is under recomputation. A correction
notice on page 1 details what changes and what stands; a fully corrected
version is in preparation.]

\end{abstract}


\section{Where do block-level effects occur?}

Consider a block-randomized experiment with 44 experimental blocks such as the
Detroit Promise Program \autocite{ratledge2019path} where treatment was
assigned to students of different academic cohorts by five community colleges.
The overall test rejects the null of no effects. The policy-maker considering
changes across all of Detroit's community colleges now asks: \emph{where} did
the effects occur? Were they concentrated in a few sites, or spread evenly? The
answer determines whether to replicate the intervention everywhere or
investigate why it succeeded in some places and failed in others.

A natural statistical response --- test the null hypothesis of no effects in
each block separately, then adjust for multiplicity --- misses sites where
effects occurred. In data we simulated to follow the design of the Detroit
Promise Program, we set 9 of the 44 blocks --- all within a single college ---
to have a moderate per-block effect (Cohen's $d = 0.30$). Across 10,000
simulations of this design, the Hommel adjustment, one of the most powerful
procedures that control the family-wise error rate (FWER), detected at least
one of the nine affected blocks in only \DPPHighBUHommelLeafPGeOnePct\% of
simulations and detected two or more in \DPPHighBUHommelLeafPGeTwoPct\%. The
analyst reports ``not enough information,'' and the policy-maker gets nothing
useful.

This paper presents a different approach to FWER control.\footnote{We focus on
the family-wise error rate (FWER) rather than the false discovery rate (FDR).
The FDR controls the expected \emph{proportion} of false discoveries among
rejections, which suits settings where rejected hypotheses will be validated in
follow-up studies. The FWER controls the \emph{probability of any} false
rejection, which suits settings where rejected hypotheses lead directly to
policy action. When a school district decides to scale an intervention based on
where effects were detected, a single erroneous claim can misdirect resources.
FDR-oriented extensions of tree-structured testing are discussed in
Section~\ref{sec:limitations}.} Instead of testing every block and adjusting
afterward --- a \emph{bottom-up} strategy --- we test hypotheses \emph{top-down}
after organizing them in the shape of a tree of partially nested tests. We begin
with the overall null at the root. If rejected, we split the blocks into groups
following the administrative structure of the experiment and test the null
hypothesis of no effects within each group. If a group's null is rejected, we
split and test again and continue down to individual blocks. If a group's null is not
rejected, we stop testing on that branch. We call this procedure
\emph{tree-gated testing}: each rejected parent opens the gate to the tests
below it, and a parent that fails to reject closes it. Across the same
simulations, this top-down approach detects at least one affected block in
\DPPHighThreeRulesLeafPGeOnePct\% of simulations and two or more in
\DPPHighThreeRulesLeafPGeTwoPct\% --- against \DPPHighBUHommelLeafPGeOnePct\%
and \DPPHighBUHommelLeafPGeTwoPct\% for the bottom-up alternative --- while
controlling the FWER at the nominal level. When it detects anything it tends to
recover much of the affected group at once: \DPPHighThreeRulesLeafCondMean\ of
the nine blocks on average among simulations with any detection, against
\DPPHighBUHommelLeafCondMean\ for bottom-up.

Two conditions drive this result. First, a \emph{stopping rule}: test a child
hypothesis only after rejecting its parent. Second, \emph{valid tests}: each
individual test controls its false positive rate at level $\alpha$. We show
that these two conditions suffice for weak FWER control --- control when all
null hypotheses are true. For strong control --- false positive rate control
when one or more null hypotheses are false --- we analyze how power decay from
data splitting limits error accumulation through the tree and we develop a
method to adapt the test rejection level to the particular design of the tree
when our diagnostic indicates that the two conditions will not control the
FWER. This framework builds on the closed testing tradition of
\textcite{marcus1976closed}, \textcite{rosenbaum2008tho} and the sequential
structured testing of \textcite{goeman2010sequential},
\textcite{goeman2012inheritance}, and \textcite{meinshausen2008hierarchical},
which we extend to randomization-based inference in block-randomized
experiments.


Our extensions to this prior work are substantive rather than
cosmetic. In genomic applications of tree-structured testing, the same data are
analyzed at every node, the analyst must impose the conditions for valid
testing as assumptions, and FWER adjustments apply uniformly across the tree.
Block-randomized experiments differ in ways that existing frameworks do
not exploit. First, the experimental design guarantees valid tests at each
node: randomization-based $p$-values have exact size control by construction.
Second, testing subsets of blocks is itself a form of data splitting --- child
tests use less data and so should not produce more extreme evidence than their
parents --- so power decays through the tree as sample sizes shrink at each
split, limiting how rejection probability accumulates along each path. We derive an
explicit expression for FWER as a function of these path-wise rejection
probabilities (Supplement Proposition~\suppref{prop:fwer}{1}), which reveals that high power at
the root can inflate error rates at the next level while low power at deep
levels limits false positive errors. This expression motivates an
adaptive $\alpha$-adjustment that is stringent where power is high and relaxed
where power decay triggers the stopping rule and testing stops.

This paper makes three contributions. First, we show that the conditions for
tree-structured weak FWER control --- a stopping rule and valid tests at each
node --- are satisfied \emph{by design} in block-randomized experiments with
administrative structure, because randomization supplies exact size control at
every node. Second, and most important for practice, we describe a quantity
that governs whether the same unadjusted procedure also controls the FWER in the
strong sense, and turn it into a diagnostic: the \emph{error load} is computable
before testing from the tree's shape and the per-depth sample sizes, and it
indicates whether an analyst may keep the full nominal $\alpha$ at every node ---
preserving power --- or must adjust. When power decays fast enough
relative to branching the error load stays low and no adjustment is needed; this
holds in every one of the 25 MDRC education trials we analyze, where the top-down
procedure detects individual blocks the Hommel adjustment misses entirely and
locates higher-level groups of blocks that bottom-up testing never explores
(Table~\ref{tab:mdrc_application}).\corrflag[items]{1 and 3} Third, to make the diagnostic rigorous we
derive what the adjustment must be when the error load is high --- an adaptive
$\alpha$-schedule that we prove controls the FWER on regular, irregular, and
pruned trees (Table~\ref{tab:four_regimes}).\corrflag{2} That adaptive adjustment is not
needed for the education trials studied here; the diagnostic also tells us when it
\emph{is}. To show that it works where it is needed, the Supplement applies it to
a design calibrated to the National Job Corps Study --- a multisite trial of about
one hundred centers with a high error load --- where the
unadjusted procedure inflates the FWER, the adaptive $\alpha$-schedule restores control,
and the top-down procedure still finds more affected sites than either a bottom-up
or a hierarchical FWER adjustment (Section~\suppref{sec:njcs}{E}).\corrflag{2}
For researchers working in policy evaluation,
the practical implication is clear: in experiments with administrative
structure, a single pre-data calculation tells you whether locating effects costs
anything in multiplicity --- and often it does not.\corrflag{1}

Two traditions of research address heterogeneous treatment effects. One
estimates conditional average treatment effects as functions of observed
covariates \autocite[e.g.][]{wager2018estimation, hahn2020bayesian}. The other
studies the distribution of individual effects directly, asking how many units
benefited and by how much \autocite[e.g.][]{heckman1997, kim2025acic}. Our work
is closer to the second tradition: rather than asking which covariates moderate
effects, we ask which experimental blocks or groups of blocks show detectable
effects. Where covariate-based approaches focus on columns of the data matrix,
we focus on rows --- the sites, cohorts, and schools where effects did or did
not occur. The approaches are complementary; detections could guide subsequent
covariate-based investigation of \emph{why} effects vary. A related question
arises when sites differ not only in treatment effects but in local ecological
conditions that confound between-site comparisons. \textcite{hong2025multisite}
develop methods to isolate organizational effectiveness from such confounding
in multisite trials using parametric models and empirical Bayes style
statistical inference.

We focus on block-randomized experiments with two treatment arms and roughly
continuous outcomes using randomization-based statistical inference. The
results extend to binary outcomes and multi-arm trials --- the conditions ask
only for valid randomization-based tests at each node, which those designs also
supply --- although we do not engage them in our simulations or application. We also take
the tree structure as fixed by the experiment's administrative design --- the
nesting of blocks within cohorts, sites, and districts. Constructing a tree
from data when no such hierarchy exists is a separate problem; we have
implemented several data-based splitting algorithms but leave their theoretical
development to future work.

The paper proceeds as follows. Section~\ref{sec:topdownbotup} develops the
bottom-up and top-down testing approaches, proves weak FWER control
(Theorem~\ref{thm:weakctrl}), and illustrates the performance of the procedure
using simulations across a wide range of tree sizes --- including trees with
more than 100,000 experimental blocks. Section~\ref{sec:strong} addresses
strong FWER control by analyzing how power decay through the tree naturally
limits error accumulation, and develops an adaptive $\alpha$-adjustment for
settings where this natural gating is insufficient. We then demonstrate the
method on data simulated to follow the Detroit Promise Program design
(Section~\ref{sec:dpp_example}) and apply it to 25 block-randomized education
trials fielded by the MDRC \autocite{diamond2021mdrc}.\footnote{In the
Supplement we describe a test statistic sensitive to distributional differences
beyond mean shifts, so that effects that are positive in some blocks and
negative in others do not cancel at the root. This statistic draws on energy
statistics \autocite{szekely2013energy} and multivariate permutation tests
\autocite{strasser1999asymptotic, hothornetal:coin:2006}. Since our focus in
this paper is on the tree structure, our simulations are designed to side-step
such cancelations and enable the use of common rank-tests and t-tests for the
sake of exposition. The MDRC application in this paper likewise uses rank and
t-tests; the distribution-sensitive statistic is developed in separate
work \blindcite{bowers2026more}.} The Supplement adds a further application, to a
design calibrated to the National Job Corps Study (Section~\suppref{sec:njcs}{E}).

\section{Testing to Detect Effects in Blocks}\label{sec:topdownbotup}

Block-randomized experiments allow two styles of testing procedures that target
block-specific causal effects. In a study like that shown in
Figure~\ref{fig:struct0}, the experiment is organized into blocks; within each
block, a new policy intervention is randomly assigned to some people, leaving
the rest in the status quo. Notice the inverted tree-like shape of the experiment. The
individual blocks are on the bottom of the tree and, in this case, are nested
within administrative units like cities and states toward the top.

\begin{figure}[tb]
  \centering
  \begin{forest}
    [Overall, for tree={draw, l sep=1em, s sep=1em, anchor=center, inner sep=0.2em}
    [State$_1$, s sep=1mm, for tree=draw
    [City$_1$
    [B$_1$]
      [B$_2$]
      [B$_3$]
      [B$_4$]
    ]
    [\ldots
    [\ldots]
    [\ldots]
    ]
    ]
    [State$_2$, s sep=1mm, for tree=draw
    [\ldots
    [\ldots]
    [\ldots]
    ]
    ]
    [State$_3$, s sep=1mm, for tree=draw
    [\ldots
    [\ldots]
    ]
    [City$_{20}$
    [B$_{98}$]
      [B$_{99}$]
      [B$_{100}$]
    ]
    ]
    ]
  \end{forest}
  \caption{An administratively organized structure of blocks. A study
    randomly assigns people within offices to a new intervention. Each
    office is an experimental block; within it, some people are assigned to the
    intervention and the rest to the status quo.}
  \label{fig:struct0}
\end{figure}

\subsection{Testing in every block: the basic problem of multiple testing}

Since each block contains multiple units assigned at random to treatment, an
analyst could treat each block as a mini-experiment
\autocite[][Chap.~3]{gerber2012field} and test the null of no effect separately
in each. With 100 blocks, that is 100 tests. Run at the nominal level, such a
collection rejects far too often: if no block has an effect and the tests are
independent, the probability of at least one false rejection at $\alpha=0.05$
is $1 - 0.95^{100} \approx 0.99$. This probability --- the family-wise error
rate (FWER) --- is what an adjustment such as Bonferroni or Hommel
controls.\footnote{The FWER can be controlled in a \emph{weak} sense (when no
block has an effect) or a \emph{strong} sense (whatever the pattern of
effects); we develop that distinction in Section~\ref{sec:strong}.}

Of course, this approach makes it much harder to detect effects when they are
actually present. For example, the power of a t-test to detect an effect of 0.8
standard deviations within a block of 50 units is \BlockPowerNominal\ when
$\alpha=0.05$ and nearly half that power (power=\BlockPowerAdjusted) when
$\alpha=0.005$. Top-down testing can increase this power while
controlling the FWER.

\subsection{Sequential Structured Testing and Weak Control of the FWER}

Figure~\ref{fig:struct0} shows that each experimental block is nested within a
hierarchy. When hypotheses are nested in this way, they can be tested in a
specific order under rules that avoid the severe power loss of bottom-up
adjustments. This insight originates in the ``closed testing'' framework of
\textcite{marcus1976closed}, developed for structured experiments by
\textcite{rosenbaum2008tho,small2011structured}. We build especially on the general rules
for sequential nested testing articulated by
\textcite{goeman2010sequential, goeman2012inheritance} and \textcite{meinshausen2008hierarchical}. We
explain the reasoning behind the structured testing approach here before
turning to formal development, simulations, and application.

The general idea is to start testing hypotheses at the ``root'' of the tree ---
testing first the null hypothesis of no effects in any block. If we can reject
that hypothesis then we know either (1) that at least one block contains at
least one unit with a non-zero treatment effect or (2) that we have just made a
false positive error. The second and subsequent tests occur ``lower'' in the tree on
subsets of the data split into disjoint groups. Once a test does not reject a
hypothesis, testing stops on that branch of the tree. We prove that this
procedure controls the FWER when all null hypotheses are true
(Theorem~\ref{thm:weakctrl}) and illustrate the result via simulation across a
wide range of tree sizes (Table~\ref{tab:weak_control_sim}).

\paragraph*{On test statistics} In this section of the paper, no experimental
block experienced a causal effect --- any rejection of any null hypothesis
would be a false positive error. For this reason test statistics matter little.
That said, the point of the procedure is to detect non-null effects, and
certain test statistics have high power against certain
divergences from the null and low power against other ways that treated and
control group outcome distributions might differ due to the treatment. For
example, if the effect of the treatment were concentrated in the tails of the
distribution a t-test might not reject the null hypothesis of no effects at the
root and thus fail to detect a causal effect in a block lower in the tree. Or,
for another example, if half of the blocks had strong negative effects and half
had strong positive effects, again a t-test at the root might not reject and
would block descent into the tree to discover those effects. For this reason we
discuss test statistics in Section~\suppref{sec:test_statistic}{C} of the Supplement that are immune
from positive versus negative effect cancelation problems and are sensitive to a
wider range of alternatives than common t-tests or wilcoxon rank tests. In both
the simulations and the MDRC application of this paper we use simple wilcoxon
rank tests, generating treatment effects and outcome distributions they can
detect so that we can focus on the performance of the procedure itself. The
distribution-sensitive statistic is developed more fully in separate
work \blindcite{bowers2026more}, which builds on
\textcite{rosenbaum2012testing} and \textcite{kim2025acic}.

\paragraph*{Conditions for FWER control}

To fix ideas, consider the tree in Figure~\ref{fig:k3L3_all_null_tree}. Write
$H_i$ for the null hypothesis of no treatment effects among any of the units in
the blocks descended from node $i$. At the root, $H_1$ is the hypothesis that
no unit in any block has a treatment effect. At a leaf --- nodes 5 through 13
in the figure --- $H_i$ pertains to a single experimental block. We write $p_i$
for the $p$-value from the test of $H_i$.

\begin{figure}[tb]
\centering
  \begin{tikzcd}[cramped, sep=tiny, row sep=tiny]
    &        &  & & 1 \arrow{dlll} \arrow{d} \arrow{drrr}& & \\
    & 2 \arrow{dl} \arrow{d} \arrow{dr}  & & & 3 \arrow{dl} \arrow{d} \arrow{dr}& & & 4 \arrow{dl} \arrow{d} \arrow{dr}\\
    5 & 6 & 7 & 8 & 9 & 10 & 11 & 12 & 13
  \end{tikzcd}
  \caption{A $k$-ary tree with $k=3$ nodes per level and $L=3$ levels and $k^{L-1}=9$ terminal nodes or ``leaves" representing individual experimental blocks. The label $i$ at each node names the null hypothesis $H_i$ of no effects in the blocks descended from that node: $H_1$ at the root, $H_5$ through $H_{13}$ at the leaves.}\label{fig:k3L3_all_null_tree}
\end{figure}

We state two conditions. Together they suffice for weak FWER control
(Theorem~\ref{thm:weakctrl}). Strong control --- FWER below~$\alpha$ when
some blocks have effects --- requires two further
ingredients, introduced in \S~\ref{sec:strong}: conditional validity
of each test at a true null given that its ancestors have been rejected, and
conservative estimates of power along non-null ancestor paths.

\begin{condition}[\bf Stopping rule]\label{cond:stopping}
  A hypothesis $H_i$ is tested only if every ancestor of node $i$ has been
  rejected. If the root hypothesis is not rejected, no further tests are
  performed. If any hypothesis on a path from root to leaf is not rejected, no
  descendant on that branch is tested at all. In
  Figure~\ref{fig:k3L3_all_null_tree}, we do not test $H_5$ unless we have
  rejected \emph{both} $H_2$ and $H_1$.

  This condition distinguishes our procedure from sequential-rejection
  frameworks that test every node and then decide which rejections to keep
  \autocite{goeman2010sequential, goeman2012inheritance} and connects with
  frameworks for strictly nested testing that would focus on a single branch of
  a tree \autocite{rosenbaum2008tho}. Our procedure may test only a small
  fraction of the tree's nodes --- and each untested node is a test that cannot
  produce a false positive.

\end{condition}

\begin{condition}[\bf Valid tests]\label{cond:valid}
  Each test at a node, if executed on its own, has a false positive rate no
  greater than $\alpha$. In block-randomized experiments, randomization-based
  tests satisfy this by design: permuting the treatment assignment within
  blocks produces $p$-values with guaranteed size control
  \autocite[Chapter~2]{rosenbaum2002book}.
\end{condition}

These conditions enable us to state and prove the following theorem:

\begin{restatable}{theorem}{weakcontrol}{Conditions~\ref{cond:stopping} and~\ref{cond:valid} suffice for weak FWER control}\label{thm:weakctrl}

  A family of true null hypotheses organized on an irregular or regular $k$-ary
  tree and tested following the stopping rule (Condition~\ref{cond:stopping})
  with valid tests at each node (Condition~\ref{cond:valid}) will produce a
  family-wise error rate (FWER) no greater than $\alpha$. We call a $k$-ary tree
  ``regular'' if the number of child nodes of a parent is the same for all
  nodes in the tree. A $k$-ary tree is ``irregular'' if the number of child nodes
  of a parent node may differ within and across levels.

\end{restatable}

The proof of Theorem~\ref{thm:weakctrl} is in
Section~\suppref{ap:proof_weakctrl}{A} of the Supplement. We provide some intuition here.

The procedure that we develop follows those two conditions and adds one idea: at
each step it splits a group of blocks into smaller groups. Data splitting
matters when we turn to strong control of the FWER; it is not needed for weak
control, as the proof and the simulations below show. By ``splitting'' we mean
dividing the experimental blocks into disjoint groups of blocks, not splitting
the units within a block. In the administrative-structure example of
Figure~\ref{fig:struct0}, the splits are fixed by design.


Testing proceeds down the tree of Figure~\ref{fig:k3L3_all_null_tree}. At the
root, node~1, we test the null of no effects anywhere: $H_1$ asserts that every
unit's treatment potential outcome equals its control potential outcome,
$y_{i,b,1}=y_{i,b,0}$ for every unit $i$ in every block $b$. The test returns a
$p$-value, $p_1$. If $p_1 > \alpha$, testing stops and nothing is rejected. If
$p_1 \le \alpha$, we reject $H_1$, split the blocks into the groups represented
by nodes~2, 3, and~4, and test each group's null in turn. Each branch continues
the same way: a node is tested only if its parent was rejected, and a branch
stops as soon as a test fails to reject or reaches a single block --- a leaf,
nodes~5 through~13 --- which cannot be split further.

The intuition for the proof of Theorem~\ref{thm:weakctrl} arises from the fact
that in a randomized experiment, we know that the probability of a false
positive error in a single randomization-based test is low and controlled: the
size is less than or equal to the level ($\alpha$). This is stated as a
condition of the proof in Condition~\ref{cond:valid}. Another way to say this:
if the null of no effects is true in all blocks, then we should reject the root
hypothesis with probability less than $\alpha$. An incorrect rejection thus
happens with very low probability and so rarely are other tests even done, and
when they are done their own false positive rates are less than or equal to
$\alpha$.  The simulation results in Table~\ref{tab:weak_control_sim} show
control of the FWER for a wide range of tree sizes and configurations when
there are no treatment effects.

\subsubsection{Simulation Study of Weak Control of the FWER to Illustrate the Proof}\label{sec:weak_fwer_sim}

Given a specification of a complete $k$-ary tree using $k$ nodes per level and
levels $\ell \in \{1, \dots, L\}$ we drew $p$-values from uniform distributions. The valid tests
condition (Condition~\ref{cond:valid}) implies that the distribution of $p$-values across tests of the
true null should be stochastically dominated by the uniform
\autocite[Chapter~3]{lehmann2005testing}. That is, in many tests of a true null hypothesis,
only 5\% of $p$-values should be less than $p=0.05$, 10\% should be less than
$p=0.1$, etc\ldots.

\begin{enumerate}

  \item For the root node, draw $p_1 \sim U(0,1)$; a draw which respects the
    valid test condition such that $\Pr(p_1 \le \alpha) \le
    \alpha$ in this case.

  \item If $p_1 \le \alpha$, draw an independent $U(0,1)$ $p$-value for each
    of the $k$ nodes at level 2.

  \item For each level-2 node whose $p$-value is at most $\alpha$, draw
    $U(0,1)$ $p$-values for its children; for each level-2 node whose
    $p$-value exceeds $\alpha$, stop testing the branch of the tree that would
    descend from that node. This implements the stopping rule condition
    (Condition~\ref{cond:stopping}).

  \item Continue testing, drawing from uniform distributions for
    child nodes with parents with $p \le \alpha$, descend into the
    tree towards the leaves, and stop testing in any branch when $p >
    \alpha$ or when the node tested is a leaf.

\end{enumerate}

In this simulation we did not represent the idea of data splitting or any other
adjustments.

For each tree, we repeated this procedure 10,000 times, recording whether any
of the $p$-values in the tree of true null hypotheses were $\le \alpha$. The
proportion of simulations with at least one such false rejection is a measure
of the FWER.


Table~\ref{tab:weak_control_sim} shows the results. Whether the tree has $k=2$
nodes per level or $k=100$, the conditions lead to a maximum FWER within
simulation error.\footnote{Whether we use  $2 \times \sqrt{0.05
(1-0.05)/10000}=0.004$ or $2 \times \sqrt{0.5 (1-0.5)/10000}=0.01$ we interpret
these results to show nominal FWER control.} A tree with $k$ nodes per level
and a maximum of $L$ levels has $k^{\ell-1}$ nodes at level $\ell$ and
$k^{L-1}$ leaves. The first row of the table shows binary trees ($k=2$): with 3
levels we have 4 leaves and 7 total nodes; with 19 levels we have 262,144
leaves and 524,287 total nodes. This set of simulations produces research
designs that tend to be unrealistic from the perspective of block-randomized
experiments in public policy. But these extreme examples provide what we hope
is a compelling illustration of Theorem~\ref{thm:weakctrl}.

\begin{table}[tb]
\centering
\caption{The top-down procedure controls the
  family-wise error rate in the weak sense across complete $k$-ary trees of
  widely varying size: with every null hypothesis true and $\alpha=.05$, the
  maximum FWER across 10,000 simulations stays within simulation error (about
  .01) of the nominal level. Each row summarizes trees with $k$ nodes per
  level and between Min and Max levels; the Nodes and Leaves columns give the
  resulting range of total and terminal nodes. `Max FWER' reports the largest
  proportion of simulations with any false rejection among the trees in the
  row; `Avg. Max Tests' reports the largest average number of tests the
  procedure ran before stopping.} 
\label{tab:weak_control_sim}
\begin{tabular}{rrrrrrrrr}
  \toprule
  &\multicolumn{2}{c}{Levels} &Max&Avg. Max& \multicolumn{2}{c}{Nodes} & \multicolumn{2}{c}{Leaves}\\ $k$ & Min & Max & FWER & Tests & Min & Max & Min & Max \\ \midrule
   2 &    3 &   19 & 0.052 & 1.008 &    7 & 524287 &    4 & 262144 \\ 
     4 &    3 &    9 & 0.055 & 1.015 &   21 & 87381 &   16 & 65536 \\ 
     6 &    3 &    7 & 0.054 & 1.026 &   43 & 55987 &   36 & 46656 \\ 
     8 &    3 &    7 & 0.050 & 1.032 &   73 & 299593 &   64 & 262144 \\ 
    10 &    3 &    5 & 0.049 & 1.044 &  111 & 11111 &  100 & 10000 \\ 
    20 &    3 &    5 & 0.050 & 1.202 &  421 & 168421 &  400 & 160000 \\ 
    50 &    3 &    3 & 0.048 & 1.402 & 2551 & 2551 & 2500 & 2500 \\ 
   100 &    3 &    3 & 0.050 & 2.498 & 10101 & 10101 & 10000 & 10000 \\ 
   \bottomrule
\end{tabular}
\end{table}

Table~\ref{tab:weak_control_sim} also shows that the two conditions lead the
algorithm to nearly always stop after testing the single root node: the average
number of tests is close to 1, even when the tree has hundreds of thousands of
nodes. In contrast, the bottom-up procedure runs a test in every block and then
adjusts the resulting $p$-values for multiplicity \autocite{hommel1988stagewise}.
In the largest tree simulated here, that is 262,144 randomization tests.

\subsubsection{Why might weak control suffice?}

The simulations and proof above assume that no units in any blocks have
treatment effects --- that all null hypotheses are true. This is what it means
for a procedure to control the FWER ``in the weak sense''
\autocite{hochberg1987multiple}. The next section develops conditions for strong
FWER control, which hold regardless of which hypotheses are true. Yet
Conditions~\ref{cond:stopping} and~\ref{cond:valid} alone suffice for some questions.

Weak FWER control is not a lesser guarantee. It is what the
Benjamini--Hochberg procedure \autocite{benjamini1995controlling} --- the
standard tool for false discovery rate (FDR) control --- provides when every
null hypothesis is true. The FDR is the expected proportion of false rejections
among all rejections, and Benjamini--Hochberg controls it at level $\alpha$
regardless of which hypotheses are true. Yet when all null hypotheses are true, every
rejection is a false rejection, so FDR equals FWER; the Benjamini--Hochberg procedure therefore
controls the FWER in the weak sense but not in the strong sense. Researchers
accept this because FDR control suits exploratory settings where findings will
be validated in follow-up studies. The tree-structured approach does not
control the FDR (developing such guarantees is beyond the scope of this paper),
but it shares with FDR procedures the property that the FWER is controlled at
$\alpha$ when all null hypotheses are true.

The choice among error-rate guarantees depends on the research question. The
question \emph{does this program work at all?} asks only whether any block has
an effect, not which ones. The error that matters is a false positive ---
claiming an effect when none exists --- and weak FWER control bounds its
probability. The question \emph{which individual outcomes are affected?} ---
the one that motivates this paper --- calls for strong FWER control. FDR
procedures fall between, controlling the expected proportion of false
discoveries, not the chance of even one.

For the first question, Conditions~\ref{cond:stopping} and~\ref{cond:valid}
are all the tree-structured procedure needs --- Theorem~\ref{thm:weakctrl}
shows they suffice for weak FWER control. Its overall test acts as a
\emph{gate}: by the Stopping Rule (Condition~\ref{cond:stopping}), no further
hypothesis is tested unless the overall null is rejected first. Because this
gate absorbs the entire multiplicity cost, the tests below it need no
adjustment of their own --- each uses the full nominal $\alpha$. That is power
both strong FWER and FDR procedures give up for their broader guarantees.

Researchers who are using the tree-structured approach to screen blocks for
closer examination can work with Conditions~\ref{cond:stopping}
and~\ref{cond:valid} alone. But in many applied settings, researchers want to
make claims about specific blocks with the same error-rate guarantees they would
demand of any confirmatory analysis. The next section shows how: analyzing the
power decay inherent in data splitting --- and, when needed, applying adaptive
$\alpha$-allocation --- yields strong FWER control while preserving the power
advantage of top-down testing.

\section{Strong control of the FWER}\label{sec:strong}
Conditions~\ref{cond:stopping} and~\ref{cond:valid} control false positive
errors when all null hypotheses are true. Control of the FWER when at least one block contains a non-zero treatment
effect is called strong control. It requires additional structure, and
identifying that structure is a key contribution of this paper. It rests on a fact that runs against the usual intuition about
multiple testing. Ordinarily we treat statistical power as the property we want
more of, and the risk of a false rejection comes from the size of each test, not
from its power. In a sequential tree the two are linked: the quantity that
governs the FWER is statistical power itself. The procedure can err only at a null node it actually
reaches, and it reaches that node only by first rejecting every ancestor above
it. The chance of reaching and then rejecting a null node is therefore a product
of rejection probabilities along the path from the root. High power makes that
product large, so it is high power that lets the procedure reach the nodes
where it can falsely reject. This is the relationship we exploit. Power decays as the data splits, which
limits how deep the procedure reaches. Where power stays high we adjust
$\alpha$ to compensate.

This section builds one quantity, the \emph{error load}, that an analyst
computes before testing to make a single decision: keep the nominal $\alpha$
at every node, or tighten it. In all 25 MDRC education trials the error load
is far below 1, so no adjustment is needed and the unadjusted top-down
procedure is the one we recommend there (Section~\ref{sec:mdrc}).\corrflag{1} When the
error load is high the same diagnostic calls for an adjustment, and the rest
of the section derives it --- the adaptive $\alpha$-schedule, budget weights,
branch pruning, and switching back to nominal $\alpha$ --- a case the trials
never reach but the Supplement's Job Corps design does
(Section~\suppref{sec:njcs}{E}).

Return to the $k=3$, $L=3$ tree of Figure~\ref{fig:k3L3_all_null_tree}, with 13
nodes and 9 leaves. Figure~\ref{fig:k3L3_tree} shows the same tree, now with a
treatment effect in one block --- block 5 (boxed). Because block 5 is non-null,
so are its ancestors, node 2 and the root: 3 of the 13 nodes are non-null and
10 are null. We use complete $k$-ary trees like this one for simplicity; the
theory generalizes to irregular trees, and we apply those methods to the MDRC field experiments in \S~\ref{sec:mdrc}.

\begin{figure}[tb]
\centering
  \begin{tikzcd}[cramped, sep=tiny, row sep=tiny]
    &        &  & & \cdbox{1}   \arrow{dlll} \arrow{d} \arrow{drrr}& & \\
    & \cdbox{2} \arrow{dl} \arrow{d} \arrow{dr}  & & & 3 \arrow{dl} \arrow{d} \arrow{dr}& & & 4 \arrow{dl} \arrow{d} \arrow{dr}\\
    \cdbox{5}  & 6 & 7 & 8 & 9 & 10 & 11 & 12 & 13
  \end{tikzcd}
  \caption{A $k$-ary tree with $k=3$ and $L=3$. Boxes show non-null nodes: block 5 has a treatment effect, so $H_5$ is false, and the hypotheses at its ancestors --- $H_2$ and $H_1$ --- are false as well. The hypotheses at the other ten nodes are true.}\label{fig:k3L3_tree}
\end{figure}

Notice that a false positive in this tree can occur only at a \emph{null node}
--- a node where the hypothesis of no effect is true, like node 3 or 13 --- and
only when that node is both tested and rejected.
Figure~\ref{fig:k3L3_tree} contains a tree where we \emph{want} to
reject the root node. In fact, we \emph{should} reject that node with
probability much greater than $\alpha$ if we have enough data and the causal
effect measured at that node is large.

This line of reasoning reveals that the probability of false positive errors
accumulating in tests through the tree depends on the power of the test at the
root node (we write the statistical power of a test at node $i$ as $\pi_i$).
Consider two extreme cases: (1) the overall test has very low power and (2) it
has very high power. If the overall test has very low power, say $\pi_1 \leq
\alpha$, then this test will reject very rarely, and, in turn, the tree will
show very few false positive errors. In fact, if $\pi_i \leq \alpha$, then the two conditions used above control
the FWER in the strong sense. Not rejecting at the root controls the FWER for
either of two reasons. In the first, all blocks have no effects and the test
has false positive error rate of less than $\alpha$. In the second, at least
one block has a non-zero effect but the test lacks the power to detect it
($\pi_1 < \alpha$).

What about when $\pi_1 > \alpha$? In this case, the multiple testing of null
nodes are no longer protected by non-rejection and application of the
stopping rule at the root. For example, in Figure~\ref{fig:k3L3_tree} rejection
of the root will lead to 3 tests, two of them of null nodes. We will have two
chances to falsely reject the null of no effects and, if the tree stopped with
nodes 2,3 and 4, the FWER would be greater than $\alpha$ without further
assumptions.

When the tree continues past that first level, the number of null nodes that
could be tested grows to ten --- every node in this thirteen-node tree but the
non-null root, node 2, and block 5 --- and with it the chances to err. Not
all ten need separate accounting. The Stopping Rule forces the procedure to
reject a node's ancestors before it can reject the node itself, so the first
false rejection on any path must occur at a true null directly below the
non-null part of the tree. Call such a node a \emph{boundary null}: a true null
whose parent, when it has one, is non-null. Under the Stopping Rule, a true null
is rejected somewhere in the tree if and only if some boundary null is rejected
(Lemma~\suppref{lem:boundary}{1} in the Supplement). Rejecting a boundary null
is itself a false rejection, and the first one on its path. The FWER is the probability that \emph{at least one} false rejection occurs,
not a count of how many. So once a boundary null is falsely rejected, any
further rejections the Stopping Rule carries into its descendants cannot raise
the FWER. To control false positives
when at least one block contains a non-zero effect, we therefore need to bound
only the rejection of boundary nulls --- four of the ten here. In Figure~\ref{fig:k3L3_tree}, the boundary nulls are four nodes. Nodes~3
and~4 are children of the non-null root. Nodes~6 and~7 are children of the
non-null node~2. The other nulls --- nodes~8 through~13 --- are descendants of
nulls~3 and~4: the procedure can reach and reject them only after falsely rejecting one of those boundary nulls. The next subsection turns this fact --- that only boundary nulls need bounding
--- into the error load.

\subsection{Strong control using error load and adaptive adjustment}

To bound the probability that a null node is falsely rejected, separate it into
two parts: the procedure must \emph{reach} that node by rejecting every ancestor
above it, and must then \emph{reject} the node itself. These two chances ---
reaching a node, and rejecting it once reached --- determine how
much that node contributes to the FWER. A null buried deep behind ancestors that
rarely reject contributes little, because the procedure seldom arrives there at
all. The error load defined in this subsection is the paper's central diagnostic:
the adjustment derived below matters only when the error load exceeds 1.

Let $\mathcal{H}_0$ denote the set of nodes $i$ for which the null $H_i$ is
true. Under the Stopping Rule, the test at node $j$ runs only after every
ancestor of $j$ has rejected, so we define its rejection probability given that
event: the \emph{conditional rejection probability} $\theta_j$ is $\alpha_j$
(the test size) when $H_j$ is true and $\pi_j$ (the power) when $H_j$ is false.
The conditioning matters because tests along a path share data. The
probability of rejecting every node on a path is exactly the product of these
conditional probabilities, by the chain rule. The same product of
unconditional sizes and powers would require an independence assumption that
data splitting does not deliver. Write $\anc(i)$ for the set of proper ancestors of node $i$
(excluding $i$ itself). And let $R_i = \mathbb{I}\{\text{reject } H_i\}$: the
rejection indicator at node $i$. Under the Stopping Rule, $R_i = 1$ requires
that $R_j = 1$ for every $j \in \anc(i)$: the procedure can reject node $i$ only
after it has already rejected every ancestor of $i$. Then the FWER satisfies

\begin{equation}
  \label{eq:fwer-bound1}
  \mathrm{FWER}
  = \Pr \left(\bigcup_{i \in \mathcal{H}_0} \{R_i = 1\}\right)
  \le
  \sum_{i \in \mathcal{H}_0} \alpha_i \prod_{j \in \mathrm{anc}(i)} \theta_j,
\end{equation}

where the product over an empty set of ancestors (for the root node) is taken to
be $1$; we prove this bound as Proposition~\suppref{prop:fwer}{1} in the Supplement.

Corollary~\suppref{cor:fwer_control_low_power}{2} in the Supplement shows how this expression recovers
FWER control when the root has very low power ($\pi_1 \leq \alpha$): the
root itself rarely rejects, so no descendant is tested. The expression also
reveals the problem when the root has high power.
When $\pi_1 = 1$, rejection at the root is certain, and the unadjusted FWER at
level 2 is approximately $k \cdot \alpha_{\ell=2}$, the sum of the false
positive rates across all $k$ children in that first branch of the tree from
the root, holding the test size at the same level for all $k$ children of the
root for notational simplicity. With 10 children each tested at $\alpha =
0.05$, the FWER reaches $0.5$.

To build intuition, imagine that the root has perfect power ($\pi_1 = 1$) and
that the blocks are split into $k$ equal-sized groups at each level. Because
each group contains $1/k$ of the data, the test at each child node has lower
power than its parent: if we approximate power using the Normal distribution, then $\theta_\ell
\approx \Phi\!\left(\delta \sqrt{N/k^{\ell-1}} - z_{1-\alpha/2}\right)$.\corrflag{1} Here
$\Phi$ is the standard normal CDF, $\delta$ a standardized effect size (here
$\delta = d/2$ for Cohen's $d$, the half reflecting the even split of each
node between treatment and control), $z_{1-\alpha/2}$ the two-sided critical
value, and a node at level $\ell$ holds $N/k^{\ell-1}$ of the $N$ total units
under equal splitting. Dividing the sample by $k$ at each level then reduces
power geometrically. At level
$\ell$, the tree has $k^{\ell-1}$ nodes, and each null node is reached only if
every ancestor on its path is rejected. The product of these rejection
probabilities along the path from the root determines how likely the procedure
is to reach --- and therefore to falsely reject --- each null node. For example, $H_6$ in Figure~\ref{fig:k3L3_tree} adds to the FWER only if the
procedure reaches block 6, which requires rejecting both of its ancestors, nodes
2 and 1, first.

A given \emph{configuration} is the pattern of which blocks carry effects and
which are null. Define the \emph{realized error load} as the path power summed
over the boundary nulls that configuration exposes, $\sum_i \prod_{j \in
\mathrm{anc}(i)} \theta_j$.
Because a true null is rejected only if a boundary null is rejected
(Lemma~\suppref{lem:boundary}{1} in the Supplement), the sum
in~\eqref{eq:fwer-bound1} can be restricted to the boundary nulls: summing
$\alpha_i \prod_{j \in \mathrm{anc}(i)} \theta_j$ over boundary nulls alone
still bounds the FWER. With every node tested at the nominal $\alpha$ --- the
unadjusted procedure --- this restricted sum equals $\alpha$ times the realized
error load, so that product bounds the FWER.
Group that bound by level. At most $k^{\ell-1}$ nodes sit at level $\ell$,
each reached with path power at most $\prod_{j=1}^{\ell-1} \theta_j$. Write
$G_\ell = k^{\ell-1}\prod_{j=1}^{\ell-1}\theta_j$ for this worst-case
per-level count (Section~\suppref{sec:strong_local_adj}{B} of the Supplement).
The realized error load is then at most $\sum_\ell G_\ell$, so

\begin{equation}
  \label{eq:fwer-natural-gating}
  \mathrm{FWER} \le \alpha \underbrace{\sum_{\ell=2}^{L} G_\ell}_{\text{error load}} .
\end{equation}

\noindent We call this worst-case sum $\sum_\ell G_\ell$ the \emph{error load}. It
counts every node as a potential exposed null, so it is at least the realized
error load of every configuration; we show below how far the two can diverge
when effects concentrate in a few branches.

When the realized error load is at most~1 the FWER is at most $\alpha$, so
Conditions~\ref{cond:stopping} and~\ref{cond:valid} --- the two that give weak
FWER control --- give strong control as well, with no per-level adjustment. When the error load is at most~1 the unadjusted procedure controls the FWER on
its own. We call this case the \emph{natural gating} regime, the first of the
four regimes of Table~\ref{tab:four_regimes}. The name reflects the mechanism: power
decays as data splitting reduces sample sizes at each level, and this decay
limits how many null nodes the procedure can reach. Because $\sum_\ell G_\ell$ depends only on the design, not on the
configuration (the unknown pattern of which blocks carry effects), the single
pre-data check $\sum_\ell G_\ell \le 1$ guarantees natural gating wherever the
effects fall (Supplement Proposition~\suppref{prop:strong_fwer_adj}{2}). In
the multi-site education policy trials that motivate this paper, the effect
sizes and tree sizes are moderate, with $k \approx 3$ and $L \approx 3$. Even
this conservative sum stays below~1 there, so no adjustment is needed.\corrflag{1}

The rest of this section treats the high-error-load case. It does not arise in
the 25 education trials of Section~\ref{sec:mdrc}, where the error load is far
below 1.\corrflag{1} It is the regime of the calibrated Job Corps design in
Section~\suppref{sec:njcs}{E} of the Supplement.

The realized error load exceeds~1 with wide, shallow trees or high root power.
When it does, the per-node bound~\eqref{eq:fwer-bound1} no longer holds the
FWER below $\alpha$, so the unadjusted procedure is no longer guaranteed to
control it. Restoring control then requires tightening $\alpha$ at each level, an adaptive
adjustment we derive in Section~\suppref{sec:adaptive_alpha_adj}{B.5} of the Supplement.
When the test statistic has an asymptotically Normal distribution, we can
compute approximate power at each depth \emph{before} testing begins.\footnote{
Section~\suppref{sec:test_statistic}{C} of the Supplement introduces such a statistic,
following \textcite{hothornetal:coin:2006}.} Each $\hat\theta_j$ is the power to detect an effect size we fix at the design
stage, computed at the sample size available at depth $j$. That effect size is
a Cohen's $d$ we anticipate rather than estimate from the realized outcomes. With these power estimates in hand we adjust $\alpha$ at each level
accordingly
(Supplement Remark~\suppref{thm:root_adjust}{10} and Theorem~\suppref{thm:fwer_control_adaptive}{3}) as follows:

\begin{equation}
\alpha_\ell^{\text{adj}} = \min\left\{\alpha, \frac{\alpha}{k^{\ell-1} \cdot \prod_{j=1}^{\ell-1} \hat{\theta}_j}\right\}
  \label{eq:optimal_alpha_d}
\end{equation}

Equation~\eqref{eq:optimal_alpha_d} controls the strong FWER on a regular tree
using an error load computed once, before testing: $G_\ell = k^{\ell-1}
\prod_{j=1}^{\ell-1} \hat\theta_j$, calculated from the estimated power at each depth, $\ell$.
The stopping rule still halts testing below any node that fails to reject, but
we do not recompute $G_\ell$ to remove those untested subtrees. Two
generalizations extend this adaptive testing regime. The first handles irregular trees, where the number of children varies across
nodes. For an irregular tree the depth-$\ell$ error load becomes a sum over
the actual nodes, $G_\ell = \sum_{i \in \text{level } \ell} \prod_{j \in
\anc(i)} \hat\theta_j$. The total $\alpha$ is then split across depths by
weights $\{w_\ell\}_{\ell=2}^{L}$ summing to at most one. We call this the \emph{error budget}. The level $\alpha$ is treated as a fixed
resource, and each depth is assigned a share $w_\ell$ of it, so the depths'
false-rejection rates sum to at most $\alpha$. The second adapts the
error load as testing proceeds. Whenever a
branch fails to reject, the stopping rule drops its whole subtree, so the null
nodes there can never be reached and need not be counted. After each depth,
then, the error load is recomputed on only the surviving subtree, replacing the
static $G_\ell$ with the smaller surviving error load $D_\ell$ from the branches that
did reject. We call this procedure \emph{branch pruning}.
Table~\ref{tab:four_regimes} summarizes the four resulting \emph{regimes}:
situations distinguished by the tree's shape, its error load, and whether
branches are pruned. For each regime it also gives the rule that controls the
FWER and the theorem proving it. All 25 education trials of Section~\ref{sec:mdrc} fall in Regime~1, natural
gating, where no adjustment enters.\corrflag{1} Regimes~2--4 cover the high-error-load
case.

\begin{table}[tb]
  \centering
  \caption{Four regimes for FWER control in top-down testing. $G_\ell$
  is the \emph{static} error load at depth~$\ell$ --- the expected count
  of null nodes the procedure reaches, computed once from estimated power
  before testing and not revised as branches drop out. $D_\ell$ is the
  \emph{surviving} error load computed after pruning on the
  subtree that actually rejected, so $D_\ell \le G_\ell$. The weights
  $w_\ell$ are budget shares with $\sum_\ell w_\ell \le 1$, fixed before
  each depth is tested and possibly using the testing
  history through depth $\ell - 1$.
  Regime~1 needs no $\alpha$-adjustment.}\label{tab:four_regimes}
  \begin{tabular}{llll}
    \toprule
    Regime & Tree shape & Adjusted $\alpha$ at depth $\ell$ & Theorem \\
    \midrule
    1. Natural gating ($\sum G_\ell \le 1$) & any & $\alpha$ &
      Prop.~\suppref{prop:strong_fwer_adj}{2} \\
    2. Regular, static & regular & $\alpha / G_\ell$ &
      Thm.~\suppref{thm:fwer_control_adaptive}{3} \\
    3. Irregular, static & irregular & $w_\ell \alpha / G_\ell$ &
      Thm.~\suppref{thm:fwer_budget}{4} \\
    4. Pruned & any & $w_\ell \alpha / D_\ell$ &
      Thm.~\suppref{thm:fwer_budget_pruning}{5}\corrflag{2} \\
    \bottomrule
  \end{tabular}
\end{table}

The third column of Table~\ref{tab:four_regimes} --- the adjusted $\alpha$ at depth~$\ell$, written $\alpha_\ell$ --- is the
significance level at which a node's test rejects. In every adjusted regime $\alpha_\ell$ falls inversely with the error load. The
``Regular, static'' row sets it to $\alpha/G_\ell$. The ``Irregular, static''
row multiplies that value by the depth's budget share $w_\ell$, giving
$w_\ell\alpha/G_\ell$. The ``Pruned'' row swaps the static error load $G_\ell$
for the surviving error load $D_\ell$, giving $w_\ell\alpha/D_\ell$. Only the
``Natural gating'' row keeps the nominal $\alpha$. A larger error load means a smaller $\alpha_\ell$, so a node needs a smaller
$p$-value to reject. In the ``Regular, static'' row, $G_\ell = 5$ drops $\alpha_\ell$ from $\alpha = 0.05$ to $0.01$. The error load is largest near the root and shrinks as power decays with each
split. The adjustment is therefore stringent near the root and relaxes at deeper
levels, where power decay has already gated the procedure. This resembles
$\alpha$-investing \autocite{foster2008alpha}, but here the allocation across
depths is set by the relationship between power and tree structure rather than
by observed $p$-values.

This downward flow of significance places our procedure in a tradition of
structured multiple testing developed largely for clinical trials. A parent
gates its children, and error budget passes to where rejections remain
possible.
Gatekeeping procedures make the test of a later family of hypotheses
contingent on rejections in an earlier ``gatekeeper'' family: all of its
hypotheses for a serial gatekeeper, at least one for a parallel gatekeeper
\autocite{dmitrienko2007gatekeeping, dmitrienko2009multiple}. The recycling
framework of \textcite{burman2009recycling} passes the $\alpha$ freed by a
rejected hypothesis to others by a fixed transfer rule. The graphical approach
of \textcite{bretz2009graphical} recasts such weighted-Bonferroni procedures
as directed graphs whose node and transition weights are updated by a simple
rule upon each rejection, passing freed weight to surviving nodes. Our
adaptive $\alpha$-schedule shares the logic of pre-specifying how significance
is allocated, but differs in what the data then do with it. In these procedures the allocation graph is fixed in advance, and the $\alpha$
freed by a rejection is recycled onto surviving hypotheses, relaxing their $\alpha$. Our per-depth budget is fixed in advance by design-stage power
estimates and the tree's shape, and a rejection never relaxes a sibling's $\alpha$. Branch pruning is just
stopping: when a node cannot be rejected, the subtree descending from it goes
untested, and the budget those tests would have consumed is returned in the
predictable, pre-specified way the guarantees require. The two traditions are
already connected: the sequential-rejection principle we build on
\autocite{goeman2010sequential} encompasses these gatekeeping and graph-based
procedures, and \textcite{goeman2012inheritance} casts hierarchical testing as a
special case of the graphical approach. Our contribution specializes that line
to randomization-based tests with an allocation fixed by the design rather than
read off the data.

\paragraph{Budget weights, $w_\ell$} In a regular, unpruned tree, setting each depth's $\alpha_\ell$ to
$\alpha/G_\ell$ needs no weights: across depths, each node's
false-rejection contributions already add up to at most $\alpha$ on their own.
They add up this way because of a cancellation in regular trees. Every node at
a depth shares the same path power, which cancels against $G_\ell$. The
procedure can reach only a \emph{boundary} null --- a null node whose parent
is non-null --- and each such node then contributes at most
$\alpha/k^{\ell-1}$ to the FWER. Write $m_\ell$ for
the number of non-null nodes at depth $\ell$, $e_\ell$ for the number of
boundary nulls there, and $L$ for the number of depths. Each non-null parent
has $k$ children, every one either non-null or a boundary null, so $m_\ell +
e_\ell = k\,m_{\ell-1}$ at each depth. Summing the boundary-null contributions
and substituting this identity makes consecutive terms cancel:
\[
  \sum_{\ell=2}^{L} \frac{e_\ell}{k^{\ell-1}}
  = \sum_{\ell=2}^{L}\left(\frac{m_{\ell-1}}{k^{\ell-2}}
      - \frac{m_\ell}{k^{\ell-1}}\right)
  = m_1 - \frac{m_L}{k^{L-1}} \le 1 ,
\]
since the root is the only depth-1 node, so $m_1 \le 1$. The per-depth
false-rejection rates therefore sum to at most $\alpha$ on their own
(Supplement Theorem~\suppref{thm:fwer_control_adaptive}{3}).
That telescoping fails in two cases. The first is an irregular tree, where
boundary nulls can concentrate in a high-power branch. The second is pruning,
which makes the denominators data-dependent. The fix is to give each depth an explicit share $w_\ell$ of the
error budget, with $\sum_\ell w_\ell \le 1$, so the per-depth false-rejection
rates again sum to at most $\alpha$. These shares need no tuning: equal weights
$w_\ell = 1/(L-1)$ control the FWER, as does any predictable allocation summing
to one (Supplement Theorems~\suppref{thm:fwer_budget}{4}
and~\suppref{thm:fwer_budget_pruning}{5}; we define predictability just below).
Unequal weights are optional. They gain power by spending more of the budget
on the shallow, high-power depths, or on the branches that survive pruning.
Remark~\suppref{rem:weight_choices}{13} of the Supplement catalogs the
choices.

Which $\alpha$-schedule to use follows three checks. First, check whether the condition $\sum_\ell G_\ell \le 1$ holds. This check
uses no information about the configuration, the pattern of which blocks carry
effects. If it holds, nominal $\alpha$ controls the FWER wherever the effects
fall, and no weights enter. This is the case for all 25 education trials of
Section~\ref{sec:mdrc}.\corrflag{1}
Second, if it fails, an assumption about which level of the tree the effects
cluster at can still guarantee the unweighted $\alpha$-schedule. Under that
assumption the FWER is at most $\alpha \sum_\ell S_\ell/G_\ell$, where
$S_\ell$ sums the path power of the exposed (boundary) nulls at depth $\ell$.
Whenever that sum stays below one, testing each node at the unweighted $\alpha/G_\ell$ holds the FWER at or below
$\alpha$, with more power than any
choice of weights. Third, only when neither check passes, or when we will not
commit to the level-clustering assumption, are budget weights needed, for a
guarantee that holds whatever the configuration. Section~\ref{sec:dpp_example}
demonstrates all three checks on the simulated Detroit Promise design.

\begin{correctionbox}
The FWER-control guarantees claimed in the next two paragraphs ---
Theorem~\suppref{thm:fwer_budget_pruning}{5} for branch pruning and
Corollary~\suppref{cor:switching}{4} for switching back to nominal $\alpha$
--- are false as stated and are withdrawn (Correction notice, item~2). The
descriptions remain so the reader can follow the simulations that use them;
the control claims should not be relied on.
\end{correctionbox}

\paragraph{Branch pruning} Branch pruning is Regime~4 of
Table~\ref{tab:four_regimes}. When a branch (a child node
and the subtree beneath it) fails to reject at depth $\ell$, the Stopping Rule
removes that whole subtree from further testing, so the null nodes in it can
never be reached and never contribute to the FWER. Recomputing the error load
on what remains --- the surviving error load $D_\ell \le G_\ell$ --- lets the deeper tests run at a larger $\alpha_\ell$. Theorem~\suppref{thm:fwer_budget_pruning}{5} in
the Supplement shows that this recomputed adjustment still controls the FWER at
level~$\alpha$ for any tree shape, provided the budget weights $w_\ell$ are
\emph{predictable}: each $w_\ell$ is fixed before the depth-$\ell$ tests
are seen, using only results from shallower depths. Predictability is what makes the guarantee hold. A weight chosen after seeing
a depth's $p$-values could be tuned to manufacture rejections, and the bound
would no longer apply. The gain over not pruning is largest in wide trees where most of the root's
branches contain no effects at all. Once those branches are pruned, the
surviving error load $D_\ell$ no longer counts their nodes, nodes the
procedure will never test. The deeper $\alpha_\ell$ then relax to $w_\ell\alpha/D_\ell$.

\paragraph{Switching to nominal $\alpha$} A practical refinement
applies once pruning has narrowed the surviving subtree enough. If at
depth $s$ the remaining surviving error load
$\sum_{\ell=s}^{L} D_\ell$ fits within the remaining budget
$1 - \sum_{\ell=2}^{s-1} w_\ell$, the procedure can test at the
nominal level~$\alpha$ for every depth $\ell \ge s$
(Corollary~\suppref{cor:switching}{4} in the Supplement). The intuition is
that the surviving subtree, by virtue of its small size, is back in
the natural-gating regime: explicit adjustment was protecting against
a multiplicity problem that the surviving tree no longer poses. In the simulated
Detroit Promise design (\S~\ref{sec:dpp_example} below), once Henry Ford
Community College (HFCC) is the only college to survive its college-level
test, the surviving tree contains only HFCC's cohorts and blocks; its surviving
error load is small enough to fit within the remaining budget, so
Corollary~\suppref{cor:switching}{4} lets the cohort and block tests proceed at
full~$\alpha$ while still controlling the FWER.

\paragraph{Sensitivity to the assumptions} The four regimes for strong control
of the FWER rest on two requirements beyond the conditions of
\S~\ref{sec:topdownbotup}, and neither holds exactly. The first is
\emph{conditional validity} at boundary nulls: a true null's false-positive rate
stays at or below $\alpha$ even after we condition on its non-null ancestors
having been rejected. The second is \emph{conservative path-power estimation}:
the estimates $\hat\theta_j$ that build $G_\ell$ do not understate the true power
along non-null paths. Remark~\suppref{rem:sensitivity}{8} in the Supplement collects departures from
both into a single worst-case inflation factor $\kappa_{\max}$. This factor
equals~$1$ when both requirements hold and grows the further either departs.
The Remark then shows $\mathrm{FWER} \le \kappa_{\max}\sum_\ell \alpha_\ell
G_\ell$. The no-adjustment threshold moves from $\sum_\ell G_\ell \le 1$ to
$\sum_\ell G_\ell \le 1/\kappa_{\max}$, equivalently testing each node at
$\alpha/\kappa_{\max}$.

The distortion is weakest where the error load is largest. Write $\rho$ for
the ratio of a true null's conditional false-positive rate to its nominal
rate, the inflation from conditioning. This $\rho$ falls toward~$1$ as
ancestor power rises, and high ancestor power is what makes the error load
large. The per-test inflation is therefore largest only along the low-power
paths the procedure rarely reaches.

\begin{correctionbox}
The paragraph below runs this sensitivity check on the incorrect error loads
of Correction notice, item~1. Corrected loads exceed 1 in all 25 MDRC studies
at the planning effect size $d = 0.20$: the largest load is not 0.051, so
there is no 19.6-fold margin and the bound of 0.011 below does not hold.
\end{correctionbox}

For the application, this margin is the point of the sensitivity check. Across the 25 MDRC trials the largest error load is below 0.051, so far below
one that nominal-$\alpha$ testing holds $\mathrm{FWER}\le\alpha$ for any
combined distortion up to $\kappa_{\max} = 1/0.051 \approx 19.6$. Suppose both departures struck at once. Conditioning doubled a true null's
false-positive rate to $0.10$, so $\rho = 2$. A two-step non-null path had
true power $0.9$ per step but was estimated at $0.6$, making the path-power
ratio $\lambda = 0.9^2/0.6^2 = 2.25$. Even then
$\kappa_{\max} = \rho\lambda = 4.5$, far under $19.6$, and $\mathrm{FWER}\le
\kappa_{\max}\,\alpha\sum_\ell G_\ell = 4.5\times0.05\times0.051 \approx 0.011$,
far below $\alpha$.

Could $\rho$ really reach $2$? The administrative splits make it unlikely in
this case. A node's test pools the students of every block beneath it, while a
null block's test uses only its own students. The two tests therefore
correlate only through the shared units, a correlation near
$\sqrt{n_{\text{block}}/n_{\text{parent}}}$. A block is a small share of that pool,
and a boundary null's parent is non-null by construction, rejecting on the
strength of its real-effect children rather than the null block's noise. Because
the blocks are disjoint, the null block shares no units with those
effect-carrying siblings; its only tie to the parent's test is its own small
share of the pooled units. Conditioning on the parent's rejection therefore shifts the null block's test
statistic by at most that small correlation times the parent's exceedance, so
$\rho$ stays near~$1$. Supplement Remark~\suppref{rem:sensitivity}{8}
quantifies $\rho$ under a bivariate-normal approximation. Reaching $\rho = 2$ would take a null block large enough to drive the
test above it and to swing it toward the block's own treated-control gap, which
a disjoint block holding a small share of the pool cannot.

The remaining assumption is that the power estimates overstate rather than
understate the truth, and a departure from it is checkable. Each
$\hat\theta_j$ predicts an ancestor's rejection rate. A design that
underestimated power from $0.9$ to $0.6$ would therefore reject far more often
than planned.

\subsection{An Example: A Simulated Version of The Detroit Promise Program}\label{sec:dpp_example}

The Detroit Promise Program (DPP) \autocite{ratledge2019path} involved enhanced
advising, financial support, and encouragement of full-time enrollment to
community college students  in an effort to increase graduation rates. The
trial randomized this bundle of treatments within five community colleges in
the Detroit area: Henry Ford Community College, Macomb Community College,
Oakland Community College, Schoolcraft College, and Wayne County Community
College District. Each college fielded the trial for 3 academic year cohorts,
and within the cohorts students were further subdivided into 1 to 4
experimental blocks within which the intervention was randomly assigned. To illustrate tree-gated testing, we generate a design that preserved the DPP
structure: 5 college nodes under the overall node, with varying numbers of
cohorts and sub-cohort blocks, for 44 blocks total. The design departs from
the real trial in ways that facilitated simulation. We specified 50 students
per block, and we drew potential outcomes under control from $N(10,9)$ to
roughly match the observed distribution of credits taken, but with more
variation. We concentrated all non-zero effects in a single college:
9 blocks within HFCC have treatment effects while the 35 blocks across the
other four colleges are pure null. This concentration mimics a pattern common
in multi-site trials and visible in the MDRC data (Section~\ref{sec:mdrc}):
effects that are localized within particular sites rather than uniformly
spread. For the illustration we use a large effect (Cohen's $d = 0.80$) within each of
the 9 blocks with non-null effects, so that the root-level test rejects and
the tree structure is visible. Effects this large are unusual in practice; the
MDRC trials show far weaker per-block signals. The simulation study below
therefore re-runs the design at the smaller effect sizes those trials suggest.

Figure~\ref{fig:dpp_sim} shows one run of the top-down procedure with
\emph{no} $\alpha$-adjustment: every reached node is tested at the nominal
$\alpha = 0.05$, and the stopping rule halts testing below any node that fails
to reject. We use the unadjusted procedure on purpose. At the large effect $d = 0.80$ the
error load is far above~1, so a single picture can show both the correct
detections and the false-rejection cost that adjustment is meant to remove. Nodes whose subtree contains a true effect are filled blue with rounded
corners; null nodes are white with square corners. Any ancestor of a non-zero
block is itself non-zero --- a rejection at the college or cohort level is a
true discovery, not an error. The algorithm tested 18 of the 44 blocks rather
than all 44, pruning branches rooted at colleges where the college-level test
did not reject. It correctly identified HFCC as the source
and descended into its cohorts and blocks, detecting four of the nine non-nulls while falsely rejecting two of the 35 nulls in Schoolcraft College (SC).

\begin{landscape}
\begin{figure}
\centering
\includegraphics[width=.95\linewidth]{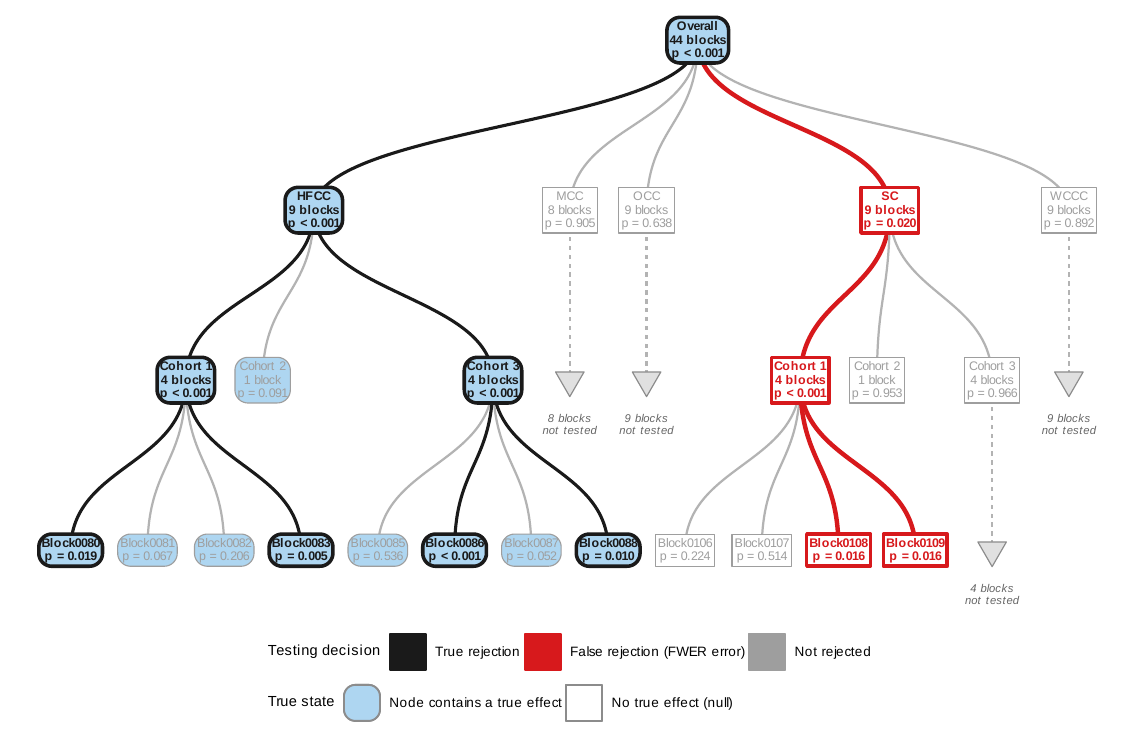}
  \caption{Top-down testing on a simulated data set built from the design of the
  Detroit Promise Program \autocite{ratledge2019path}, applying the stopping
  rule with no $\alpha$-adjustment. Nine blocks within HFCC have non-zero
  effects (Cohen's $d = 0.80$); the remaining 35 blocks are null. A thick black
  border marks a true rejection, a thick red border a false rejection (that
  would add to FWER), and a thin grey border a node that was tested but $p \ge
  \alpha$. Truth is shown by the box itself: nodes whose subtree contains a
  true effect are filled blue with rounded corners, while null nodes are white
  with square corners. Faded grey triangles represent branches the procedure
  pruned: the parent node was tested but not rejected, so it was never split
  into the individual blocks counted beneath the triangle. The red path traces the false-rejection cascade. The null Schoolcraft college
(SC) is rejected, the procedure then descends into SC Cohort~1, and the null
blocks Block0108 and Block0109 are falsely rejected. Abbreviations: Henry Ford Community College
  (HFCC), Macomb Community College (MCC), Oakland Community College (OCC),
  Schoolcraft College (SC), Wayne County Community College District
  (WCCC).}\label{fig:dpp_sim}
\end{figure}
\end{landscape}

This illustration shows that a policy maker seeking to identify which colleges
or blocks showed treatment effects would find the bottom-up approach
disappointing: the \textcite{hommel1988stagewise} adjustment across 44 blocks
detects just one block despite nine having genuine effects. The top-down approach leads attention away from three of the five colleges. It
identifies not only more individual blocks but also higher-level groupings,
such as the HFCC college itself, that can direct further exploration toward
understanding why effects concentrated there.\footnote{See
\textcite{hong2025multisite} for methods to disentangle organizational
effectiveness from ecological confounding.}\footnote{See \textcite{hong2025multisite} for methods to disentangle
organizational effectiveness from ecological confounding.}


The figure also shows the cost of using no $\alpha$-adjustment: Schoolcraft
College (SC), a purely null college, rejected at the college level ($p = 0.020$),
and the procedure descended through SC Cohort~1 ($p = 0.001$) to falsely reject
two null blocks (Block0108 and Block0109, both $p = 0.016$). The adaptive
$\alpha$-adjustment of Equation~\eqref{eq:optimal_alpha_d} stops this cascade at
its source. It tests each college not at $0.05$ but at $\alpha$ divided by the
error load $G_\ell$ at that depth. At $d = 0.80$ the root rejects with near-certainty (its estimated power is
$1.00$), so each of the five colleges is reached at path power one. The
depth-2 error load is therefore exactly $G_2 = 5$: the five colleges summed at
path power one, in the notation of Section~\ref{sec:strong}. The college-level
$\alpha_2$ therefore tightens to $\alpha/G_2 = \alpha/5 = 0.01$. SC's $p = 0.020$ clears the
nominal $0.05$ but not this $0.01$, so the adjusted procedure never rejects SC,
never descends to SC Cohort~1, and never reaches Block0108 or Block0109. The two
false rejections disappear.

This is a single simulated dataset with a large effect size chosen to make the
tree structure visible. Power and error rates describe the behavior of a
procedure across repetitions of design. We now turn to a simulation at
realistic effect sizes but using regular $k$-ary trees for simplicity to
evaluate these properties systematically.

\subsection{Simulations to show strong FWER control.}

We used the same DPP design: 9 non-null blocks in one college and 35 pure-null
blocks across the other four. We set a moderate per-block effect of Cohen's $d
= 0.30$. Because only 9 of the 44 blocks carry an effect, the overall
intention-to-treat effect divided by the outcome standard deviation is several
times smaller than that of the typical MDRC study. This is the quantity
reported for the MDRC trials in Section~\ref{sec:mdrc}. The per-block effect
is large enough for the procedure to detect structure, while the overall
signal stays realistic. We re-randomized treatment within each block 10,000 times, holding
the effect structure fixed, recording in each iteration whether a false positive
occurred at an individual block (a leaf) or at any node.
Table~\ref{tab:big_sim_xtab} summarizes the results.

\begin{table}[tb]
\centering
\small\setlength{\tabcolsep}{4pt}
\caption{Operating characteristics of five top-down testing
  procedures and the bottom-up Hommel adjustment across 10,000 simulations
  ($\alpha=.05$) using simulated Detroit Promise Program data where 9 of 44
  blocks (all in one college) have non-zero effects (Cohen's $d=0.30$). The
  full tree has 65 nodes: 1 overall test, 5 colleges, 15 cohorts, and 44
  blocks. Each row is a procedure: rows marked `, w' add equal budget weights
  (Regime~3 for the adaptive row, Regime~4 for the pruned); the unmarked rows
  use the unweighted $\alpha$-schedule, bounded at $0.8\alpha$ for this
  concentrated design (see text). The final row is the flat Hommel adjustment,
  which tests all 44 blocks at once. Columns: `Blocks tested' is the mean
  number of individual blocks tested; `Power' the mean proportion of the 9
  non-null blocks rejected; `True rej.' the mean number of non-null blocks
  rejected; `P($\geq 1$)' and `P($\geq 2$)' the proportions of simulations
  detecting at least one and at least two non-null blocks; `Mean if any' the
  mean number detected among simulations that detect any; `FWER' the
  proportion of simulations with at least one false rejection.}
\label{tab:big_sim_xtab}
\begin{tabular}{lrrrrrrr}
\toprule
 & Blocks & & True & & & Mean & \\
Procedure & tested & Power & rej. & P($\geq 1$) & P($\geq 2$) & if any & FWER \\
\midrule
\multicolumn{8}{l}{\emph{Top-down}} \\
  \quad 2 Conditions & 4.705 & 0.101 & 0.796 & 0.340 & 0.249 & 2.339 & 0.038 \\
  \quad + Adapt.~$\alpha$ & 3.835 & 0.032 & 0.218 & 0.167 & 0.044 & 1.305 & 0.003 \\
  \quad + Adapt.~$\alpha$, w & 3.380 & 0.014 & 0.086 & 0.078 & 0.008 & 1.108 & 0.001 \\
  \quad + Ad.~$\alpha$~Pr. & 4.159 & 0.039 & 0.277 & 0.205 & 0.062 & 1.353 & 0.004 \\
  \quad + Ad.~$\alpha$~Pr., w & 4.238 & 0.055 & 0.424 & 0.259 & 0.125 & 1.636 & 0.003 \\
\midrule
  Bottom-up Hommel & 44 & 0.010 & 0.088 & 0.085 & 0.003 & 1.034 & 0.024 \\
\bottomrule
\end{tabular}
\end{table}

Table~\ref{tab:big_sim_xtab} compares five top-down variants against
bottom-up Hommel. The rows name the procedures, each adding one tool
introduced above: ``2~Conditions'' applies only the stopping rule and
valid tests (Conditions~\ref{cond:stopping} and~\ref{cond:valid});
``+~Adapt.~$\alpha$'' adds the adaptive $\alpha$-adjustment of
Equation~\eqref{eq:optimal_alpha_d}; and ``+~Ad.~$\alpha$~Pr.'' combines the
adaptive $\alpha$-adjustment with branch pruning --- recomputing the
$\alpha$-schedule on the surviving subtree after each depth, so that dead
branches free up $\alpha$ for the remaining tests. The two rows marked
``,~w'' repeat the adaptive and pruned $\alpha$-schedules with equal budget weights
($w_\ell = 1/(L-1)$, Regime~3 of Table~\ref{tab:four_regimes}); this is the
adjustment that controls the FWER whatever the tree-shape, shown here to expose its power cost. The
columns report the operating characteristics: the mean number of blocks
tested, power (the mean proportion of the nine non-null blocks rejected), the
mean number of true rejections together with the distribution behind that mean
(the proportions of simulations detecting at least one and at least two
non-null blocks, and the mean detected among simulations that detect any), and
the FWER (the proportion of simulations with at least one false rejection).
The final row reports bottom-up Hommel, which tests all 44 blocks at once.

At $d = 0.30$, the error load is $\DPPHighErrorLoad$ --- well above the
threshold of~1 where Supplement Proposition~\suppref{prop:strong_fwer_adj}{2}
guarantees control.\corrflag{1} Without adjustment, whether the FWER stays below $\alpha$
depends on where we count false rejections: among the 44 individual blocks
only, or across all nodes --- colleges and cohorts as well as blocks. Among blocks the unadjusted two-conditions procedure still controls it, at
$\DPPHighThreeRulesFWER$. The error-load bound is a union bound over all nodes
at a level, and it overstates the FWER when the boundary nulls cluster at one
level, as the four null colleges do here. A \emph{block} is therefore seldom
falsely rejected. But with only the 2 conditions a null \emph{college} is rejected more freely; the procedure
rejects Schoolcraft in Figure~\ref{fig:dpp_sim}, and counting a false rejection
at any node raises the FWER to $\DPPHighTwoCondNodeFWER$
(Supplement Table~\suppref{tab:dpp_inheritance}{S7}). This inflation at the internal nodes is what the error load above~1 predicts,
and what adjustment removes. The adaptive $\alpha$-adjustment
(Equation~\ref{eq:optimal_alpha_d}) holds the node FWER to
$\DPPHighAdaptNodeFWER$ and the leaf FWER to $\DPPHighAdaptFWER$. The pruned
variant holds the leaf FWER at $\DPPHighAdaptPrFWER$.\corrflag{2} Bottom-up Hommel also controls the FWER ($\DPPHighBUHommelFWER$), but at a
steep cost. It detects effects in only $\DPPHighBUHommelLeafPowerPct$\% of the
truly non-null blocks. The unadjusted top-down procedure detects
$\DPPHighThreeRulesLeafPowerPct$\% and the adaptive variant
$\DPPHighAdaptLeafPowerPct$\%.


Although absolute power is modest at $d = 0.30$ with $N = 50$ per block, the
relative advantage is large. The unadjusted top-down approach averages
$\DPPHighThreeRulesLeafTR$ true leaf rejections per simulation,
$\DPPHighThreeRulesOverBUHommelTRRatio$ times the bottom-up Hommel rate; even
the most conservative top-down variant (adaptive $\alpha$) averages
$\DPPHighAdaptLeafTR$, $\DPPHighAdaptOverBUHommelTRRatio$ times the bottom-up
rate. Branch pruning improves further.\corrflag{2} The pruned adaptive variant averages
$\DPPHighAdaptPrLeafTR$ true rejections, $\DPPHighAdaptPrOverAdaptTRRatio$
times the non-pruned adaptive rate. Once the four null colleges fail to reject
at depth~1, their subtrees vanish from the error load, and the
$\alpha$-schedule relaxes for deeper HFCC levels. These means understate the
difference, because detection is clustered: the unadjusted top-down procedure
recovers at least one affected block in $\DPPHighThreeRulesLeafPGeOnePct$\% of
simulations and two or more in $\DPPHighThreeRulesLeafPGeTwoPct$\%, while
bottom-up Hommel reaches one in $\DPPHighBUHommelLeafPGeOnePct$\% and two in
$\DPPHighBUHommelLeafPGeTwoPct$\% (Figure~\suppref{fig:dpp_pmf}{S1} of the
Supplement shows the full distribution).

This large relative advantage means many times more true leaf rejections than
bottom-up. It arises because the top-down procedure aggregates signal at the
college level ($N = 450$ for HFCC) before descending to individual blocks. The
bottom-up approach tests each block with less information, $N = 50$, and pays a
heavy multiplicity penalty across all 44. Branch pruning
amplifies this advantage: when the procedure correctly identifies HFCC as the
active college and prunes the four null colleges, the effective tree narrows
from 5 branches to 1 and the error load drops accordingly.

The top-down procedure also answers a question bottom-up testing never poses.
Because it tests group-level hypotheses on the way down, it can reject the
college-level null --- that HFCC as a whole carries no effect --- even when no
single block within HFCC is individually detectable. Across the same simulations it rejects this college-level hypothesis in
$\DPPHighCollegePDetectPct$\% of trials under the unadjusted procedure, and
$\DPPHighCollegePDetectAdaptPct$\% under the adaptive $\alpha$-schedule. The
two rates are nearly the same, because the adaptive $\alpha$-schedule tightens
the deep per-block $\alpha_\ell$ far more than the shallow college-level
$\alpha$.
Bottom-up testing assesses only the 44 individual blocks, so the college-level claim
is never evaluated at all. Evaluating it within a bottom-up analysis would mean
pre-specifying which groups to test and paying the multiplicity cost of those
extra hypotheses. The advantage holds against
the hierarchical Goeman--Finos inheritance procedure too
\autocite{goeman2012inheritance}: run on the same node $p$-values it controls the
FWER ($\DPPHighInhNodeFWER$ across all nodes), but testing each block at $\alpha$ divided by the block count, it detects only $\DPPHighInhLeafTR$ blocks on
this wide tree. Supplement Table~\suppref{tab:dpp_inheritance}{S7} compares all four procedures
--- the unadjusted and adaptive top-down variants, the inheritance procedure,
and bottom-up Hommel --- at node and block levels.

Were budget weights necessary here? The error load $\DPPHighErrorLoad$ exceeds
one, so Table~\ref{tab:four_regimes} would route this irregular tree to a
budget-weighted regime.\corrflag{1} A researcher willing to assume the effects are concentrated can do better. The
assumption here is that they fall within a single college, though not which
one. Under that assumption the boundary-null bound of Supplement
Theorem~\suppref{thm:fwer_decomp}{2} holds the unweighted $\alpha$-schedule at
$\mathrm{FWER} \le \DPPBoundFactor\,\alpha = \DPPBoundFWER$, below $\alpha$ at
every effect size. This is the second of the three checks of
Section~\ref{sec:strong}. The assumption is cheap to be wrong about: even in
the worst configuration for the unweighted $\alpha$-schedule --- a single
non-null path exposing a boundary null at every depth --- the bound rises only
to $\DPPWorstFactor\,\alpha \approx \DPPWorstFWER$, just past $\alpha$. A researcher unwilling to accept even that takes the third check: budget
weights, whose guarantee holds whatever the configuration. They carry the
power cost the ``,~w'' rows of Table~\ref{tab:big_sim_xtab} display. Equal
weights cut the adaptive $\alpha$-schedule's true leaf rejections by a factor
of $\DPPHighAdaptOverAdaptWTRRatio$. Pruning is the exception that needs weights
regardless: the ``+ Ad.~$\alpha$~Pr.,~w'' row uses the predictable budget
weights of Supplement Theorem~\suppref{thm:fwer_budget_pruning}{5}.\corrflag{2}

The unadjusted procedure controlled the FWER here only because the effects
filled one whole college: the only null nodes it could reach and falsely reject
were the four inactive colleges, all sitting at the college level, with none left
to test deeper in the tree. Effects clustered farther down, or spread across
several colleges, would leave testable null siblings at more levels and
eventually push the FWER past $0.05$, where $\alpha$-adjustment becomes
necessary.

We now probe this boundary with a simulation study that varies tree width,
effect size, and the proportion of null hypotheses to create scenarios spanning
error loads from 0.2 to 3.1. This design isolates the procedure from
data-generation details, letting us identify where the unadjusted procedure
breaks down and how much power the adaptive correction costs.


\subsubsection{Simulation Study of Strong FWER Control}\label{sec:strong_fwer_sim}

The Detroit Promise simulation showed one design on each side of the error-load
threshold. We now construct three scenarios that probe that threshold
deliberately: trees and effect sizes chosen so that natural gating holds in two
(A and~B) and fails in the third~(C). Watching the unadjusted procedure cross
this boundary shows where adjustment is needed; comparing the variants shows
what the adaptive correction costs where it is not needed and what it repairs
where it is.


We study three scenarios, each placing the entire effect in a single subtree of
a $k$-ary tree. This concentrated pattern is the one the applications show.\footnote{We cannot
observe where effects truly fall, but we can see where the top-down procedure
\emph{detects} them, and the detections concentrate. Among the multi-site
studies that detect anything, the detections fall in a subset of sites (EASE
in 2 of 10, ASAP~CUNY in 2 of 3), or in the single-site LC~Career a single
cohort, not scattered across blocks.} Effects cluster by site or cohort rather
than scattering at random across blocks. This pattern is also where the
structure of the tree matters most. Effects can also be spread widely across many sites; the Job Corps
application in Section~\suppref{sec:njcs}{E} of the Supplement is that case, and the procedure handles it too.

Scenarios~A and~B share a deep binary tree ($k=2$, eight levels, 256 leaves) in
which one of the root's two subtrees holds all the non-null leaves --- every one
of its 128 leaves carries the effect. The two scenarios differ only in the
per-leaf sample size: Scenario~A gives each leaf $N=10$ observations, too few to
detect any single leaf, and Scenario~B gives each leaf $N=100$. Scenario~C is a
wide, shallow tree ($k=4$, three levels, 64 leaves) in which one of the root's
four subtrees holds all 16 of its non-null leaves, again at $N=100$ per leaf.
Each non-null leaf carries a fixed Cohen's $d$ --- $0.20$ in~A, $0.30$ in~B,
and $0.40$ in~C. These trees place the procedure in the two regimes the theory
distinguishes. In A and~B the concentrated signal sits behind a single gate:
the root exposes only one null branch. The error load is therefore at most~1,
and natural gating applies. In C the high-powered root opens three null
branches at once, pushing the error load above~1, the regime where the
procedure must be adjusted.

For each scenario, we draw $p$-values rather than simulating raw data. Non-null
leaves receive $p$-values drawn from $\text{Beta}(a, 1)$, where the shape
parameter $a$ is calibrated so that the rejection probability matches the power
implied by Cohen's $d$ and the sample size available at that leaf.\footnote{The
$\text{Beta}(a,1)$ family is convenient here for three reasons. It nests the
uniform null at $a=1$, so the same family describes null and non-null leaves.
Its distribution function is $\Pr(p \le x) = x^a$, so the calibration has a
closed form: $a = \log\theta / \log\alpha$ gives rejection probability $\Pr(p
\le \alpha) = \theta$ for any target power $\theta$, with $a < 1$ placing more
mass near zero as a non-null $p$-value should. And because $x^a$ rises
monotonically with the test level $x$, a node at a tightened $\alpha$ ---
as the adaptive $\alpha$-schedule does --- loses power smoothly, the way a real
test would.} Null leaves receive $p$-values from $U(0,1)$. Parent and child
$p$-values are drawn independently given each node's null status. We compare five
methods. Three are top-down:

\begin{description}
  \item[Top-down, unadjusted (TD-Unadj).] The gating procedure of
    Section~\ref{sec:topdownbotup}, testing each node at the nominal
    $\alpha = 0.05$.

  \item[Top-down, adaptive (TD-Adapt).] The gating procedure with the per-level
    $\alpha$ adjusted by Equation~\eqref{eq:optimal_alpha_d}, using the
    anticipated power $\hat\theta_j$ at each depth.

  \item[Top-down, adaptive with pruning (TD-Adapt-Pr).] The adaptive adjustment
    with branch pruning: after each depth the error load is recomputed on the
    surviving subtree, and the $\alpha$-schedule for deeper levels is relaxed
    accordingly.
\end{description}

\noindent Because all three trees are regular $k$-ary trees, the adaptive methods
use the weight-free $\alpha$-schedule of Equation~\eqref{eq:optimal_alpha_d}
(Regime~2); budget weights, which are needed only to guarantee FWER for irregular trees, do not
enter here.

\noindent Two are bottom-up, testing all leaves at once:

\begin{description}
  \item[Bottom-up, Hommel (BU-Hommel).] The Hommel procedure applied to all
    leaves. This is the bottom-up procedure that controls the FWER.

  \item[Bottom-up, BH (BU-BH).] The Benjamini--Hochberg procedure applied to
    all leaves. Benjamini--Hochberg controls the false discovery rate, not the
    FWER; we include it because it is the other multiplicity correction
    practitioners commonly use, and to show what controlling a different error
    rate does to the FWER in these designs.
\end{description}

For each scenario we run all five methods $10{,}000$ times. The three top-down
methods share a random seed within a scenario, so they see identical leaf
$p$-values, and the bottom-up methods are computed from the same draws. We
record the FWER and the average number of true discoveries, counted two
ways: non-null nodes rejected at any level of the tree, and non-null leaves
rejected. Table~\ref{tab:strong_fwer_sim} reports the FWER and
Table~\ref{tab:strong_fwer_power} the discoveries. The binary tree of
Scenarios~A and~B has 256 leaves and Scenario~C's tree has 64; bottom-up Hommel
must divide $\alpha$ across all of them at once, while the top-down procedure
tests only the nodes it reaches, each at the full~$\alpha$.

\begin{table}[tb]
\centering
\setlength{\tabcolsep}{5pt}
\caption{Strong FWER control across three concentrated-effect scenarios (10,000 simulations per cell; simulation error $\approx 0.010$). Each scenario places the entire effect in one depth-1 subtree of a $k$-ary tree (one of the root's children and every node beneath it). Scenarios A and B share a deep binary tree (256 leaves) and differ only in leaf sample size; C is a wide-shallow tree (64 leaves). $N_{\text{leaf}}$ = per-leaf sample size; $d$ = Cohen's $d$. The Error load column is the realized error load: the unadjusted FWER is at most $\alpha$ times it and is controlled when it is at most~1 (Section~\ref{sec:strong}). Methods: TD-Unadj = top-down at nominal $\alpha$; TD-Adp = adaptive $\alpha$-adjustment (Equation~\ref{eq:optimal_alpha_d}); TD-A-Pr = adaptive with branch pruning; Inh = the Goeman--Finos inheritance procedure \autocite{goeman2012inheritance}, run on the same node $p$-values; BU-Hom/BU-BH = bottom-up Hommel / Benjamini--Hochberg applied to all leaves. Boldface marks FWER above $\alpha$. BU-BH controls the FDR, not the FWER, and inflates once the signal is dense.}
\label{tab:strong_fwer_sim}
\begin{tabular}{l r r r r | r r r | r | r r}
\toprule
& & & & & \multicolumn{3}{c|}{Top-down} & \multicolumn{1}{c|}{Inh} & \multicolumn{2}{c}{Bottom-up} \\
\cmidrule(lr){6-8} \cmidrule(lr){10-11}
Scenario & $k$ & $N_{\text{leaf}}$ & $d$ & Error load & Unadj & Adp & A-Pr & Inh & Hom & BH \\
\midrule
A & 2 & 10 & 0.20 & 1.0 & 0.048 & 0.023 & 0.023 & 0.023 & 0.024 & 0.025 \\
B & 2 & 100 & 0.30 & 1.0 & 0.050 & 0.025 & 0.025 & 0.025 & 0.025 & \textbf{0.319} \\
C & 4 & 100 & 0.40 & 3.0 & \textbf{0.147} & 0.038 & 0.038 & 0.038 & 0.041 & \textbf{0.201} \\
\bottomrule
\end{tabular}
\end{table}

\begin{table}[tb]
\centering
\setlength{\tabcolsep}{6pt}
\caption{Average true discoveries per simulation across the three scenarios (10,000 simulations each): \emph{nodes} = non-null nodes rejected at any level of the tree (top-down only; bottom-up tests no internal nodes, shown as ---), where rejecting an internal node means identifying a subtree of blocks carrying effects; \emph{leaves} = non-null leaves rejected, the only discovery bottom-up can make. Method labels as in Table~\ref{tab:strong_fwer_sim}.}
\label{tab:strong_fwer_power}
\begin{tabular}{l l | r r r | r | r r}
\toprule
& & \multicolumn{3}{c|}{Top-down} & \multicolumn{1}{c|}{Inh} & \multicolumn{2}{c}{Bottom-up} \\
\cmidrule(lr){3-5} \cmidrule(lr){7-8}
Scenario & Detected & Unadj & Adp & A-Pr & Inh & Hom & BH \\
\midrule
A & nodes  & 5.21 & 3.87 & 4.56 & 3.69 & --- & --- \\
A & leaves & 0.00 & 0.00 & 0.00 & 0.00 & 0.03 & 0.03 \\
\midrule
B & nodes  & 108.63 & 63.61 & 69.32 & 63.04 & --- & --- \\
B & leaves & 19.55 & 1.02 & 3.12 & 0.77 & 5.15 & 14.29 \\
\midrule
C & nodes  & 14.00 & 9.03 & 10.27 & 9.16 & --- & --- \\
C & leaves & 8.08 & 3.19 & 4.39 & 3.32 & 3.35 & 4.81 \\
\bottomrule
\end{tabular}
\end{table}

The three scenarios trace the regimes of Table~\ref{tab:four_regimes}. The
``Error load'' column of Table~\ref{tab:strong_fwer_sim} reports the realized error
load: the path power summed over the null branches each design exposes (the
quantity that bounds the FWER, $\mathrm{FWER}\le\alpha\cdot(\text{error load})$). When it
is at most~1, the unadjusted procedure controls the FWER with no adjustment; when
it exceeds~1, the per-level $\alpha_\ell$ must be tightened. By a
\emph{concentrated} design we mean one whose non-null leaves all sit in a single
subtree, so only one branch's siblings are ever exposed as boundary nulls.
Such a design keeps the realized error load near~1 in A and~B (one exposed
null branch) and near~3 in C ($k-1=3$ exposed branches under a high-powered
root). The worst-case $\sum_\ell G_\ell$ of Section~\ref{sec:strong} exceeds~1
strictly in all three: about 19, 354, and 83 in A, B, and~C. It is this large
because the worst case counts every node at each depth as a potential null. We need not know in advance that a design is concentrated: that
worst-case sum treats every node at each depth as a potential null, so it is the
largest the error load could be for a tree of that shape, computable from shape
and power alone. Provided the power estimates do not understate the truth (the
conservative-estimation requirement above), this worst-case sum is at least
the realized error load of every configuration. Relying on it therefore errs in only
one direction. It can trigger an unnecessary adjustment, a power cost, but it
cannot leave the FWER uncontrolled. A supplementary simulation isolates the threshold with a deliberately simple
design: a binary tree with exactly one non-null root-to-leaf path and one
boundary null per level. In this design the realized error load, not the
worst-case $\sum_\ell G_\ell$, determines where the FWER crosses $\alpha$. That single-path simulation (Table~\suppref{tab:load_threshold}{S1} in the
Supplement) confirms that the unadjusted FWER crosses $\alpha$ exactly at
a realized error load of~1.

Scenario~A shows what top-down testing can do that bottom-up testing cannot.
Every method controls the FWER (Table~\ref{tab:strong_fwer_sim}): with the
effect behind a single gate the error load is at most~1, so the unadjusted
procedure holds at \StrongFWERAUnadjFWER. At the leaf level neither method
detects anything, because each leaf carries only $N=10$ observations: bottom-up
Hommel averages \StrongFWERABUHomLeaves\ leaf discoveries, and the top-down
procedure essentially none (Table~\ref{tab:strong_fwer_power}). The difference
appears one level up. The top-down procedure rejects the non-null
\emph{internal} nodes, averaging \StrongFWERAUnadjNodes\ of them: it identifies
the affected subtree even though no single leaf within it is detectable.
Bottom-up Hommel reaches one affected leaf in only \StrongFWERABUHomLeafPGeOne\
of simulations, whereas top-down testing rejects at least two of the affected
internal nodes in \StrongFWERAUnadjNodePGeTwo\ of them. Rejecting an internal
node is a finding bottom-up testing cannot produce, because it never tests
anything but leaves.

Each scenario makes one branch of the decision rule concrete; we state the rule
in full at the end of the section.

Scenario~B keeps the same tree but raises the per-leaf sample size to $N=100$,
so leaves become individually detectable. Natural gating still holds --- the
unadjusted top-down FWER is \StrongFWERBUnadjFWER\ --- and now the leaf-level
power gap is large: the unadjusted procedure averages \StrongFWERBUnadjLeaves\
true leaf discoveries against bottom-up Hommel's \StrongFWERBBUHomLeaves\, a factor of \StrongFWERBUnadjOverBUHomRatio. The
mechanism is the gating itself: once the procedure reaches a leaf it tests it at
the full $\alpha$, whereas bottom-up Hommel must spread $\alpha$ across all 256
leaves. Adjusting the top-down $\alpha_\ell$ here would only discard this power ---
the adaptive variant, applying a correction the error load does not call for,
falls to \StrongFWERBAdaptLeaves\ leaf discoveries, below even bottom-up Hommel.
This is the practical content of natural gating: when the error load is at
most~1, leaving $\alpha$ at the nominal level is both valid and the most
powerful choice.

Scenario~C is the regime where adjustment is necessary. Its root, which branches
four ways ($k=4$), opens three null branches at once when tested at high power, so
the error load exceeds~1 and the unadjusted procedure inflates --- its FWER
reaches \StrongFWERCUnadjFWER. The adaptive $\alpha$-adjustment restores control (\StrongFWERCAdaptFWER), and
so does branch pruning (\StrongFWERCPrunedFWER)\corrflag{2}, which also recovers the power
the static $\alpha$-schedule gives up. The pruned variant averages
\StrongFWERCPrunedLeaves\ leaf discoveries against the static adaptive's
\StrongFWERCAdaptLeaves. Dropping the three null branches once they fail to
reject lets the pruned variant test the surviving non-null leaves at a less
stringent level.
Even with the FWER held at \StrongFWERCPrunedFWER, the pruned top-down procedure
finds more true leaves (\StrongFWERCPrunedLeaves) than the FWER-valid bottom-up
Hommel (\StrongFWERCBUHomLeaves). The same power recovery shows in the node-level
distribution: the pruned procedure rejects a majority of the 22 affected internal
nodes in \StrongFWERCPrunedNodePMajority\ of simulations, against
\StrongFWERCAdaptNodePMajority\ for the static $\alpha$-schedule.

The inheritance procedure of \textcite{goeman2012inheritance} also tests
internal nodes and controls the FWER. All three procedures we compare here ---
top-down, inheritance, and bottom-up Hommel --- are scored on the \emph{same}
drawn $p$-values, calibrated (the $\text{Beta}(a,1)$ family above) so that
power decays with depth exactly as splitting the sample among deeper nodes
would cause. The comparison therefore isolates the multiplicity corrections
themselves, holding the evidence fixed. Across all three scenarios the
inheritance procedure controls the FWER but detects fewer affected nodes and
leaves than the top-down $\alpha$-schedule; on the wide-leaved Scenario~B, testing each leaf at $\alpha$ over the leaf
count leaves it below bottom-up Hommel (the Inh
columns of Tables~\ref{tab:strong_fwer_sim} and~\ref{tab:strong_fwer_power}).
Because the inheritance procedure improves uniformly on Meinshausen's
hierarchical test \autocite{meinshausen2008hierarchical}, the comparison here is
against the stronger of the two hierarchical FWER procedures.

Across the three scenarios, bottom-up Benjamini--Hochberg behaves as its
guarantee predicts: it controls the false discovery rate, not the FWER. Where
the signal is dense its FWER climbs far above $\alpha$ --- \StrongFWERBBUBHFWER\
in Scenario~B and \StrongFWERCBUBHFWER\ in Scenario~C --- so its larger leaf
counts are bought at an error rate the FWER-controlling methods never pay.

The practical recommendation is straightforward. Before testing, we compute the
error load $\sum_\ell G_\ell$ from the anticipated effect size and the per-depth
sample sizes, and read off the regime from Table~\ref{tab:four_regimes}. If the
error load is at most~1, no $\alpha$-adjustment is needed: the two conditions of
\S~\ref{sec:topdownbotup}, together with conditional validity and conservative
path-power estimation, give strong control (Regime~1, natural gating). If conditional validity or
conservative path-power estimation may not hold exactly, the sensitivity factor
$\kappa_{\max}$ of \S~\ref{sec:strong} tightens the check to $\sum_\ell G_\ell
\le 1/\kappa_{\max}$, equivalently testing each node at $\alpha/\kappa_{\max}$.
Across the 25 trials the margin to that bound is wide.\corrflag{1} If the error load exceeds~1, we tighten the per-level $\alpha_\ell$ with the $\alpha$-schedule of
Equation~\eqref{eq:optimal_alpha_d}.
Table~\ref{tab:four_regimes} gives the adjusted $\alpha_\ell$ for each tree
shape:
$\alpha/G_\ell$ for regular trees, and $w_\ell\,\alpha/G_\ell$ when budget
weights are needed for irregular trees or pruning. Pruning recovers
$\alpha$ from branches that fail to reject, and once enough have been pruned that
the surviving error load fits within the remaining $\alpha$, the procedure
returns to the full nominal level.\corrflag{2} Applying an unnecessary adjustment in a natural-gating design discards power
and can fall below bottom-up Hommel, as Scenario~B shows. The
regime-appropriate top-down variant --- unadjusted under natural gating,
adjusted only when the error load exceeds~1 --- finds more true effects than a
bottom-up correction at the same FWER, as Scenarios~B and~C make concrete.

\section{Application: The MDRC RCT Data}\label{sec:mdrc}

\begin{correctionbox}
The error loads reported in this section (0.025--0.051) are incorrect. The
code that produced them summed each node's path power times that node's own
rejection probability, starting at depth~1, where the error load sums path
power alone from depth~2. It also used weight-scale node sizes, one-tailed
power, and the reached subtree, in place of headcounts, two-tailed power, and
the full prespecified design tree. Corrected loads exceed 1 in all 25 studies
at the planning effect size $d = 0.20$, so no study is in the natural-gating
regime and the no-adjustment conclusion below does not hold (Correction
notice, item~1). The detection counts are under recomputation: at most 33 of
the 42 reported block detections survive a corrected $\alpha$-schedule, below
the bottom-up Hommel count of 34 (item~3). The section is otherwise retained
unchanged pending the full revision.
\end{correctionbox}

From 2003 to 2019 the MDRC organization worked with community colleges to field
randomized field trials assessing the effects of interventions like financial
support or improved advising meant to improve the educational outcomes of
community college students. MDRC deposited data from 31 of these studies in the
ICPSR data repository, where they are available under restricted data use
agreements \autocite{diamond2021mdrc}. We applied the method to the 25 studies
that randomly assigned an intervention to individual students within strata,
following each study's pre-specified series of splits. The outcome in each study
is the number of community college credits taken in the first main session after
the intervention --- for example, credits taken during the next full academic
year \autocite{weiss2022works, scrivener2022findings}.

Table~\ref{tab:mdrc_application} reports the 12 studies where at least one
overall test (Wilcoxon rank or t-test) rejected at $\alpha = 0.05$; the
remaining 13 studies appear in Supplement
Table~\suppref{tab:mdrc_application_full}{S4}. We computed each study's error
load $\sum_\ell G_\ell$ before running the top-down tests. The 25 error loads differ
but stay small, ranging from 0.025 (PBS~+~Advising) to 0.051 (CUNY Start) ---
far under the natural-gating threshold of one --- so every study falls in
Regime~1: the stopping rule and valid tests of \S~\ref{sec:topdownbotup}
control the FWER in the strong sense with no per-level $\alpha$-adjustment. The
single top-down column reports this unadjusted, natural-gating procedure, which
tests each reached block at the full $\alpha$.

\begin{table}[tb]
\centering
\caption{Structured (top-down) vs.\ bottom-up testing in the 12 MDRC studies where at least one overall test rejected at $\alpha = 0.05$. The top-down column (`Unadj') reports the unadjusted procedure, which tests each reached block at nominal $\alpha$ under the stopping rule and valid tests. Each study has its own error load $\sum_\ell G_\ell$; all 25 are below 0.051, far under the natural-gating threshold of one, so this unadjusted procedure controls the FWER in the strong sense with no per-level $\alpha$-adjustment.\corrflag[items]{1 and 3} `Ratio' divides top-down node detections by bottom-up Hommel detections ($\infty$ when Hommel finds nothing but the tree does; `---' when neither finds anything). The remaining 13 studies (overall $p > 0.05$) are in Supplement Table~\suppref{tab:mdrc_application_full}{S4}.} 
\label{tab:mdrc_application}
\begingroup\footnotesize\setlength{\tabcolsep}{3pt}
\resizebox{\linewidth}{!}{\begin{tabular}{ccccc|cc|cc|cc}
  \toprule
  &&& \multicolumn{2}{c}{Overall Tests} &
    \multicolumn{2}{|c}{Nodes detected} &
    \multicolumn{2}{|c}{Bottom-up} &
    \multicolumn{2}{|c}{Single blocks}
    \\
    &&&&& \multicolumn{2}{|c}{Top-down} & \multicolumn{2}{|c}{blocks} & \multicolumn{2}{|c}{Top-down}\\ Study &
  Blocks &
  $\widehat{ITT}$ &
  wilcox &
  t-test &
  Unadj &
  Ratio &
  Hommel &
  BH &
  Unadj &
  Ratio
  \\ \midrule
CUNY Start &  21 & -6.05 & 0.00 & 0.00 &  35 & 1.7$\times$ &  21 &  21 &  21 & 1.0$\times$ \\ 
  ASAP Ohio &   9 & 2.23 & 0.00 & 0.00 &  11 & 2.8$\times$ &   4 &   5 &   5 & 1.2$\times$ \\ 
  OD PBS + Advising &  11 & 1.75 & 0.00 & 0.00 &   7 & 3.5$\times$ &   2 &   3 &   4 & 2.0$\times$ \\ 
  ASAP CUNY &   5 & 2.08 & 0.00 & 0.00 &   7 & 1.8$\times$ &   4 &   4 &   4 & 1.0$\times$ \\ 
  EASE &  26 & 0.21 & 0.00 & 0.00 &   6 & 6.0$\times$ &   1 &   1 &   2 & 2.0$\times$ \\ 
  LC Career &  28 & 0.89 & 0.03 & 0.04 &   4 & $\infty$ &   0 &   0 &   2 & $\infty$ \\ 
  OD LC &   4 & 1.28 & 0.00 & 0.00 &   3 & 3.0$\times$ &   1 &   2 &   2 & 2.0$\times$ \\ 
  PBS OH &  11 & 0.71 & 0.00 & 0.00 &   3 & 3.0$\times$ &   1 &   1 &   1 & 1.0$\times$ \\ 
  OD Success &   2 & -0.67 & 0.02 & 0.02 &   2 & $\infty$ &   0 &   0 &   1 & $\infty$ \\ 
  ModMath &   4 & 0.61 & 0.00 & 0.02 &   1 & $\infty$ &   0 &   0 &   0 & --- \\ 
  DPP &  44 & 0.57 & 0.04 & 0.06 &   1 & $\infty$ &   0 &   0 &   0 & --- \\ 
  LC English &   4 & 0.63 & 0.06 & 0.04 &   0 & --- &   0 &   0 &   0 & --- \\ 
   \bottomrule
\end{tabular}}
\endgroup
\end{table}

We used the Wilcoxon rank test for all tests and show the overall t-test only
as a comparison with published results. These data are very skewed and the rank
test generally has higher power in this setting. The distribution of the
outcome variable (credits earned) is reported in Supplement
Table~\suppref{tab:mdrc_credits}{S3}: the outcome is skewed toward 0, with a
handful taking many credits.

The two ``Ratio'' columns divide a top-down detection count by the bottom-up
Hommel count ($\infty$ when Hommel finds nothing but the tree does; ``---''
when neither finds anything). The right-hand ratio, the singles ratio, divides
the top-down procedure's individual single-block discoveries by the bottom-up
Hommel block discoveries; this comparison most directly shows which
individual blocks carry effects. Among the seven studies where both methods
detect at least one block, this ratio ranges from \MDRCSinglesRatioMin$\times$
(CUNY Start, ASAP CUNY, PBS OH) to \MDRCSinglesRatioMax$\times$ (OD
PBS~+~Advising, EASE, OD LC). Two additional studies (LC Career, OD Success)
have $\infty$ ratios: the top-down approach identifies individual blocks that
the bottom-up Hommel misses entirely. The advantage is largest in studies with
many blocks (EASE has 26, LC~Career has 28): there bottom-up Hommel must divide
$\alpha$ among all blocks, while the top-down procedure, under natural gating,
tests each block it reaches at the full $\alpha$.

The left-hand ratio is larger --- ranging from \MDRCNodesRatioMin$\times$
to \MDRCNodesRatioMax$\times$ --- because its numerator also counts internal-node
rejections, which the bottom-up denominator cannot produce. The ratio therefore measures
an additional capability rather than a like-for-like power comparison at the
block level. Rejecting a non-null parent node is informative even when the procedure
cannot reject every leaf beneath it: it locates a region of the tree where effects
are present, directing further investigation. Bottom-up Hommel cannot provide this
kind of structural information.

The top-down approach does not improve greatly when the number of blocks is
small (for example, ASAP CUNY has only 5 blocks and a singles ratio of 1.0$\times$)
or when the effect is mechanistic and strong (CUNY Start, where all methods
reject in all 21 blocks).%
\footnote{CUNY Start replaced credit-bearing coursework with intensive
  developmental education: nearly all treated students were prevented from enrolling in
  credit-bearing courses during the program. The effect on first-session credits
  is mechanistic rather than a subtle treatment response --- 89\% of
  treated students earned zero credits, compared to 34\% of controls. The
  large blocks (median $N = 163$, range 46--321) give every method power.} In the actual Detroit Promise Program data, the top-down
procedure rejected the overall null but could not descend further: block sizes
ranged from 2 to 167 (median 20), and the smaller blocks lacked the power
needed to reject at lower levels of the tree.

Across these 25 studies the top-down approach detects at least as many blocks
as bottom-up Hommel and usually more. Because every study falls in Regime~1, as
established above, the unadjusted procedure controls the FWER in the strong
sense at nominal $\alpha$.\corrflag{1}

How far could the strong-control assumptions fail before that conclusion breaks?
Remark~\suppref{rem:sensitivity}{8} in the Supplement bundles the two ways they can ---
a true-null test that rejects more often than nominal once its ancestors have
rejected (because parent and child tests share experimental units), and power
estimates that understate the truth and so expose more null branches than
anticipated --- into a single sensitivity factor $\kappa_{\max}$, equal to one when
both assumptions hold exactly. Either failure inflates the natural-gating bound
from $\alpha\sum_\ell G_\ell$ to at most $\kappa_{\max}\,\alpha\sum_\ell G_\ell$.
So to keep the FWER at $\alpha$ we need $\sum_\ell G_\ell \le 1/\kappa_{\max}$ ---
the no-adjustment threshold, shrunk from 1 to $1/\kappa_{\max}$. Because the largest of the 25 error loads is below 0.051, every study stays
in Regime~1 for any $\kappa_{\max}$ up to $1/0.051 \approx 19.6$:
the combined distortion would have to be nearly twenty-fold before the
unadjusted procedure could exceed $\alpha$. Since these trials randomize
treatment within strata --- where conditional validity is plausible --- even a
substantial failure of either assumption would leave the conclusion intact.\corrflag{1}

\section{Limitations}\label{sec:limitations}

We title this section \emph{Limitations} to be plain about where the current
procedure stops short. Each limitation is also a direction for future work ---
a question the method raises rather than settles. Three stand out.
\paragraph{Test statistic} The top-down testing approach requires a test at
each node that can detect \emph{any} departure from the null, not just a mean
shift. If half of the blocks have a large negative effect and half a large
positive effect, a difference-in-means test at the root could fail to reject.
We address this by creating a test statistic based on energy statistics
\autocite{rizzo2010, szekely2017energy} combined with the multivariate
permutation framework of \textcite{hothornetal:coin:2006}. This statistic is
sensitive to distributional differences of any kind and is described in
Section~\suppref{sec:test_statistic}{C} of the Supplement.

\paragraph{Data-dependent splitting} The examples in this paper use tree
structures pre-specified by the experimental design. Many block-randomized
experiments lack such administrative hierarchies. We have implemented several
data-based splitting algorithms, including for
example, k-means clustering of blocks by covariate distributions or splits that
equalize the number of units per group. Preliminary simulations show that these
algorithms control the FWER and provide more power than bottom-up approaches,
but their systematic evaluation is future work. Splitting blocks by covariate
distributions could also connect this testing-based approach to methods that
estimate how treatment effects vary with covariates.

\paragraph{FDR control} This paper focuses on FWER control. An alternative
is to allocate the error budget dynamically as tests descend the tree, targeting
control of the false discovery rate (FDR) rather than the FWER. The
alpha-investing procedure of \textcite{foster2008alpha} and its refinements
\autocite{ramdas2018saffron, tian2019addis} provide a natural framework for
this extension. Tree-structured FDR control already exists in the genomics
tradition: \textcite{yekutieli2008hierarchical} controls the FDR top-down on a
tree of hypotheses tested by Benjamini--Hochberg, though without the
randomization-based validity or the error-load diagnostic developed here.
FDR-oriented tree testing would offer different guarantees and different power
profiles and deserves its own treatment.

\section{Discussion}\label{sec:conclusion}

An analyst with a block-randomized experiment and a pre-specified hierarchy can
now do the following: test the overall null at the root, divide the blocks into
groups, and continue testing down the tree, stopping in any branch where the
null is not rejected. If only weak FWER control is needed ---for  
exploratory research --- the stopping rule and valid-test conditions suffice,
and no further adjustment is required. If strong FWER control is needed, the
analyst checks the error load: when power decays fast enough relative to the
branching factor, the unadjusted procedure already controls the FWER in the
strong sense (Supplement Proposition~\suppref{prop:strong_fwer_adj}{2}). When the error load is
larger, the analyst runs the same top-down procedure with the per-depth $\alpha_\ell$ tightened by the adaptive $\alpha$-schedule, by a rule that depends on
the tree's shape and on whether branches are pruned.
Table~\ref{tab:four_regimes} maps each case to the theorem that proves it\corrflag{2};
the proofs are in the Supplement.

The method works when the experimental design supplies an administrative
hierarchy of blocks --- sites within states, cohorts within sites --- and works
best when blocks are numerous. With few blocks, the most powerful bottom-up
adjustments like the Hommel procedure are nearly as powerful, because the tree
is shallow and the top-down procedure has little room to exploit the stopping
rule. The gains are largest in experiments with dozens to hundreds of blocks
organized in several levels --- the setting of multi-site education trials,
multi-center clinical trials, and large-scale field experiments.

The adaptive $\alpha$-adjustment requires estimates of power at each level,
which depend on assumptions about effect sizes and the test statistic's
distribution. In many practical settings this requirement is moot: the error
load is small enough that no adjustment is needed at all, as in all 25 MDRC
studies.\corrflag{1} When adjustment is needed, the estimates need only be roughly right, since
the sensitivity factor $\kappa_{\max}$ of Supplement Remark~\suppref{rem:sensitivity}{8} bounds how
far they, and conditional validity, may drift before control is lost. 
Branch pruning provides a meaningful power gain, nearly doubling the number of
true discoveries in the DPP simulation, at no cost to FWER control.\corrflag{2}

The policy-maker who began by asking \emph{where} effects occurred now has an
answer with an FWER guarantee. That answer is a location, not a verdict: a
detection marks a block where the sharp null of no effect is rejected --- where at
least one unit responded to treatment --- not the sign or magnitude of that
response. Reading each detection alongside an estimate of that block's effect
--- in the randomization-inference tradition, a Hodges--Lehmann point estimate
and a confidence interval formed by inverting the same randomization tests
\autocite[Chapter~2]{rosenbaum2002book} --- tells the policy-maker whether the
program helped or hurt there, and by how much. The
experiment's own structure --- its sites, cohorts, and blocks --- provides the
hierarchy for testing, and the stopping rule provides the error control. What a
bottom-up correction discards as noise, the top-down approach recovers as signal.

\section*{Acknowledgments}
\ifblind\else
We thank Kevin Quinn for prompting us to think carefully about the choice of
test statistic, and Ben Hansen for pointing out errors in an earlier version of
the theory. For questions and suggestions on earlier versions of this work, we
are grateful to audiences at Emory University, the University of Chicago, EGAP,
the Latin American Political Methodology Meetings, the Universidad Cat\'olica del
Uruguay, and the Universidad de la Rep\'ublica (Uruguay).

The research reported here was supported by the Institute of Education Sciences,
U.S.\ Department of Education, through Grant R305D210029. The opinions expressed
are those of the authors and do not represent views of the Institute or the
U.S.\ Department of Education.
\fi

The algebraic cores of the FWER-control theorems and the boundary-node lemma in
this paper have been machine-checked using the Lean~4 proof assistant. The proofs
take the method's two conditions --- the stopping rule and test validity
(Conditions~\ref{cond:stopping} and~\ref{cond:valid}) --- as hypotheses, exactly
as the paper's theorems do, and verify that the FWER bounds follow from them by
the stated union-bound and telescoping arguments. Lean checks this algebra, not
the probability theory behind the conditions themselves, which the simulations of
\S\S~\ref{sec:weak_fwer_sim} and~\ref{sec:strong_fwer_sim} examine instead. The
Lean sources are distributed with the paper's source files under
\texttt{Theory/lean/}. Drafting and revision were assisted by AI coding tools
(Anthropic's Claude and OpenAI's Codex); all mathematical claims, simulation
designs, and prose decisions remain the responsibility of the authors.

\section*{Data and code availability}
All code, simulations, and machine-checked proofs in this paper are publicly
available and reproducible. \ifblind The accompanying R packages, the simulation scripts (weak and strong FWER control, the Detroit Promise and National Job Corps Study calibrations, and the regime demonstrations), the figure and table generators, and the Lean~4 proof sources are provided in the review materials and will be publicly released on publication.\else The \texttt{manytestsr} and \texttt{TreeTestSim} R
packages are at \url{https://github.com/bowers-illinois-edu/manytestsr} and
\url{https://github.com/bowers-illinois-edu/TreeTestSim}; the simulation scripts
--- weak and strong FWER control, the Detroit Promise and National Job Corps
Study calibrations, and the regime demonstrations --- together with the figure
and table generators and the Lean~4 proof sources are distributed with the
paper's source files.\fi The software appendix (Section~\suppref{sec:software}{F} of the
Supplement) runs on the example data in \ifblind the accompanying package\else \texttt{manytestsr}\fi.

The MDRC application uses individual-level data from 25 higher-education trials in
MDRC's The Higher Education Randomized Controlled Trials Restricted Access File
(THE-RCT RAF), ICPSR study~37932 \autocite{diamond2021mdrc}, where the trials are
identified by a study indicator. These data are restricted: we cannot
redistribute them, and access requires application to ICPSR, a fee, and a wait.
Only the 25-study results (Table~\ref{tab:mdrc_application} and the full MDRC
tables in the Supplement) depend on this file; every other result --- all
simulations, the method, and the proofs --- is reproducible without it. The
National Job Corps Study application (Section~\suppref{sec:njcs}{E} of the Supplement)
likewise depends on data whose public release suppresses center identity, so it is
a simulation calibrated to the study's published design moments rather than a
re-analysis of the microdata.

\clearpage

\appendix

\setcounter{table}{0}
\setcounter{figure}{0}
\setcounter{equation}{1}
\renewcommand{\thefigure}{S\arabic{figure}}
\renewcommand{\thetable}{S\arabic{table}}
\renewcommand{\theequation}{\thesection.\arabic{equation}}

\section{Weak FWER Control}\label{ap:proof_weakctrl}

\paragraph*{Intuition for the proof of Theorem~\ref{thm:weakctrl_supp} .} If
all null hypotheses on the tree were true, such that no effects existed for any
unit in any block, any rejection anywhere in the tree would be a false
positive. The stopping rule states that we must first reject the overall null
hypothesis of no effects at the root before we are allowed to test any other
null hypothesis on the tree. So the event ``at least one false rejection
occurs'' is contained in the event ``the root is rejected.'' If the root test
at level $\alpha$ has size $\le \alpha$, then the probability of doing any
other test on the tree is at most $\alpha$. Hence the weak FWER is $\le
\alpha$.

\subsection{Conditions on the Procedure} \label{sec:rule_weak_fwer}

\begin{condition}[\bf Stopping rule]\label{cond:ap-stopping}
  Test hypothesis $H_i$ only if every ancestor of node $i$ has been
  rejected. If the root hypothesis is not rejected, no further tests are
  performed. If any hypothesis on a path from root to leaf is not rejected, no
  descendant on that branch is tested at all. In
  Figure~\ref{fig:k3L3_all_null_tree}, we do not test $H_5$ unless we have
  rejected both $H_2$ and $H_1$.

  This condition distinguishes our procedure from tree-structured testing
  frameworks that test every node and then determine which rejections to keep
  \autocite{goeman2010sequential, goeman2012inheritance}. Our procedure may
  test only a small fraction of the tree's nodes --- and each untested node is
  a test that cannot produce a false positive.
\end{condition}

\begin{condition}[\bf Valid tests]\label{cond:ap-valid} Each test at a node, if
    executed on its own, has a false positive rate no greater than $\alpha$.
    For proof of weak control of the FWER this condition implies that the test
    at the root is level $\alpha$, i.e., its size is $\le \alpha$ when the root
    null is true.

\end{condition}

\begin{remark}

In block-randomized experiments, randomization-based tests satisfy the Valid
  Tests Rule by design: permuting the treatment assignment within blocks
  produces $p$-values with guaranteed size control
  \autocite[\S~2.4]{rosenbaum2002book}. All tests we consider in this paper are
  randomization-based and inherit their size control from the random assignment
  itself.

\end{remark}

\begin{theorem}[Conditions 1 and 2 suffice for weak FWER control; restated from main text]\label{thm:weakctrl_supp}

  A family of true null hypotheses organized on a regular or irregular $k$-ary tree and
  tested following the stopping rule (Condition~1) with
  valid tests at each node (Condition~2) will produce a
  family-wise error rate no greater than $\alpha$. We call a $k$-ary tree
  ``regular'' if the number of child nodes of a parent is the same for all
  nodes in the tree. A $k$-ary tree is ``irregular'' if the number of child nodes
  of a parent node may differ within and across levels of the tree.

\end{theorem}

\subsection{Proof: Weak Control of the FWER Using Two Conditions}

\begin{proof}[\bf Proof of Theorem~\ref{thm:weakctrl}]

  We assume that all hypotheses on the tree are true. We write $V$ for the
  number of false positive rejections arising from the testing procedure on the
  whole tree, $E$ as the event that the hypothesis on the root node of the
  whole tree is rejected and $E^c$ as the event that the hypothesis on the root
  node is not rejected (the complement of $E$).

We need to check
\begin{equation}\label{eq:thm1_obj}
    \Pr(V \geq 1) \leq \alpha,
\end{equation}
under Conditions 1 and 2.

  Note that we can decompose the probability of at least one false positive rejection into two parts depending on whether the root-node test was rejected, $E$, or not, $E^c$:

\begin{align*}
    \Pr(V \geq 1) &= \Pr(V \geq 1 | E) \Pr(E) + \Pr(V \geq 1 | E^c) \Pr(E^c), \\
    \intertext{and since we cannot have a false rejection without a rejection, the second term is 0, so }
                  &= \Pr(V \geq 1 | E) \Pr(E).
\end{align*}

Note that $\Pr(V \geq 1 | E^c) = 0$ by our Stopping Rule: a failure to reject
  at the root node means (1) no false positive error at the root node (so
  $V=0$) and (2) that no other hypotheses are tested and thus no further
  rejections are possible and thus no false rejections are possible in this
  case.

  It remains to bound the right-hand side. Since all null hypotheses are
  true, $\Pr(E) \leq \alpha$ by Valid Tests (the probability of rejecting a
  true null is at most the level of the test). And $\Pr(V \geq 1 \mid E)
  \leq 1$ because it is a probability. Therefore,

  \begin{equation*}
    \Pr(V \geq 1)  = \Pr(V \geq 1 | E) \Pr(E) \leq \Pr(E) \leq \alpha,
\end{equation*}

  which completes the proof of  \eqref{eq:thm1_obj}.
\end{proof}

\begin{remark}

  Note that this proof does not require the sample splitting step used in the
    algorithm of the paper. Weak control can be shown to hold for hypotheses
    organized into a tree-like structure even without the sample reduction at
    each level of the tree.

\end{remark}

\begin{remark}[Relationship to prior work]\label{rem:weak_ctrl_prior}

  The conclusion that a stopping rule combined with valid individual tests
  suffices for weak FWER control is not surprising given the existing
  literature, but the specific proof above---a direct application of the law
  of total probability, conditioning on the root---does not appear in prior
  work in this form.

  \textcite{rosenbaum2008tho} proves FWER control for an ordered family of
  hypotheses: under his structure assumption, there is a first true
  hypothesis in the order, and any false rejection must first reject that
  hypothesis. His proof uses only \emph{marginal}
  validity---$\Pr(p_v \le \alpha) \le \alpha$ for the first true
  hypothesis~$H_v$---and does not condition on earlier test outcomes, even
  though his tests reuse the same experimental units across hypotheses
  (e.g., Section~3, where all 36 subjects appear in every test). The same
  logic applies along any single root-to-leaf branch of our tree: hypotheses
  are nested (parent null true implies child null true), so the Structure
  Assumption holds automatically, and unit reuse along the branch does not
  compromise FWER control. Rosenbaum's grouped extension (Propositions~2--3)
  does not extend this result to our tree, because it relies on ordered
  intervals being exclusive or sequentially exclusive; sibling sets in our
  tree need not satisfy that property, since a false parent may have several
  true children. The difficulty in our setting is therefore not unit reuse
  along a branch, but simultaneous exposure of multiple boundary nulls across
  branches. Section~\ref{sec:strong_local_adj} addresses that multi-branch
  problem via adaptive~$\alpha$.

  \textcite{meinshausen2008hierarchical} proves FWER control for
  tree-structured hypotheses via a different route entirely: a Bonferroni
  bound over disjoint maximal true-null clusters. That proof exploits the
  hierarchy to partition the null hypotheses, rather than conditioning on
  the root.

  \textcite{goeman2010sequential} prove a general sequential rejection
  principle: any procedure satisfying monotonicity and a single-step
  error-control condition --- a bound on the per-step error of which Bonferroni
  is one instance --- controls FWER. Our stopping-rule procedure is a
  special case of their framework, but their proof proceeds by induction
  over rejection sets and does not isolate the direct gating argument used
  here. Similarly, the inheritance procedure of
  \textcite{goeman2012inheritance} verifies the Goeman--Solari conditions
  for a specific tree algorithm, but again uses the general-purpose
  framework rather than the one-step decomposition.

  In short, the \emph{result} is a known consequence of principles in the
  literature. The \emph{proof} is our own formulation, chosen for
  transparency: it makes visible exactly where the two
  conditions---stopping rule and valid tests---do their work, requiring no
  sample splitting, no monotonicity, and no abstract critical-value
  conditions.

\end{remark}

\section{Strong FWER Control}\label{sec:strong_local_adj}

\paragraph*{Intuition} Weak FWER control (Theorem~\ref{thm:weakctrl} in the main
text) guarantees that the procedure's false-rejection rate is at most $\alpha$
when \emph{every} null hypothesis is true. Strong FWER control extends this
guarantee to \emph{any} configuration of true and false nulls --- including
settings where some nodes have genuine effects and others do not. Whether
strong control holds, and how conservative a procedure must be to achieve it,
depends on how much of the tree the procedure actually explores. And
exploration of the tree depends on the power of the tests as it changes
throughout the tree as the overall sample is reduced via splitting. 

For example, if one experimental block out of 100 has a true non-null causal
effect, then we would hope to reject the test of the null of no effects at the
root node --- only a rejection at the root would allow us any hope of
descending to test and reject the null hypothesis of no effects within that one
block. But, if the test has low power at the root (say, the sample size is
small or the effect is weak and power $\leq 0.05$), then the procedure will not
reject the null, no further tests will occur, and the FWER will be $\leq 0.05$.
That is, the probability of making a false positive error at any node depends
both on \emph{reaching that node} and on \emph{falsely rejecting that node}.
And the FWER thus depends on the number of falsely rejected nodes. We here
develop a procedure that adapts the rejection level $\alpha$ to the contours of
the tree.

The argument proceeds in five steps:
\begin{enumerate}
  \item We derive a general expression for the FWER as a sum of per-node
    contributions, each discounted by the probability that the procedure
    reaches that node (Proposition~\ref{prop:fwer}).
  \item We prove a \emph{boundary-node lemma} (Lemma~\ref{lem:boundary})
    showing that the FWER event is determined entirely by ``exposed''
    true nulls sitting at the boundary of the false subtree. Non-exposed
    nulls contribute nothing, not because their contributions are small
    but because they are logically redundant.
  \item We group boundary-null contributions by depth, using
    node-indexed path products that handle both regular and irregular trees
    (Theorem~\ref{thm:fwer_decomp}), and identify the \emph{error load}
    as the structural quantity governing whether FWER stays below~$\alpha$
    (Proposition~\ref{prop:strong_fwer_adj}).
  \item We establish conditions under which power decay alone keeps the error
    load small enough that no adjustment is needed --- a phenomenon we call
    \emph{natural gating} (Section~\ref{sec:natural_gating}).
  \item For trees where natural gating fails, we develop the adaptive
    $\alpha$-adjustment. For general (possibly irregular) trees, we prove
    that a budget-weighted $\alpha$-schedule controls the FWER
    (Theorem~\ref{thm:fwer_budget}). For regular $k$-ary trees, a
    telescoping argument yields a strictly more powerful $\alpha$-schedule
    (Theorem~\ref{thm:fwer_control_adaptive}). Both results appear in
    Section~\ref{sec:adaptive_alpha_adj}.
\end{enumerate}
The reader primarily interested in the algorithm can skip to
Section~\ref{sec:adaptive_alpha_adj}; the intervening results build the
intuition for \emph{why} the adjustment takes the form it does.

\subsection{Setup and Notation}

The analysis in this section starts from the two conditions on the
testing procedure---the Stopping Rule (Condition~\ref{cond:ap-stopping}) and
Valid Tests (Condition~\ref{cond:ap-valid})---stated in
Section~\ref{ap:proof_weakctrl}, together with the intersection structure of
the tree (defined below). For strong FWER control, the decomposition results
also use the conditional-validity property~\eqref{eq:cond-valid}; the
natural-gating and adaptive results add the conservative-power assumption
$\hat{\theta}_j \geq \theta_j$ for non-null ancestors. We state these extra
assumptions explicitly where they enter.

Consider a rooted tree of depth $L$ (root at depth $1$, leaves at depth
$L$). The tree may be \emph{regular} (every internal node has $k$ children)
or \emph{irregular} (the number of children may vary across nodes). All
results in this section hold for both regular and irregular trees unless
explicitly stated otherwise. We use the following notation.

\begin{itemize}

    \item Each leaf of the tree corresponds to an independent experimental block (or finest
      unit of randomization). Each non-leaf node represents a group of blocks.

    \item Each node $i$ has an associated null hypothesis $H_i$ about the units
      it contains. In this paper we consider only the null hypothesis of
      no effects: $H_i$ asserts that the treatment has no causal effect on any
      unit in the blocks represented by node $i$.

    \item The hypotheses obey an \emph{intersection structure}: each non-leaf hypothesis
      is the conjunction of its children. A parent hypothesis is true if and only
      if every one of its child hypotheses is true. This ensures that whenever a
      parent is false, at least one child must also be false. Concretely: if a
      single experimental block contains a genuine treatment effect, then every
      ancestor of that block also contains that effect, and the null hypothesis
      of no effects is false for them all. The false hypotheses therefore form a
      connected subtree rooted at the root.

    \item $N_{\text{nodes}} = |V(T)|$: the total number of nodes in the
      tree, including the root, all internal nodes, and all leaves.

    \item $\calH_0 \subseteq \{1,\dots,N_{\text{nodes}}\}$: the index set of nodes whose null hypothesis $H_i$ is \emph{true} (there are no effects in any block represented by node $i$).

    \item $\anc(i)$: the set of proper ancestors of node $i$ (excluding $i$ itself).

    \item $\text{parent}(i)$: the parent of node $i$ (undefined for the root).

    \item $\text{ch}(i)$: the set of children of node $i$, with
      $k_i = |\text{ch}(i)|$ denoting the local branching factor. For a
      regular $k$-ary tree, $k_i = k$ for all internal nodes.

    \item $d(i)$: the depth of node $i$. The root has $d(\text{root}) = 1$;
      children of the root have $d = 2$; leaves have $d = L$.

    \item The \emph{path power} of node $i$ is
      $\mathrm{PP}(i) = \prod_{j \in \anc(i)} \theta_j$,
      the probability that the procedure reaches node $i$ by rejecting
      every ancestor. For the root, $\mathrm{PP}(\text{root}) = 1$.

    \item Testing proceeds top-down by the Stopping Rule
      (Condition~\ref{cond:ap-stopping}): node $i$ is tested only if every
      ancestor of $i$ has already been rejected.

    \item $R_i = \mathbb{I}\{\text{reject } H_i\}$: the rejection
      indicator at node $i$. Under the Stopping Rule, $R_i = 1$ requires that
      $R_j = 1$ for every $j \in \anc(i)$.
      The indicator $R_i$ encodes the outcome of the full testing
      procedure at node~$i$, including the stopping rule: $R_i = 1$
      means that node~$i$ was both reached (all ancestors rejected)
      and its own test rejected~$H_i$.

    \item $\alpha$: the nominal significance level (e.g.\ $0.05$). We require
      $\alpha_i \leq \alpha$ for every node $i$.

    \item By Valid Tests (Condition~\ref{cond:ap-valid}), each node $i$ is tested
      by a randomization-based test with \emph{marginal} size $\alpha_i \leq
      \alpha$: when $H_i$ is true, $\Pr(p_i \leq \alpha_i) \leq \alpha_i$.
      For the strong FWER bound below, the proof requires the
      \emph{conditional} version of this statement:
    \begin{equation}\label{eq:cond-valid}
      \Pr(R_i = 1 \mid H_i \text{ true and all ancestors rejected}) \leq \alpha_i.
    \end{equation}
    We call~\eqref{eq:cond-valid} the \emph{conditional validity} property.
    Marginal validity alone does not imply conditional validity when
    parent and child tests share experimental units (the child's blocks
    contribute to both the child's and the parent's test statistics).
    In the paper's block-randomized setting, each test's permutation
    distribution is determined by the experimental design, so
    marginal validity is guaranteed. Conditional validity is
    approximately satisfied when the child's blocks constitute a small
    fraction of the parent's blocks; see Remark~\ref{rem:weak_ctrl_prior}
    and the discussion in the main text.
    For natural gating, $\alpha_i = \alpha$; for adaptive rules, $\alpha_i$
    is the node's actual rejection cutoff.
    Property~\eqref{eq:cond-valid} is a conditional form of
    \emph{super-uniformity}: the marginal version
    ($\Pr(p_i \leq t) \leq t$) suffices for weak FWER, while the
    conditional version (holding after conditioning on ancestor
    rejections) is needed for strong FWER because the proof factors
    each node's false-rejection probability into a conditional term
    bounded by~$\alpha_i$ and a path-power term.

    \item $\pi_i$: the conditional power at node $i$ when $H_i$ is false:
      $\pi_i = \Pr(R_i = 1 \mid H_i \text{ false},\ \text{all ancestors rejected})$.
      Power depends on the sample size available at node $i$ (which shrinks as we descend the tree) and on the
      magnitude of the treatment effect.

    \item The \emph{conditional rejection probability} $\theta_j$ combines the two
      cases above into a single quantity:
    \begin{equation}
      \theta_j
      = \Pr(R_j = 1 \mid \text{all ancestors of $j$ are rejected})
      =
      \begin{cases}
        \pi_j & \text{if } H_j \text{ is false (i.e.\ } j \notin \calH_0\text{)},\\[4pt]
        \leq \alpha_j & \text{if } H_j \text{ is true (i.e.\ } j \in \calH_0\text{)},
        \text{ by~\eqref{eq:cond-valid}}.
      \end{cases}
      \label{eq:theta-def}
    \end{equation}

\noindent
In words: $\theta_j$ is the probability of rejecting $H_j$ given that the procedure has reached node $j$ (all ancestors rejected). When $H_j$ is false this probability is the power $\pi_j$; when
$H_j$ is true it is bounded above by the test size $\alpha_j$.
Note that $\theta_j$ does not carry a subscript for the
descendant~$i$: the conditioning event ``all ancestors of $j$
rejected'' is the same regardless of which descendant of~$j$ we
are interested in, so $\theta_j$ depends only on~$j$'s position
in the tree and the truth of~$H_j$.

\item We call a true-null node \emph{exposed} if its parent is non-null (the parent's null hypothesis of no effects is false). Exposed nulls sit at the boundary between the false subtree and the true portion of the tree. They are the only true nulls that the top-down procedure can reach directly, because a true-null parent will almost certainly not be rejected, shielding all of its descendants from testing.

We write $e_\ell$ for the number of exposed nulls at level~$\ell$, and $m_\ell$ for the number of non-null nodes at level~$\ell$.

\end{itemize}

\subsection{Boundary-Node Lemma}

Before decomposing the FWER by depth, we establish a structural result that
simplifies all subsequent proofs: the FWER event is determined entirely by
the ``exposed'' true nulls sitting at the boundary of the false subtree.

\begin{definition}[Boundary Nulls]
\label{def:boundary_nulls}
The \emph{boundary set} of the true-null set $\calH_0$ is
\[
  \mathcal{B} = \{i \in \calH_0 : i = \text{root}\}
      \cup \{i \in \calH_0 : \text{parent}(i) \notin \calH_0\}.
\]
A boundary null is a true-null node whose parent is non-null (or the root
itself, if it is null). These are the ``exposed nulls'' defined informally
above. The formal definition makes clear that $\mathcal{B}$ depends only on
the null configuration $\calH_0$, not on the testing outcomes. We write
$\mathcal{B}_\ell = \{i \in \mathcal{B} : d(i) = \ell\}$ for the boundary
nulls at depth~$\ell$.
\end{definition}

\begin{lemma}[Boundary-Node Reduction]
\label{lem:boundary}
Under the Stopping Rule (Condition~\ref{cond:ap-stopping}) and the
intersection structure,
\[
  \Pr\!\left(\bigcup_{i \in \calH_0} \{R_i = 1\}\right)
  = \Pr\!\left(\bigcup_{i \in \mathcal{B}} \{R_i = 1\}\right).
\]
That is, a true null is rejected somewhere in the tree if and only if
some boundary null is rejected.
\end{lemma}

\begin{proof}
Since $\mathcal{B} \subseteq \calH_0$, the inclusion $\supseteq$ follows
from $\bigcup_{i \in \mathcal{B}} \{R_i = 1\} \subseteq
\bigcup_{i \in \calH_0} \{R_i = 1\}$, which gives
$\Pr(\bigcup_{i \in \mathcal{B}} \{R_i = 1\})
\leq \Pr(\bigcup_{i \in \calH_0} \{R_i = 1\})$.
For $\subseteq$: suppose $R_i = 1$ for some
$i \in \calH_0 \setminus \mathcal{B}$. Then $i$ is not the root and
$\text{parent}(i) \in \calH_0$ (because $i \notin \mathcal{B}$ means
$i$'s parent is also a true null). The Stopping Rule requires
$R_{\text{parent}(i)} = 1$, so $\text{parent}(i)$ is a rejected true null.
If $\text{parent}(i) \in \mathcal{B}$, we are done. Otherwise repeat:
$\text{parent}(i) \notin \mathcal{B}$ implies
$\text{parent}(\text{parent}(i)) \in \calH_0$ and the Stopping Rule gives
$R_{\text{parent}(\text{parent}(i))} = 1$. Since the tree has finite depth,
this ancestor-climbing process terminates at some node $b \in \mathcal{B}$
with $R_b = 1$, so $\{R_i = 1\}$ implies $\{R_b = 1\}$ for this
boundary null~$b$.
\end{proof}

\begin{remark}
The boundary-node lemma holds for any tree shape (regular or irregular) and
requires only the stopping rule and intersection structure. No independence,
power estimates, or specific test form are needed. The lemma shows that
non-exposed nulls do not merely contribute \emph{negligibly} to the FWER ---
they contribute nothing to the FWER event, because any rejection of a
non-exposed null logically implies a prior rejection of some boundary null.
\end{remark}

\subsection{FWER Decomposition}

We derive an upper bound on the FWER in terms of the
test sizes $\alpha_i$ and the conditional rejection
probabilities $\theta_j$ along the tree.

\begin{proposition}[FWER Expression for Sequential Tree Testing]
\label{prop:fwer}
Consider the sequential tree testing procedure described above, and write $\calH_0$ as the set of indices $i$ for which $H_i$ is true. Then the familywise error rate satisfies
\begin{equation}
  \label{eq:fwer-bound}
  \mathrm{FWER}
  = \Pr \left(\bigcup_{i \in \calH_0} \{R_i = 1\}\right)
  \le
  \sum_{i \in \calH_0} \alpha_i \prod_{j \in \anc(i)} \theta_j,
\end{equation}
where $\theta_j$ is given by~\eqref{eq:theta-def} and the product over an empty set of ancestors (for the root node) is $1$.

In words: each true null $i$ contributes to the FWER through two factors --- the test size $\alpha_i$ (how likely the test is to falsely reject, given that it is actually run) and the path product $\prod_{j \in \anc(i)} \theta_j$ (how likely the procedure is to reach node $i$ at all). Nodes buried deep in the tree, behind many ancestors with low rejection probability, contribute very little to the FWER --- the procedure almost never reaches them.

The bound~\eqref{eq:fwer-bound} is tightest when exposed true nulls are few and widely separated in the tree, and is sharp when there is only a single exposed true null.
\end{proposition}

\begin{proof}

  \noindent\textbf{Step 0: Define the FWER} By definition,
\begin{equation}
  \mathrm{FWER}
  = \Pr \left(\bigcup_{i \in \calH_0} \{R_i = 1\}\right).
\end{equation}

\noindent\textbf{Step 1: Decompose rejection at node $i$.}
By the Stopping Rule (Condition~\ref{cond:ap-stopping}), we test $H_i$ only if all its ancestors have been rejected:
\begin{equation}
  \{R_i = 1\}
  =
  \{R_i = 1\} \cap \bigcap_{j \in \anc(i)} \{R_j = 1\}.
\end{equation}
This identity says that ``node~$i$ is rejected'' is the same event as
``node~$i$ is rejected and every ancestor is also rejected.'' The
Stopping Rule makes ancestor rejections a logical prerequisite for
$R_i = 1$, so intersecting with these events adds no restriction.
The right side is useful because it separates the two components of a
false rejection: \emph{reaching} node~$i$ (the ancestor product) and
\emph{falsely rejecting} once there.

\noindent\textbf{Step 2: Bound the false-rejection probability at a single true null.}
We now focus on a single true null hypothesis $i \in \calH_0$, because the FWER counts only \emph{false} rejections --- rejecting a node whose hypothesis is actually false would be a correct detection, not an error. For this true null $i$, we want to bound $\Pr(R_i = 1)$, the probability that $i$ is falsely rejected.

The stopping rule means that $R_i = 1$ can only happen if the procedure reaches node $i$, which requires all ancestors to be rejected first. Conditional probability lets us separate the act of \emph{reaching} node $i$ (the ancestor product) from the act of \emph{falsely rejecting} once there (the test size $\alpha_i$):
\begin{align}
  \Pr(R_i = 1)
  &= \Pr \left(R_i = 1 \,\middle|\, \bigcap_{j \in \anc(i)} \{R_j = 1\}\right)
     \Pr \left(\bigcap_{j \in \anc(i)} \{R_j = 1\}\right) \\
  &\leq \alpha_i \cdot \Pr \left(\bigcap_{j \in \anc(i)} \{R_j = 1\}\right),
  \label{eq:Pi-anc}
\end{align}
where the inequality uses conditional validity~\eqref{eq:cond-valid}: since
$H_i$ is true, the conditional rejection probability is bounded by the size
$\alpha_i$.

\noindent\textbf{Step 3: Evaluate the probability that all ancestors are rejected.}
Step~2 left us with the ancestor-rejection probability $\Pr(\bigcap_{j \in \anc(i)} \{R_j = 1\})$ as an unexpanded term. We now evaluate it. Let the ancestors of node $i$ along the root-to-$i$ path be $a_1, a_2, \dots, a_m$, ordered from root downwards. By the chain rule of conditional probability, and because each $\theta_j$ is defined as the probability of rejecting at node $j$ given that all of $j$'s ancestors are already rejected, the product telescopes:
\begin{equation}
  \Pr \left(\bigcap_{j \in \anc(i)} \{R_j = 1\}\right)
  = \prod_{j \in \anc(i)} \theta_j.
\end{equation}
Each factor $\theta_j$ in the product conditions on all prior rejections, so the chain rule applies directly.
The telescoping works precisely because $\theta_j$ is defined as a
conditional probability---the probability of rejecting $H_j$ given
that all of $j$'s ancestors have already rejected. No independence
assumption is needed; this is an exact application of the chain rule.
Substituting into~\eqref{eq:Pi-anc}:
\begin{equation}
  \Pr(R_i = 1)
  \leq \alpha_i \prod_{j \in \anc(i)} \theta_j
  \quad \text{for each } i \in \calH_0.
  \label{eq:Pi-final}
\end{equation}

\noindent\textbf{Step 4: Union bound.}
Steps~1--3 gave us a bound on $\Pr(R_i = 1)$ for each \emph{individual} true null $i$: the probability that the procedure reaches node $i$ and then falsely rejects it. The FWER asks about a different event --- that \emph{at least one} true null anywhere in the tree is falsely rejected. The union bound (also called Boole's inequality) bridges this gap: the probability that at least one event in a collection occurs is at most the sum of their individual probabilities.
\begin{equation}
  \mathrm{FWER}
  = \Pr \left(\bigcup_{i \in \calH_0} \{R_i = 1\}\right)
  \le \sum_{i \in \calH_0} \Pr(R_i = 1)
  \le \sum_{i \in \calH_0} \alpha_i \prod_{j \in \anc(i)} \theta_j,
\end{equation}
which is~\eqref{eq:fwer-bound}. The first inequality is the union bound; the second substitutes the per-node bound from~\eqref{eq:Pi-final}. The union bound can overcount --- if two sibling nulls are both falsely rejected in the same simulation run, that run is counted twice --- but it never undercounts, so it gives a valid upper bound on the FWER.
\end{proof}

\begin{remark}[When is the bound tight?]
The union bound is tightest when exposed true nulls are few and widely separated in the tree, and is sharp when there is only a single exposed true null. It can be loose when multiple exposed true nulls cluster at the same level, because the union bound counts the possibility that two or more nulls are both falsely rejected in the same run, even though this event is very rare.
When there is exactly one exposed true null, the union is over a
singleton, so no multiplicity slack enters; the only remaining inequality is the
per-node conditional-size bound $\theta_i \le \alpha_i$, an equality for an exact
test and strict for the conservative (discrete) randomization tests we use. When exposed
nulls sit in different branches with independent block randomizations,
the joint false-rejection probability
$\Pr(\{R_i=1\} \cap \{R_j=1\})$ equals the product of marginal
probabilities (both of order~$\alpha$), so the pairwise overlaps are
of order~$\alpha^2$ and the union bound closely approximates the
exact inclusion-exclusion formula. Independence of sibling tests is the favorable
case for this approximation; the FWER upper bounds of
Proposition~\ref{prop:fwer} and the adaptive theorems never invoke it, so
dependence among siblings can only make the realized FWER smaller than the bound.

\textbf{Example.} Consider a tree with $k = 3$ and $L = 2$ where the root is false and two of its three children are true nulls, each tested at level $\alpha = 0.05$. The union bound gives $\text{FWER} \leq 2 \times 0.05 = 0.10$, while the exact probability is $1 - (1 - 0.05)^2 = 0.0975$. The bound is nearly tight because simultaneous false rejection of both nulls is rare (probability $0.05^2 = 0.0025$). With more exposed nulls the gap widens: with $k = 10$ and 9 exposed nulls, the union bound gives $9 \times 0.05 = 0.45$ while the exact value is $1 - (1-0.05)^9 = 0.37$. In both cases the bound is conservative (never too small), which is what a valid FWER guarantee requires.
\end{remark}

The node-level bound in Proposition~\ref{prop:fwer} expresses the FWER as a sum over individual nodes. To design a practical adjustment, we need a formula that depends on \emph{depth} rather than on individual node identities --- because the adaptive $\alpha$-schedule assigns one significance level per depth, not per node. Grouping terms in~\eqref{eq:fwer-bound} by level reveals the structural quantities that govern the adjustment. The false nulls form a connected subtree rooted at the root (by the intersection structure), and if the root is true, the global null holds and weak FWER control (Theorem~\ref{thm:weakctrl_supp}) already applies.

\begin{theorem}[FWER Decomposition by Level]
\label{thm:fwer_decomp}

Let $\mathcal{B}$ be the boundary set (Definition~\ref{def:boundary_nulls})
and write $\mathcal{B}_\ell$ for the boundary nulls at depth~$\ell$.

\smallskip\noindent
(a) Each boundary null $i \in \mathcal{B}$ at depth $\ell$ satisfies
\begin{equation}
  \Pr(R_i = 1)
  \leq \alpha_\ell \prod_{j \in \anc(i)} \theta_j
  = \alpha_\ell \cdot \mathrm{PP}(i),
  \label{eq:boundary-null-contribution}
\end{equation}
where $\alpha_\ell$ is the test size at depth~$\ell$ and
$\mathrm{PP}(i) = \prod_{j \in \anc(i)} \theta_j$ is the path power of
node~$i$. Since every ancestor of a boundary null is non-null, each factor
$\theta_j = \pi_j$ in the product is the power at a non-null node.

\smallskip\noindent
(b) The total FWER satisfies
\begin{equation}
  \FWER \leq \sum_{\ell=1}^{L} \alpha_\ell \cdot S_\ell,
  \qquad
  S_\ell = \sum_{i \in \mathcal{B}_\ell}
    \prod_{j \in \anc(i)} \pi_j,
  \label{eq:fwer-boundary-sum}
\end{equation}
where $S_\ell$ is the total path power of boundary nulls at depth~$\ell$.
Each boundary null contributes through its own ancestor path product, so
this bound holds for regular and irregular trees alike.
The level-wise form~\eqref{eq:fwer-boundary-sum} is not merely a
rearrangement of the node-level bound---it is the form that enables a
practical adjustment. The adaptive $\alpha$-schedule assigns a single
test size $\alpha_\ell$ to all nodes at depth~$\ell$, so the FWER
factors into per-depth contributions $\alpha_\ell \cdot S_\ell$.
Equation~\eqref{eq:fwer-boundary-sum} is therefore the starting point for the
results that follow. By itself it does \emph{not} imply that the naive choice
$\alpha_\ell = \alpha / G_\ell$ controls FWER on every tree: natural gating,
the budget theorem, and the regular-tree telescoping theorem provide three
different ways to turn~\eqref{eq:fwer-boundary-sum} into a valid
$\alpha$-level guarantee.

\smallskip\noindent
(c) Under the global null (the root is true, and therefore every node is
true by the intersection structure), $\FWER \leq \alpha_1 \leq \alpha$.
\end{theorem}

\begin{proof}
Part~(a) applies Proposition~\ref{prop:fwer}'s per-node
bound~\eqref{eq:Pi-final} to each boundary null $i \in \mathcal{B}$.
Since every ancestor of a boundary null is non-null (by definition of
$\mathcal{B}$ and the intersection structure), $\theta_j = \pi_j$ for each
$j \in \anc(i)$.

Part~(b) combines the Boundary-Node Lemma (Lemma~\ref{lem:boundary}) with
the union bound:
\[
  \FWER
  = \Pr\!\left(\bigcup_{i \in \calH_0} \{R_i = 1\}\right)
  = \Pr\!\left(\bigcup_{i \in \mathcal{B}} \{R_i = 1\}\right)
  \leq \sum_{i \in \mathcal{B}} \Pr(R_i = 1)
  \leq \sum_{i \in \mathcal{B}}
       \alpha_{d(i)} \prod_{j \in \anc(i)} \pi_j.
\]
The first equality is the Boundary-Node Lemma. The first inequality is the
union bound. The second substitutes part~(a), where
$\alpha_{d(i)} = \alpha_\ell$ when $d(i) = \ell$ (the common test size for
all nodes at depth~$\ell$). Grouping by depth gives~\eqref{eq:fwer-boundary-sum}:
\[
  \sum_{i \in \mathcal{B}} \alpha_{d(i)} \prod_{j \in \anc(i)} \pi_j
  = \sum_{\ell=1}^{L} \alpha_\ell
    \sum_{i \in \mathcal{B}_\ell} \prod_{j \in \anc(i)} \pi_j
  = \sum_{\ell=1}^{L} \alpha_\ell \cdot S_\ell.
\]

Part~(c): if the root is true, every node is true by the intersection
structure, and $\mathcal{B} = \{\text{root}\}$. The FWER equals
$\Pr(R_{\text{root}} = 1) \leq \alpha_1 \leq \alpha$.
\end{proof}

\begin{remark}
When tests at sibling nodes are conducted using independent randomizations (as they are in block-randomized experiments with independent blocks), the exact FWER contribution from $e_\ell$ exposed nulls at level $\ell$ is $\prod_{j=1}^{\ell-1} \theta_j \cdot [1 - (1 - \alpha_\ell)^{e_\ell}]$. The union bound $e_\ell \cdot \alpha_\ell$ closely approximates $1 - (1 - \alpha_\ell)^{e_\ell}$ when $\alpha_\ell$ is small. The union bound is simpler and sufficient for our purposes; we use it throughout.
\end{remark}

\begin{remark}[Where FWER Risk Concentrates]
Under the global null, FWER is exactly $\alpha_1 \leq \alpha$ --- the stopping rule at the root provides complete protection. The more interesting case arises when the root is false but some descendants are true nulls. If root power is high, the procedure almost always descends into the tree, exposing those true nulls to testing. This is precisely what Theorem~\ref{thm:fwer_decomp}(a) quantifies: the FWER contribution from level $\ell$ is the number of exposed nulls $e_\ell$, times the test size $\alpha_\ell$, times the path product to that level. When root power is high, the path product is close to~1, and the FWER contribution from each level is approximately $e_\ell \cdot \alpha_\ell$ --- which can exceed $\alpha$ when many nulls are exposed. The next subsection formalizes the conditions under which this exposure keeps the FWER at or below $\alpha$.
\end{remark}

\subsection{Natural Gating}\label{sec:natural_gating}

In many practical settings---including the multi-site education policy trials
that motivate this paper---the tree procedure controls strong FWER without any
alpha adjustment.\corrflag{1} Splitting data across branches reduces sample size and
therefore power at each level. When power decays fast enough relative to the
branching factor, the procedure cannot reach enough true nulls to inflate the FWER.
We formalize the relationship between power and error on the tree via the
\emph{error load}.

\begin{definition}[Error Load]
\label{def:error-load}
For a tree $T$ with estimated conditional rejection probabilities
$\hat{\theta}_j$, define the \emph{estimated path power} of node~$i$ as
\[
  \widehat{\mathrm{PP}}(i) = \prod_{j \in \anc(i)} \hat{\theta}_j,
\]
with $\widehat{\mathrm{PP}}(\text{root}) = 1$. The \emph{error load} at
depth~$\ell$ is the sum of estimated path powers over all nodes at that
depth:
\begin{equation}
  G_\ell = \sum_{\substack{i \in V(T) \\ d(i) = \ell}}
           \widehat{\mathrm{PP}}(i).
  \label{eq:error-load}
\end{equation}
The total error load is $\sum_{\ell=2}^{L} G_\ell$. For a regular $k$-ary
tree where all nodes at depth~$\ell$ share the same estimated path power
$\prod_{j=1}^{\ell-1} \hat{\theta}_j$, this reduces to
$G_\ell = k^{\ell-1} \prod_{j=1}^{\ell-1} \hat{\theta}_j$.
\end{definition}

The error load $G_\ell$ measures the severity of exposure to false rejections at level $\ell$. It is the product of two quantities: $k^{\ell-1}$, the number of nodes at level $\ell$ in the full tree (the maximum number of true nulls that could be exposed there), and $\prod_{j=1}^{\ell-1} \hat{\theta}_j$, an estimated upper bound on the probability that the testing procedure reaches level $\ell$ by rejecting every ancestor along the way. When $G_\ell$ is small, either few nulls are exposed at that depth or the procedure is unlikely to reach it --- either way, level $\ell$ contributes little to the FWER. When $G_\ell$ is large, many true-null nodes may be tested at that depth; under the conditional-validity bound~\eqref{eq:cond-valid}, each such node contributes at most its path power times the node's testing level. The total level-$\ell$ contribution is therefore bounded by $G_\ell \cdot \alpha$ in the natural-gating regime.

\begin{proposition}[Natural Gating Sufficiency]
\label{prop:strong_fwer_adj}

If $\sum_{\ell=2}^{L} G_\ell \leq 1$ and each node is tested at
level~$\alpha$ (no adjustment), and the conservative-power assumption holds
($\hat{\theta}_j \geq \theta_j$ for all non-null~$j$), then for any
configuration of true and false nulls:
\[
  \mathrm{FWER} \leq \alpha.
\]
This result holds for regular and irregular trees alike.
\end{proposition}

\begin{proof}
By the FWER Decomposition (Theorem~\ref{thm:fwer_decomp}(b)):
\begin{equation}
  \FWER \leq \sum_{\ell=2}^{L} \alpha \cdot S_\ell
  = \alpha \sum_{\ell=2}^{L} S_\ell.
  \label{eq:fwer-level-sum}
\end{equation}
Each boundary null $i \in \mathcal{B}_\ell$ has all non-null ancestors.
Under the conservative-power assumption, $\hat{\theta}_j \geq \pi_j$ for
each non-null ancestor~$j$, so
\[
  \mathrm{PP}(i) = \prod_{j \in \anc(i)} \pi_j
  \leq \prod_{j \in \anc(i)} \hat{\theta}_j
  = \widehat{\mathrm{PP}}(i).
\]
Since $\mathcal{B}_\ell$ is a subset of all nodes at depth~$\ell$:
\[
  S_\ell = \sum_{i \in \mathcal{B}_\ell} \mathrm{PP}(i)
  \leq \sum_{i \in \mathcal{B}_\ell} \widehat{\mathrm{PP}}(i)
  \leq \sum_{\substack{i \in V(T) \\ d(i) = \ell}}
       \widehat{\mathrm{PP}}(i) = G_\ell.
\]
If the root is true, $\FWER = \Pr(R_{\text{root}} = 1) \leq \alpha$. If the
root is false, $S_1 = 0$ (the root has no parent, so it cannot be a
boundary null) and
\[
  \FWER \leq \alpha \sum_{\ell=2}^{L} S_\ell
  \leq \alpha \sum_{\ell=2}^{L} G_\ell \leq \alpha. \qedhere
\]
\end{proof}

\begin{remark}[Sensitivity to the conditional validity assumption]
\label{rem:sensitivity}

The FWER bounds in Proposition~\ref{prop:fwer} and
Theorem~\ref{thm:fwer_decomp}, and all results that follow from
them---Proposition~\ref{prop:strong_fwer_adj} (natural gating) and
Theorems~\ref{thm:fwer_budget} and~\ref{thm:fwer_control_adaptive}
(adaptive adjustment)---rely on two ingredients beyond weak control:
conditional validity at true nulls and conservative path-power estimates
along non-null ancestor paths. Both are hard to verify exactly from first
principles. Conditional validity can fail because parent and child tests
share experimental units, so conditioning on ancestor rejection can
inflate a child's false-rejection rate. Conservative path-power
estimation can fail because the true conditional rejection probabilities
$\theta_j$ at non-null ancestors depend on unknown effect sizes. We now
bundle both possible failures into a single sensitivity factor.
If the root is true, then $\FWER = \Pr(R_{\text{root}} = 1)$ and the
discussion reduces to sensitivity of the root test itself. The more
interesting case is the one relevant for natural gating and adaptive
adjustment: the root is false, so every boundary null lies at depth
$\ell \geq 2$.

\medskip\noindent
\textbf{A combined sensitivity factor.}
Define the inflation factor at a true-null node~$i$ as
\[
  \rho_i
  = \frac{\Pr(R_i = 1 \mid H_i \text{ true},\
    \text{all ancestors rejected})}{\alpha_i},
\]
so that conditional validity is equivalent to $\rho_i \leq 1$. For a
boundary null $i$, define also
\[
  \lambda_i
  = \frac{\mathrm{PP}(i)}{\widehat{\mathrm{PP}}(i)}
  = \frac{\prod_{j \in \anc(i)} \theta_j}
         {\prod_{j \in \anc(i)} \hat{\theta}_j},
\]
which measures how much the estimated path power understates the true path
power along the non-null ancestor path leading to~$i$. Conservative
path-power estimation is the statement $\lambda_i \leq 1$ for every
boundary null~$i$. Define the combined sensitivity factor
\[
  \kappa_i = \rho_i \lambda_i,
  \qquad
  \kappa_{\max} = \max_{i \in \mathcal{B}} \kappa_i.
\]
Then the proof of Proposition~\ref{prop:fwer} gives, for each boundary
null~$i$,
$\Pr(R_i = 1) \leq \kappa_{\max}\, \alpha_i \widehat{\mathrm{PP}}(i)$.
Summing over boundary nulls yields
\begin{equation}
  \label{eq:fwer-rho}
  \FWER
  \;\leq\;
  \kappa_{\max}
  \sum_{\ell=2}^{L} \alpha_\ell G_\ell.
\end{equation}
If the path-power estimates are conservative, then $\lambda_i \le 1$ and
$\kappa_{\max} \le \rho_{\max}$, so~\eqref{eq:fwer-rho} reduces to the
conditional-validity sensitivity analysis as a special case.

\medskip\noindent
\textbf{Robust natural gating.}
When each node is tested at level $\alpha$ (no adjustment), the
right-hand side of~\eqref{eq:fwer-rho} reduces to
$\kappa_{\max} \cdot \alpha \cdot \sum_{\ell=2}^{L} G_\ell$. Hence
$\FWER \leq \alpha$ whenever
\begin{equation}
  \label{eq:robust-gating}
  \textstyle\sum_{\ell=2}^{L} G_\ell \;\leq\; 1/\kappa_{\max}.
\end{equation}
This condition does not derive conditional validity from marginal
validity. Rather, it replaces exact conditional validity and exact
conservative path-power estimation with the weaker scalar assumption that
their combined distortion is bounded by~$\kappa_{\max}$ throughout the
tree.
For the 25 MDRC education-policy trials analyzed in the main text,
$\sum G_\ell < 0.051$, so~\eqref{eq:robust-gating} holds for any
$\kappa_{\max} < 1/0.051 \approx 19.6$. If the path-power estimates are
already conservative, this reduces to the simpler requirement
$\rho_{\max} < 19.6$.\corrflag{1}

\medskip\noindent
\textbf{Modified adaptive adjustment.}
When the error load exceeds $1/\kappa_{\max}$, the adaptive $\alpha$-schedule
can be modified to absorb the inflation by replacing
$\alpha_\ell^{\text{adj}} = \min\{\alpha,\, \alpha / G_\ell\}$ with
\[
  \alpha_\ell^{\text{adj}}
  = \min\!\left\{\alpha,\;
    \frac{\alpha}{\kappa_{\max}\, G_\ell}\right\}.
\]
The $\kappa_{\max}$ in~\eqref{eq:fwer-rho} cancels the
$\kappa_{\max}$ in the denominator of the adjusted~$\alpha_\ell$,
giving $\FWER \leq \alpha$. The cost is a more stringent rejection cutoff
at each depth, reducing power.

\medskip\noindent
\textbf{Quantifying the sensitivity factor in practice.}
Under a bivariate normal approximation for parent and child test
statistics with correlation
$r = \sqrt{n_{\text{child}} / n_{\text{parent}}}$ --- the child's units are a
subset of the parent's, and for a statistic that sums independent unit-level
contributions the correlation between the child's sum and the parent's equals
the square root of the shared fraction $n_{\text{child}} / n_{\text{parent}}$ ---
the inflation factor $\rho$ from conditioning is approximately
\[
  \rho
  \;\approx\;
  \frac{\Phi\bigl((z_t - r\, z_\pi) / \sqrt{1 - r^2}\bigr)}
       {\Phi(z_t)},
\]
where $z_t$ is the critical value for the child's test and
$z_\pi = \Phi^{-1}(\pi_{\text{parent}})$ captures the parent's
marginal power. When parent power is high
($\pi_{\text{parent}} \to 1$, so $z_\pi \to -\infty$), the
conditioning event is near-certain and $\rho \to 1$. When parent
power is near~$\alpha$, numerical evaluation of this expression gives
$\rho \approx 5.0$ ($k=2$), $3.9$ ($k=3$), and $3.4$ ($k=4$), growing as the tree
narrows; these values describe the Gaussian approximation only, not the discrete
randomization or combined-energy tests the procedure actually uses. The conservative-power assumption is
therefore least problematic whenever ancestors have high power---which is
the regime where the error load is large and the adaptive adjustment
matters. Conversely, when conditioning inflation is most strained
($\rho$ large), parent power is low and the error load is small, so the
multiplicative inflation has little effect. If the estimated path powers
are also approximately conservative, then $\kappa_{\max} \approx
\rho_{\max}$; otherwise, $\kappa_{\max}$ absorbs both sources of
misspecification.
\end{remark}

In the homogeneous regular-tree setting, the key ratio governing whether
the error load grows or shrinks with depth is:
\begin{equation}
  \frac{G_{\ell+1}}{G_\ell} = k \cdot \hat{\theta}_\ell.
  \label{eq:error-load-ratio}
\end{equation}
Here $\hat{\theta}_\ell$ denotes the common estimated conditional
rejection probability at depth~$\ell$. When $\hat{\theta}_\ell < 1/k$,
the error load decreases at level $\ell$; when
$\hat{\theta}_\ell > 1/k$, it increases. The ratio
$k \cdot \hat{\theta}_\ell$ compares the branching factor (how many new
nodes appear at the next level) against the probability that each one is
actually tested. When $\hat{\theta}_\ell$ is below $1/k$, fewer nodes are
tested at the next level than at the current level --- the tree is
``thinning out'' as the procedure descends.

\begin{corollary}[Critical Power Threshold]
In a homogeneous regular $k$-ary tree, if the estimated conditional
rejection probability at every depth satisfies $\hat{\theta}_\ell < 1/k$,
then the error load decreases geometrically with depth. Writing
$\hat{\theta}_{\max} = \max_\ell \hat{\theta}_\ell$, the total error load
satisfies
\[
  \sum_{\ell \geq 2} G_\ell \leq \frac{G_2}{1 - k\hat{\theta}_{\max}}.
\]
Here $G_2 = k \cdot \hat{\theta}_1$ when the root is false. Thus, if
$G_2 / (1 - k\hat{\theta}_{\max}) \le 1$ and the assumptions of
Proposition~\ref{prop:strong_fwer_adj} hold, natural gating suffices to
control the FWER.
\end{corollary}

The following corollary illustrates an extreme case: the root test has very low power, so the procedure almost never descends into the tree. This scenario is not desirable in practice --- we want to detect effects when they exist --- but it establishes a principle that will prove useful. The FWER depends on how far into the tree the procedure reaches, and low power at the root provides automatic FWER control by preventing the procedure from exposing any descendants to testing.

\begin{corollary}[FWER Control via a Root-Level Bound]
\label{cor:fwer_control_low_power}
Recall that the root node is indexed by $1$ and that the testing procedure above is used. Assume:
\begin{enumerate}
  \item When $H_1$ is true, its conditional Type~I error satisfies $\alpha_1 \le \alpha$. (This is easy to maintain for randomization-based tests and is the basis of the Valid Tests condition)
  \item When $H_1$ is false, its conditional power satisfies $\pi_1 \le \alpha$. (That is, the test of the root has very low power.)
\end{enumerate}
Then, for any configuration of true and false nulls in the tree,
\[
  \mathrm{FWER} \le \alpha.
\]
\end{corollary}

\begin{proof}
This is the special case $G_2 = k \cdot \theta_{\text{root}} \leq \alpha < 1$,
i.e., the root has power at most $\alpha / k$.

When $H_1$ is true, every node is true by the intersection structure, and $\FWER = \alpha_1 \leq \alpha$.

When $H_1$ is false, every true null is a descendant of the root. A Type~I error requires rejecting the root first:
\begin{align}
  \FWER
  &= \Pr(R_1 = 1)\, \Pr(\text{at least one true descendant rejected} \mid R_1 = 1) \\
  &\leq \Pr(R_1 = 1) = \pi_1 \leq \alpha. \qedhere
\end{align}
\end{proof}

Corollary~\ref{cor:fwer_control_low_power} shows that low root power controls the FWER, but the mechanism --- the procedure rarely reaches any descendants --- is exactly the outcome we want to avoid. In practice, we design experiments to have enough power to detect effects when they exist. When root power is high, the procedure almost always descends into the tree, exposing true-null siblings to testing. We now examine how power decays with depth and then show what happens when root power is high.

\begin{correctionbox}
Equation~\eqref{eq:power-decay} below and the matching step in
Algorithm~\ref{alg:adaptive} print the upper tail only. Two-sided power at
level~$\ell$ is
$\Phi\!\left(\delta\sqrt{N/k^{\ell-1}} - z_{1-\alpha/2}\right) +
\Phi\!\left(-\delta\sqrt{N/k^{\ell-1}} - z_{1-\alpha/2}\right)$.
The omitted term is at most $\alpha/2$ and shrinks as $\delta$ grows, so it
matters little at a well-powered node and much at a low-power one: at
$\delta = 0$ the printed expression returns $0.025$ where the test has size
$0.05$. Because the error load sums power over every node at a depth,
including the many low-power nodes of a wide tree, the printed formula
understates it. This is one of the four computational errors behind the
incorrect error loads (Correction notice, item~1).
\end{correctionbox}

\begin{remark}[Connection to Data Splitting]
\label{rem:power_decay}
The next display is a heuristic approximation, not part of any proof in this
section. For a standardized effect size $\delta$ and total sample size $N$,
the power at level $\ell$ when testing at significance level $\alpha$ is approximately
\begin{equation}
\theta_\ell \approx \Phi\left(\delta \sqrt{\frac{N}{k^{\ell-1}}} - z_{1-\alpha/2}\right),
\label{eq:power-decay}
\end{equation}
where $\Phi$ is the standard normal CDF. For a balanced two-arm test the standardized effect $\delta$ relates to Cohen's $d = (\mu_T - \mu_C)/\sigma$ by $\delta = d/2$, since the noncentrality of the mean-difference test is $d \sqrt{n_\ell/4} = (d/2)\sqrt{n_\ell}$ --- the half absorbing the even split of each node's $n_\ell$ units between treatment and control. Power decays because $n_\ell = N/k^{\ell-1}$, using equal splitting for simplicity. At the root, the full sample $N$ is available and power is typically high. At level $\ell$, the sample has been split among $k^{\ell-1}$ branches, so each branch uses roughly $N/k^{\ell-1}$ observations. Power drops because the test statistic's signal-to-noise ratio scales with $\sqrt{n}$: splitting the sample among three branches cuts the effective sample by a factor of three and reduces the signal-to-noise ratio by $\sqrt{3} \approx 1.7$. In a tree with $k = 3$ and moderate effect sizes, power typically drops below the critical power level $1/k = 1/3$ by level~3 or~4, so the error load decreases from that point onward.

In the multi-site education policy trials that motivate this paper, $k \approx 3$, $L \approx 3$, and power is moderate. Under these conditions, $\sum_{\ell=2}^L G_\ell < 1$ and the unadjusted procedure controls FWER --- explaining why the simulation in Section~\ref{sec:strong_fwer_sim} of the main text shows $\mathrm{FWER} \leq \alpha$ without any adjustment.\corrflag{1}

\textbf{Conditional versus marginal power.}
Formula~\eqref{eq:power-decay} gives \emph{marginal} power: the probability
of rejection at depth~$\ell$ without conditioning on ancestor outcomes. The
proofs in this section use the \emph{conditional} power
$\theta_\ell = \Pr(R_\ell = 1 \mid \text{all ancestors rejected})$. When
ancestor power is high ($\theta_{\text{parent}} \to 1$), the conditioning
event is nearly certain, so the conditional and marginal powers are nearly
equal and~\eqref{eq:power-decay} is a good approximation to $\theta_\ell$.
When ancestor power is moderate, the conditional power can exceed the marginal
by a factor~$\rho$ that depends on the correlation between parent and child
test statistics. See Remark~\ref{rem:sensitivity} for the sensitivity
framework that absorbs this gap: any inflation is captured
by~$\kappa_{\max}$, and the robust gating condition shows it is tolerable
whenever the error load is small.
\end{remark}

\begin{corollary}[FWER Inflation with High Root Power]
\label{cor:fwer_inflation}

Consider a two-level tree ($L = 2$, root plus children) where the root is
false with root power $\pi_{\text{root}} = 1$ (the root always rejects), and
$e_2$ of the root's children are true nulls ($1 \leq e_2 \leq k-1$; the
intersection structure requires at least one child to be false). If tests at
sibling nodes are independent (as in block-randomized experiments with
independent blocks) and all children are tested at level
$\alpha_{\text{lev.\,2}}$, then:
  \begin{equation}
\FWER = 1 - (1-\alpha_{\text{lev.\,2}})^{e_2} \approx e_2 \cdot \alpha_{\text{lev.\,2}} \quad \text{for small } \alpha_{\text{lev.\,2}}.
\end{equation}
Without adjustment ($\alpha_{\text{lev.\,2}} = \alpha$), FWER exceeds
$\alpha$ whenever $e_2 \geq 2$, since
$1 - (1-\alpha)^{e_2} > \alpha$ for $e_2 \geq 2$.

For deeper trees ($L \geq 3$), false rejections at the first level can
expose additional true nulls at lower levels, so the total FWER can only be
higher than this expression. The formula is therefore a lower bound on the
FWER for trees with $L \geq 2$.
\end{corollary}

This corollary shows concretely what goes wrong when root power is high. When the root always rejects, the procedure tests every child of the root. If $e_2$ of those children are true nulls, each tested at level $\alpha$, then the probability of at least one false rejection is $1 - (1 - \alpha)^{e_2}$, which grows roughly in proportion to $e_2$.

\textbf{Example.} With $k = 5$ branches and $e_2 = 4$ true-null children (one child is non-null), each tested at $\alpha = 0.05$: $\text{FWER} = 1 - 0.95^4 = 0.186$, nearly four times the nominal rate. With $k = 10$ and $e_2 = 9$: $\text{FWER} = 1 - 0.95^9 = 0.370$. The adaptive adjustment of the next subsection corrects this by testing each child at $\alpha/k$ rather than $\alpha$, so that the total stays bounded.

\textbf{Simulation: the realized error load governs FWER control.}
Table~\ref{tab:load_threshold} illustrates Proposition~\ref{prop:strong_fwer_adj}
and the gap between the realized error load and its worst-case bound
$\sum_\ell G_\ell$. We take a binary tree of depth~3 with a single non-null path
from the root to one leaf, so every off-path node is a true null and the
realized error load is the path power summed over the one off-path boundary null
at each level. Drawing $p$-values as in the main text (null $p$-values uniform,
non-null calibrated to the power at each node's sample size,
Section~\ref{sec:strong_fwer_sim}), we sweep the effect size to walk the realized
error load $D$ across~1. The unadjusted FWER tracks the bound $\alpha D$ and crosses
$\alpha$ where $D = 1$, confirming that the realized error load --- not the worst-case
sum --- governs control. The worst-case $\sum_\ell G_\ell$ exceeds~1 well before
$D$ does, so the natural-gating condition $\sum_\ell G_\ell \le 1$ is sufficient
but conservative: a tree can violate it and still control the FWER, as
Scenarios~A and~B of the main text do. The adaptive $\alpha$-schedule of the next
subsection, which budgets for the worst case, controls the FWER throughout. This
abstract single-path design isolates the error load--FWER relationship and is not a
realistic configuration; the realistic concentrated-effect scenarios are in the
main text.

\begin{table}[tb]
\centering
\small
\setlength{\tabcolsep}{6pt}
\caption{The realized error load governs the unadjusted FWER. A single non-null path runs from the root to one leaf of a binary tree of depth~3 ($N_{\text{total}} = 2000$; 10,000 simulations per row); every off-path node is a true null. As the effect size $d$ grows, the realized error load $D$ --- the path power summed over the one off-path boundary null at each level --- sweeps across~1. The unadjusted FWER tracks the bound $\alpha D$ and crosses $\alpha = 0.05$ where $D = 1$; the adaptive $\alpha$-schedule controls the FWER throughout. The worst-case $\sum_\ell G_\ell$, which counts every node as a potential null, exceeds~1 well before $D$ does: the natural-gating condition $\sum_\ell G_\ell \le 1$ is sufficient but conservative. Boldface marks an FWER above $\alpha$ plus simulation error ($\approx 0.010$).}
\label{tab:load_threshold}
\begin{tabular}{r r r | r r | r}
\toprule
& & & \multicolumn{2}{c|}{Unadjusted} & Adaptive \\
\cmidrule(lr){4-5} \cmidrule(lr){6-6}
$d$ & $D$ & $\sum_\ell G_\ell$ & FWER & $\alpha D$ & FWER \\
\midrule
0.06 & 0.31 & 0.74 & 0.015 & 0.016 & 0.015 \\
0.08 & 0.55 & 1.41 & 0.026 & 0.028 & 0.026 \\
0.10 & 0.87 & 2.42 & 0.041 & 0.043 & 0.036 \\
0.12 & 1.23 & 3.76 & 0.060 & 0.061 & 0.039 \\
0.15 & 1.76 & 6.14 & \textbf{0.086} & 0.088 & 0.041 \\
0.20 & 2.41 & 9.79 & \textbf{0.117} & 0.120 & 0.044 \\
\bottomrule
\end{tabular}
\end{table}

\subsection{Adaptive Alpha Adjustment When Gating Is Insufficient}\label{sec:adaptive_alpha_adj}

When $\sum_{\ell=2}^{L} G_\ell > 1$ --- because the tree is wide, deep, and
root power is high --- the unadjusted procedure inflates FWER. The adaptive
$\alpha$-adjustment compensates by reducing the test size at each depth in
proportion to the error load. All results in this subsection provide
\textbf{upper bounds} on FWER; practitioners should first check whether their
tree falls in the natural gating regime (Proposition~\ref{prop:strong_fwer_adj})
using \ifblind the accompanying R package\else the \texttt{compute\_error\_load()} function in the \texttt{manytestsr}
package\fi (Section~\ref{sec:software}).

We present the FWER control results in order of increasing generality.
Case~(a) covers regular $k$-ary trees without pruning, where a telescoping
argument yields the tightest bound. Case~(b) extends to irregular trees
without pruning via a budget-weighted framework. Cases~(c) and~(d) handle
branch pruning --- where the $\alpha$-schedule is recomputed on the surviving
subtree after each depth --- for both regular and irregular trees.


\paragraph{Case (a): Regular $k$-ary trees without pruning.}

\begin{remark}[Adaptive Alpha Adjustment with Power Decay on Regular $k$-ary Trees]
\label{thm:root_adjust}
For regular $k$-ary trees, the adjusted significance levels at depth $\ell$ are:
\begin{equation}
\alpha_\ell^{\text{adj}} = \min\left\{\alpha, \frac{\alpha}{G_\ell}\right\}
= \min\left\{\alpha, \frac{\alpha}{k^{\ell-1} \cdot \prod_{j=1}^{\ell-1} \hat{\theta}_j}\right\}
\label{eq:optimal_alpha_d_supp}
\end{equation}
where $\hat{\theta}_j$ is the estimated conditional rejection probability at depth $j$, computed from the power decay model~\eqref{eq:power-decay} using an estimated effect size $\hat{\delta}$. This formula divides the nominal $\alpha$ by the error load $G_\ell$ at each depth. When $G_\ell = 0$, we interpret $\alpha_\ell^{\text{adj}} = \alpha$: that depth contributes no estimated error load, so no additional shrinkage is needed there.
\end{remark}

\begin{theorem}[FWER Control Under Adaptive Adjustment --- Regular $k$-ary Trees]
\label{thm:fwer_control_adaptive}
In a regular $k$-ary tree of depth $L$, Algorithm~\ref{alg:adaptive}
controls FWER at level $\alpha$ for any configuration of null hypotheses,
provided the power estimates satisfy $\hat{\theta}_j \geq \theta_j$ for all
non-null ancestors~$j$.

The regular-tree result holds because symmetry across nodes at the same
depth lets the depth-wise contributions telescope. That extra structure is
what permits the $\alpha$-schedule $\alpha_\ell^{\mathrm{adj}} =
\alpha/G_\ell$ without any budget allocation; in general trees, where no such
telescoping is available, one must fall back on the budget-weighted theorems
below.
\end{theorem}

\begin{proof}
By the Boundary-Node Lemma (Lemma~\ref{lem:boundary}) and the FWER
Decomposition (Theorem~\ref{thm:fwer_decomp}(b)):
\[
  \FWER \leq \sum_{\ell=2}^{L}
       \alpha_\ell^{\mathrm{adj}} \cdot S_\ell.
\]
(The $\ell = 1$ term vanishes because the root is either true --- giving
$\FWER \leq \alpha$ directly --- or false, contributing no boundary nulls at
depth~1 since the root has no parent.)
Also, if $G_\ell = 0$ at some depth, then $S_\ell \le G_\ell = 0$, so that
depth contributes nothing to the sum. We may therefore restrict attention to
depths with $G_\ell > 0$.

\medskip\noindent
\textbf{Contribution of each boundary null.}
Each boundary null $i$ at depth~$\ell$ has all non-null ancestors. In a
regular $k$-ary tree, all nodes at the same depth share the same estimated
path power, so $G_\ell = k^{\ell-1} \cdot \prod_{j=1}^{\ell-1} \hat{\theta}_j$
and the adaptive $\alpha$-schedule gives:
\[
  \alpha_\ell^{\mathrm{adj}} \cdot \mathrm{PP}(i)
  = \frac{\alpha}{k^{\ell-1} \prod_{j=1}^{\ell-1} \hat{\theta}_j}
    \cdot \prod_{j=1}^{\ell-1} \pi_j
  \leq \frac{\alpha}{k^{\ell-1}},
\]
where the inequality uses $\hat{\theta}_j \geq \pi_j$ for non-null
ancestors (the conservative-power assumption). The path powers cancel
cleanly because all nodes at depth~$\ell$ share the same path power.

\medskip\noindent
\textbf{The telescoping sum.}
Let $m_\ell$ be the number of non-null nodes at depth~$\ell$ and
$e_\ell = |\mathcal{B}_\ell|$ be the number of boundary nulls at
depth~$\ell$. In a regular $k$-ary tree, each non-null parent at
depth~$\ell-1$ has exactly $k$ children:
\[
  m_\ell + e_\ell = k \cdot m_{\ell-1}, \qquad \ell = 2, \ldots, L.
\]
Summing over boundary nulls:
\[
  \FWER \leq \alpha \sum_{\ell=2}^{L} \frac{e_\ell}{k^{\ell-1}}.
\]
Substituting $e_\ell = k \cdot m_{\ell-1} - m_\ell$:
\begin{align*}
  \sum_{\ell=2}^{L} \frac{e_\ell}{k^{\ell-1}}
  &= \sum_{\ell=2}^{L} \frac{k \cdot m_{\ell-1} - m_\ell}{k^{\ell-1}}
   = \sum_{\ell=2}^{L} \frac{m_{\ell-1}}{k^{\ell-2}}
     - \sum_{\ell=2}^{L} \frac{m_\ell}{k^{\ell-1}}.
\end{align*}
Re-indexing the first sum with $s = \ell - 1$ (so $s$ runs from $1$ to
$L-1$):
\begin{align*}
  &= \sum_{s=1}^{L-1} \frac{m_s}{k^{s-1}}
     - \sum_{s=2}^{L} \frac{m_s}{k^{s-1}}
   = \frac{m_1}{k^0} - \frac{m_L}{k^{L-1}}
   = m_1 - \frac{m_L}{k^{L-1}}.
\end{align*}
Since $m_1 \leq 1$ (the root is at most one node) and $m_L / k^{L-1} \geq 0$:
\[
  m_1 - \frac{m_L}{k^{L-1}} \leq 1.
\]
Therefore $\FWER \leq \alpha \cdot 1 = \alpha$.

If the root is a true null ($m_1 = 0$), then the global null holds and
$\FWER = \Pr(R_{\text{root}} = 1) \leq \alpha$.
\end{proof}


\paragraph{Case (b): Irregular trees without pruning.}
In irregular trees, different nodes at the same depth can have different
numbers of children and different path powers. The telescoping identity
$m_\ell + e_\ell = k \cdot m_{\ell-1}$ no longer holds, and the per-depth
formula $\alpha_\ell = \alpha / G_\ell$ can fail to control FWER. The
budget-weighted framework replaces the structural cancellation with an
explicit allocation of error budget across depths.

\begin{theorem}[FWER Control Under Budget-Weighted Adaptive Alpha --- General Trees]
\label{thm:fwer_budget}
Consider a rooted tree $T$ (regular or irregular) satisfying the Stopping
Rule (Condition~\ref{cond:ap-stopping}), Valid Tests
(Condition~\ref{cond:ap-valid}), and the intersection structure. Choose
depth weights $w_2, \ldots, w_L \geq 0$ with
$\sum_{\ell=2}^{L} w_\ell \leq 1$. At each depth~$\ell$, test every
reached node at level
\[
  \alpha_\ell^{\mathrm{adj}}
  = \min\!\left\{\alpha,\;
    \frac{w_\ell \cdot \alpha}{G_\ell}\right\},
  \qquad
  G_\ell = \sum_{\substack{i \in V(T) \\ d(i)=\ell}}
    \widehat{\mathrm{PP}}(i).
\]
When $G_\ell = 0$, set $\alpha_\ell^{\mathrm{adj}} = \alpha$.
If $\hat\theta_j \geq \theta_j$ for every non-null node~$j$ (power is not
underestimated), then
\[
  \FWER \leq \alpha
\]
for any configuration of true and false nulls.
\end{theorem}

The intuition: the budget framework allocates a fraction $w_\ell$ of the
total $\alpha$ ``budget'' to each depth. At depth~$\ell$, the adjusted test
size $\alpha_\ell^{\mathrm{adj}}$ is chosen so that even in the worst case
(all nodes at depth~$\ell$ are boundary nulls), the FWER contribution from
that depth is at most $w_\ell \cdot \alpha$. Summing across depths, the
total contribution is at most $\alpha \sum w_\ell \leq \alpha$.

\begin{proof}
If the root is true, the intersection structure makes every node true. The
Stopping Rule prevents any descendant from being tested unless the root is
rejected, so $\FWER = \Pr(R_{\text{root}} = 1) \leq \alpha$.

Assume the root is false. By the Boundary-Node
Lemma~(Lemma~\ref{lem:boundary}) and the FWER Decomposition
(Theorem~\ref{thm:fwer_decomp}(b)):
\[
  \FWER \leq \sum_{\ell=2}^{L}
       \alpha_\ell^{\mathrm{adj}} \cdot S_\ell.
\]
For each boundary null $i \in \mathcal{B}_\ell$, all ancestors are non-null.
The conservative-power assumption gives
$\mathrm{PP}(i) \leq \widehat{\mathrm{PP}}(i)$. Since
$\mathcal{B}_\ell \subseteq \{i : d(i) = \ell\}$:
\[
  S_\ell
  \leq \sum_{i \in \mathcal{B}_\ell} \widehat{\mathrm{PP}}(i)
  \leq G_\ell.
\]
If $G_\ell = 0$, then this bound forces $S_\ell = 0$, so depth~$\ell$
contributes nothing to the FWER sum. For depths with $G_\ell > 0$, the
definition of $\alpha_\ell^{\mathrm{adj}}$ gives
$\alpha_\ell^{\mathrm{adj}} \leq w_\ell \cdot \alpha / G_\ell$. Therefore
\[
  \FWER \leq \sum_{\ell:\, G_\ell > 0}
    \frac{w_\ell \cdot \alpha}{G_\ell} \cdot G_\ell
  = \alpha \sum_{\ell:\, G_\ell > 0} w_\ell
  \leq \alpha. \qedhere
\]
\end{proof}

\begin{remark}[Why the per-depth formula fails for irregular trees]
\label{rem:irregular_tree_failure}
In a regular $k$-ary tree, all nodes at depth~$\ell$ share the same path
power, so the ratio $S_\ell / G_\ell = e_\ell / k^{\ell-1}$ is a pure count
ratio that telescopes. In an irregular tree, different nodes at the same
depth have different path powers. Boundary nulls can be concentrated in
branches with high path power, making $S_\ell / G_\ell$ exceed what the
count ratio would suggest. In the worst case,
$\sum_{\ell=2}^{L} S_\ell / G_\ell > 1$, and the per-depth formula
$\alpha_\ell = \alpha / G_\ell$ does not control FWER.

\textbf{Counterexample.} Consider a 3-depth tree where the root has 2
children. The left child (non-null, $\hat{\theta}_{\text{left}} = 0.8$) has
3 children; the right child (null, $\hat{\theta}_{\text{right}} = 0.05$) has
1 child. With $\hat{\theta}_{\text{root}} = 0.9$: the boundary nulls at
depth~2 (two of the left child's children) have high path power
$0.9 \times 0.8 = 0.72$ each, while the null right child's descendant
adds little to $G_2$. Across all depths,
$\sum S_\ell / G_\ell = 1.153 > 1$, giving
$\FWER \leq 1.153\alpha > \alpha$. The budget-weighted formula
(Theorem~\ref{thm:fwer_budget}) with $\sum w_\ell \leq 1$ corrects this by
construction.
\end{remark}


\paragraph{Cases (c) and (d): Trees with branch pruning.}
In the preceding results, the error load $G_\ell$ and the significance
levels $\alpha_\ell^{\mathrm{adj}}$ are computed from the full tree before
testing begins. A natural refinement is \emph{branch pruning}: after
testing at depth~$\ell - 1$, remove branches whose parents were not
rejected and recompute the error load on the surviving subtree. This can
yield less conservative significance levels at deeper depths, since the
surviving tree is narrower than the full tree.

Branch pruning changes the mathematical structure. The denominators become
data-dependent: let $\mathcal{C}_\ell$ denote the set of nodes at
depth~$\ell$ whose parents were all rejected (the ``testable'' or
``surviving'' nodes), and define the pruned error load
\[
  D_\ell = \sum_{i \in \mathcal{C}_\ell} \widehat{\mathrm{PP}}(i).
\]
Because $\mathcal{C}_\ell$ depends on test outcomes at shallower depths,
$D_\ell$ is random: it depends on the block randomizations through
the rejection decisions at depths $1, \ldots, \ell - 1$.

The telescoping proof of Theorem~\ref{thm:fwer_control_adaptive} and
the static budget proof of Theorem~\ref{thm:fwer_budget} both require
fixed denominators and do not extend to the pruned setting. The fix is
to use \emph{predictable} weights --- weights that can depend on testing
history but are determined before each depth's tests are run.

\begin{correctionbox}
Theorem~\ref{thm:fwer_budget_pruning} and Corollary~\ref{cor:switching}, both below, are false as stated
and are withdrawn (Correction notice, item~2, at the front of the paper). An
exact counterexample --- binomial enumeration, no Monte Carlo error ---
attains $\FWER = 0.063$ at $\alpha = 0.05$ with every stated hypothesis
satisfied, including exact conditional validity. The statements are retained
unchanged for the record; the proof's failing step, and a repaired schedule,
are treated in the fully corrected version.
\end{correctionbox}

\begin{theorem}[FWER Control Under Branch Pruning with Predictable Budget Weights]
\label{thm:fwer_budget_pruning}
Consider a rooted tree $T$ (regular or irregular) satisfying the Stopping
Rule, Valid Tests, and the intersection structure. At each
depth~$\ell \in \{2, \ldots, L\}$, let $\mathcal{C}_\ell$ be the set of
surviving nodes (children of rejected parents) and
$D_\ell = \sum_{i \in \mathcal{C}_\ell} \widehat{\mathrm{PP}}(i)$ the
pruned error load. Choose \emph{predictable} weights
$w_2, \ldots, w_L \geq 0$ --- where $w_\ell$ depends only on the testing
history $\mathcal{H}_{\ell-1}$ (test outcomes through depth $\ell - 1$),
not on depth-$\ell$ results --- satisfying
\[
  \sum_{\ell=2}^{L} w_\ell \leq 1 \quad \text{almost surely.}
\]
At each depth~$\ell$, test every surviving node at level
\[
  \alpha_\ell^{\mathrm{adj}}
  = \min\!\left\{\alpha,\;
    \frac{w_\ell \cdot \alpha}{D_\ell}\right\}.
\]
If $D_\ell = 0$, set $\alpha_\ell^{\mathrm{adj}} = 0$ (no nodes to test).
If $\hat\theta_j \geq \theta_j$ for every non-null node~$j$, then
\[
  \FWER \leq \alpha
\]
for any configuration of true and false nulls.
\end{theorem}

\begin{proof}
If the root is true, $\FWER = \Pr(R_{\text{root}} = 1) \leq \alpha$ as
before. Assume the root is false. Let $V_\ell$ denote the event that at
least one false rejection occurs at depth~$\ell$, and let
$\mathcal{E}_\ell \subseteq \mathcal{C}_\ell$ be the exposed true nulls
(boundary nulls that survive pruning) at depth~$\ell$.

\medskip\noindent
\textbf{Step 1: Conditional bound at each depth.}
Conditioning on the testing history $\mathcal{H}_{\ell-1}$, the union bound
and test validity give:
\[
  \Pr(V_\ell \mid \mathcal{H}_{\ell-1})
  \leq \sum_{i \in \mathcal{E}_\ell} \alpha_\ell^{\mathrm{adj}}
       \cdot \prod_{a \in \anc(i)} \theta_a.
\]

\medskip\noindent
\textbf{Step 2: Substitute the testing level.}
Using $\alpha_\ell^{\mathrm{adj}} \leq w_\ell \cdot \alpha / D_\ell$:
\[
  \Pr(V_\ell \mid \mathcal{H}_{\ell-1})
  \leq \frac{w_\ell \cdot \alpha}{D_\ell}
       \sum_{i \in \mathcal{E}_\ell}
       \prod_{a \in \anc(i)} \theta_a.
\]

\medskip\noindent
\textbf{Step 3: Conservative-power cancellation.}
For each exposed null $i$, all ancestors are non-null, so the
conservative-power assumption gives
$\prod_{a \in \anc(i)} \theta_a \leq \widehat{\mathrm{PP}}(i)$.
Since $\mathcal{E}_\ell \subseteq \mathcal{C}_\ell$:
\[
  \sum_{i \in \mathcal{E}_\ell}
  \prod_{a \in \anc(i)} \theta_a
  \leq \sum_{i \in \mathcal{C}_\ell} \widehat{\mathrm{PP}}(i)
  = D_\ell.
\]
The $D_\ell$ cancels, leaving
$\Pr(V_\ell \mid \mathcal{H}_{\ell-1}) \leq \alpha \cdot w_\ell$.

\medskip\noindent
\textbf{Step 4: Aggregate across depths.}
Since the conditional bound holds for every realization of
$\mathcal{H}_{\ell-1}$, the law of total expectation gives
$\Pr(V_\ell) = \E[\Pr(V_\ell \mid \mathcal{H}_{\ell-1})]
\leq \alpha \cdot \E[w_\ell]$.
The union bound and linearity of expectation then yield:
\[
  \FWER = \Pr\!\left(\bigcup_{\ell=2}^{L} V_\ell\right)
  \leq \sum_{\ell=2}^{L} \Pr(V_\ell)
  \leq \alpha \sum_{\ell=2}^{L} \E[w_\ell]
  = \alpha \cdot \E\!\left[\sum_{\ell=2}^{L} w_\ell\right]
  \leq \alpha \cdot 1 = \alpha.
  \qedhere
\]
\end{proof}

\begin{remark}[Relationship among the three adaptive theorems]
The three adaptive results are nested. Theorem~\ref{thm:fwer_budget}
(static budget weights) is the special case of
Theorem~\ref{thm:fwer_budget_pruning} where the tree is not pruned
($\mathcal{C}_\ell$ equals the full set of nodes at depth~$\ell$, so
$D_\ell = G_\ell$) and the weights are deterministic. Theorem~\ref{thm:fwer_control_adaptive} (regular-tree telescoping) is a
strictly tighter result for regular trees: it requires no weight allocation
at all, because the structural identity $m_\ell + e_\ell = k \cdot m_{\ell-1}$
ensures the depth-wise contributions sum to at most~1. Where the telescoping
argument applies, it dominates the budget framework; the budget framework
extends the guarantee to settings (irregular trees, branch pruning) where
telescoping is not available.\corrflag{2}
\end{remark}

\begin{correctionbox}
This corollary is withdrawn because it inherits the counterexample to
Theorem~\ref{thm:fwer_budget_pruning} above (Correction notice, item~2).
\end{correctionbox}

\begin{corollary}[Switching to nominal $\alpha$ after pruning]
\label{cor:switching}
Suppose testing has proceeded through depth~$s - 1$ under the budget
framework with remaining budget $B_s = 1 - \sum_{\ell=2}^{s-1} w_\ell$.
If the surviving subtree's error load at depth~$s$ --- the depth-wise loads
$D_\ell$ recomputed on the branches still alive, assuming no further pruning ---
satisfies
\[
  \sum_{\ell=s}^{L} D_\ell \leq B_s,
\]
then setting $\alpha_\ell^{\mathrm{adj}} = \alpha$ for all $\ell \geq s$ is valid.
Computed this way each $D_\ell$ is fixed once depth~$s$ is reached --- it uses only
the surviving subtree, not the depth-$\ell$ test outcomes --- so the condition is
checkable at depth~$s$ and upper-bounds every error load the remaining tests could
realize through further pruning. The idea: set $w_\ell = D_\ell$ for the remaining
depths, so $\alpha_\ell = \alpha \cdot D_\ell / D_\ell = \alpha$; these weights are
predictable and sum to $\sum_{\ell=s}^{L} D_\ell \leq B_s$, which does not exceed
the remaining budget, so Theorem~\ref{thm:fwer_budget_pruning} still
guarantees $\FWER \leq \alpha$.

In words: once pruning has narrowed the surviving tree enough that the
remaining error load fits within the remaining budget, the procedure can
stop adjusting and test at full nominal~$\alpha$. This is where the
pruning benefit is sharpest --- natural gating on the surviving subtree
provides the protection that explicit adjustment provided on the full tree.
\end{corollary}


\paragraph{Empirical demonstration of Regimes~3 and~4.}
Table~\ref{tab:regimes_demo} verifies the two counterexamples by simulation.
For Scenario~A (irregular tree from
Remark~\ref{rem:irregular_tree_failure}), the naive per-depth formula
$\alpha_\ell = \alpha / G_\ell$ gives empirical $\FWER \approx 0.055 > \alpha$,
while the budget-weighted formula
$\alpha_\ell = w_\ell\,\alpha / G_\ell$ gives $\FWER \approx 0.027$. For
Scenario~B (regular $k=3$, $L=3$ tree with a single non-null path), the
naive pruned formula $\alpha_\ell = \alpha / D_\ell$ gives empirical
$\FWER \approx 0.065 > \alpha$, the predictable budget-weighted formula
gives $\FWER \approx 0.033$, and the switching corollary likewise controls.\corrflag{2}
The empirical inflations of the naive procedures match the theoretical
bounds of $1.153\alpha$ and $4\alpha/3$ within Monte Carlo error
(SE $\approx 0.0007$ at $n_{\text{sim}} = 10^5$). Code:
\texttt{Analysis/sim\_regimes\_demo.R}.

\begin{correctionbox}
The ``Budget pruned'' and ``Switching corollary'' rows for Scenario~B in
Table~\ref{tab:regimes_demo} run the schedules of
Theorem~\ref{thm:fwer_budget_pruning} and Corollary~\ref{cor:switching},
both withdrawn (Correction notice, item~2). The caption and the ``Theory
bound'' column state a guarantee of $\FWER \le \alpha$ that no longer holds:
an exact counterexample attains $\FWER = 0.063$ at $\alpha = 0.05$ in a
different configuration. The simulated rates stay below $\alpha$ in the
configuration reported here, which is not evidence that either schedule
controls the FWER in general. The two naive rows and both Scenario~A rows are
unaffected.
\end{correctionbox}

\begin{table}[tb]
  \centering
  \caption{Empirical FWER for the irregular-tree (Regime 3) and pruned-tree (Regime 4) counterexamples, 100,000 replicates per row at $\alpha = 0.05$. Naive per-depth and naive pruned formulas inflate FWER above $\alpha$, matching the theoretical bounds in Remark~\ref{rem:irregular_tree_failure} and the branch-pruning counterexample. Budget-weighted alphas (Theorems~\ref{thm:fwer_budget} and~\ref{thm:fwer_budget_pruning}) and the switching corollary (Corollary~\ref{cor:switching}) control FWER at $\alpha$.\corrflag{2}}\label{tab:regimes_demo}
  \begin{tabular}{lllc}
    \hline
    Scenario & Method & Empirical FWER & Theory bound \\
    \hline
    A: irregular tree & Naive per-depth: $\alpha_\ell = \alpha / G_\ell$ & 0.0553 & $\le 1.153\,\alpha$ \\
    A: irregular tree & Budget weights: $\alpha_\ell = w_\ell\,\alpha / G_\ell$ & 0.0268 & $\le \alpha$ \\
    \hline
    B: pruned tree & Naive pruned: $\alpha_\ell = \alpha / D_\ell$ & 0.0654 & $\le 4\alpha/3$ \\
    B: pruned tree & Budget pruned: $\alpha_\ell = w_\ell\,\alpha / D_\ell$ & 0.0325 & $\le \alpha$ \\
    B: pruned tree & Switching corollary & 0.0331 & $\le \alpha$ \\
    \hline
  \end{tabular}
\end{table}


\begin{remark}[Weight choices for the budget framework]
\label{rem:weight_choices}
Several choices of weights $w_\ell$ are available:
\begin{itemize}
  \item \emph{Equal weights:} $w_\ell = 1/(L-1)$ for all $\ell \geq 2$. Simple and
    parameter-free, but ignores the tree structure and cannot exploit pruning.
  \item \emph{Error-load-proportional weights:}
    $w_\ell = G_\ell / \sum_{\ell'} G_{\ell'}$. This gives a uniform
    adjusted alpha of $\alpha / \sum G_{\ell'}$ at every depth ---
    equivalent to a single global correction by the total error load.
  \item \emph{Depth-decay weights:} $w_\ell \propto 1/\ell$ or
    $w_\ell \propto 1/2^\ell$, normalized to sum to~1. This allocates more
    budget to shallow depths where power is higher.
  \item \emph{Depth-sequential process (predictable remaining-budget;
    recommended for pruning):}\corrflag{2} Let $B_2 = 1$ be the initial budget. At each depth~$\ell$,
    choose a spending amount $c_\ell \in [0, B_\ell]$ (which may depend on
    the testing history $\mathcal{H}_{\ell-1}$) and set $w_\ell = c_\ell$,
    $B_{\ell+1} = B_\ell - c_\ell$. Then $\sum w_\ell = B_2 - B_{L+1}
    \leq 1$ automatically. A convenient parameterization: choose a fraction
    $f_\ell \in (0, 1]$ of the remaining budget, so that
    $c_\ell = f_\ell \cdot B_\ell$ and
    $B_{\ell+1} = B_\ell (1 - f_\ell)$.
    The key advantage is that $c_\ell$ can adapt to the pruned tree: after
    observing that most branches died at depth~$\ell - 1$, the procedure
    knows the surviving tree is narrow and can allocate a larger share to
    the current depth.
\end{itemize}
The choice of weights is a design parameter. All three adaptive theorems
guarantee $\FWER \leq \alpha$ for any valid weight allocation\corrflag{2}; the user
selects weights to maximize power for their application.
\end{remark}

\begin{correctionbox}
The \texttt{else} branch of the algorithm below cites
Theorem~\ref{thm:fwer_budget_pruning}, which is withdrawn (Correction notice,
item~2), and offers the pruned error load $D_\ell$ as a denominator. Readers
should not use that option. The static denominator $G_\ell$ with
deterministic weights (Theorem~\ref{thm:fwer_budget}) stands for trees without
pruning, subject to the conditional super-uniformity caveat in the notice.
This version supports no $\alpha$-schedule for the pruning regime.
\end{correctionbox}

\begin{algorithm}[H]
\caption{Adaptive Alpha Adjustment for Tree-Structured Testing}
\label{alg:adaptive}
\begin{algorithmic}[1]
\REQUIRE Tree structure $(k, L)$, sample size $N$, nominal $\alpha$, effect size estimate $\hat{\delta}$
\ENSURE Adjusted significance levels $\{\alpha_\ell^{\text{adj}}\}_{\ell=1}^L$
\STATE Calculate estimated power at each depth:
\FOR{$\ell = 1$ to $L$}
    \STATE $n_\ell \leftarrow N/k^{\ell-1}$
    \STATE $\hat{\theta}_\ell \leftarrow \Phi\left(\hat{\delta}\sqrt{n_\ell} - z_{1-\alpha/2}\right)$\corrflag{1}
\ENDFOR
\STATE Calculate error load at each depth: $G_\ell \leftarrow \sum_{i: d(i) = \ell} \widehat{\mathrm{PP}}(i)$
\IF{$\sum_{\ell=2}^L G_\ell \leq 1$}
    \STATE \textbf{Natural gating sufficient:} $\alpha_\ell^{\text{adj}} \leftarrow \alpha$ for all $\ell$
\ELSIF{tree is regular $k$-ary and pruning is not used}
    \STATE \textbf{Telescoping (Theorem~\ref{thm:fwer_control_adaptive}):}
    \STATE $\alpha_1^{\text{adj}} \leftarrow \alpha$
    \FOR{$\ell = 2$ to $L$}
        \STATE $\alpha_\ell^{\text{adj}} \leftarrow \min\{\alpha,\; \alpha / G_\ell\}$
    \ENDFOR
\ELSE
    \STATE \textbf{Budget framework (Theorems~\ref{thm:fwer_budget}/\ref{thm:fwer_budget_pruning}):}
    \STATE Choose weights $w_2, \ldots, w_L$ with $\sum w_\ell \leq 1$
    \STATE $\alpha_1^{\text{adj}} \leftarrow \alpha$
    \FOR{$\ell = 2$ to $L$}
        \STATE $D_\ell \leftarrow$ error load at depth $\ell$ (static $G_\ell$ or pruned)
        \IF{$D_\ell = 0$}
            \STATE $\alpha_\ell^{\text{adj}} \leftarrow \alpha$
        \ELSE
            \STATE $\alpha_\ell^{\text{adj}} \leftarrow \min\{\alpha,\; w_\ell \cdot \alpha / D_\ell\}$
        \ENDIF
    \ENDFOR
\ENDIF
\RETURN $\{\alpha_\ell^{\text{adj}}\}_{\ell=1}^L$
\end{algorithmic}
\end{algorithm}

\begin{remark}[Four regimes for FWER control]
\label{rem:three_regimes}
The results above identify four regimes, in order of increasing generality:
\begin{enumerate}
  \item \emph{Natural gating} ($\sum_{\ell=2}^{L} G_\ell \leq 1$): test at $\alpha$
    everywhere. No adjustment needed. Any tree shape, with or without pruning
    (Proposition~\ref{prop:strong_fwer_adj}). This regime covers all 25 MDRC
    studies analyzed in this paper, where the error load is below $0.051$.\corrflag{1}
  \item \emph{Adaptive alpha, regular trees without pruning}
    ($\sum_{\ell=2}^{L} G_\ell > 1$, regular $k$-ary): use
    $\alpha_\ell = \alpha / G_\ell$. The telescoping argument gives
    $\FWER \leq \alpha$ (Theorem~\ref{thm:fwer_control_adaptive}).
  \item \emph{Adaptive alpha, irregular trees without pruning}
    ($\sum_{\ell=2}^{L} G_\ell > 1$, irregular): use
    $\alpha_\ell = w_\ell \cdot \alpha / G_\ell$ with deterministic weights
    summing to~1 (Theorem~\ref{thm:fwer_budget}).
  \item \emph{Adaptive alpha with branch pruning} (any tree shape): use
    $\alpha_\ell = w_\ell \cdot \alpha / D_\ell$ with predictable weights
    summing to at most~1 almost surely
    (Theorem~\ref{thm:fwer_budget_pruning}). The switching corollary
    (Corollary~\ref{cor:switching}) allows the procedure to revert to
    nominal~$\alpha$ once the surviving tree is small enough.\corrflag{2}
\end{enumerate}
\ifblind The accompanying R implementation uses\else The implementation in the \texttt{manytestsr} package uses\fi the node-indexed
error load $G_\ell = \sum_{i: d(i) = \ell} \widehat{\mathrm{PP}}(i)$,
which handles irregular trees correctly, and supports branch pruning via
\texttt{alpha\_adaptive\_tree\_pruned()}. In the natural gating regime ---
where all 25 MDRC applications fall\corrflag{1} --- no adjustment is needed
regardless of tree shape.
\end{remark}

\begin{remark}[$p$-value monotonicity as algorithmic coherence]
\label{rem:monotonicity}
One could additionally require that $p_{\text{child}} \geq p_{\text{parent}}$
at every node, enforcing this by replacing each child $p$-value with
$\max(p_{\text{child}}, p_{\text{parent}})$. This would guarantee that
rejection decisions are coherent with the tree structure: if a parent is not
rejected, no child would be rejected either, so the Stopping Rule
(Condition~\ref{cond:ap-stopping}) would be satisfied automatically by the
$p$-values rather than imposed by the algorithm. Monotonicity also ensures that
power decays with depth, aligning the algorithm's behavior with the power-decay
model in~\eqref{eq:power-decay}. However, none of the results in this section
require $p$-value monotonicity: the FWER bounds in
Proposition~\ref{prop:fwer}, Theorem~\ref{thm:fwer_decomp},
Theorem~\ref{thm:fwer_budget}, Theorem~\ref{thm:fwer_control_adaptive},
and Theorem~\ref{thm:fwer_budget_pruning}\corrflag{2} hold for any valid tests
satisfying the Stopping Rule. \ifblind The accompanying R package does not enforce\else The \texttt{manytestsr} R package does not enforce\fi $p$-value monotonicity, and
the simulations in this paper confirm that control holds without monotonicity:
the weak FWER simulations yield identical results with and without it, and the
strong FWER simulations draw parent and
child $p$-values independently.
\end{remark}

\begin{remark}[Optimality of per-node Bonferroni within levels]
\label{rem:within_level_hommel}
A natural question is whether the per-node cutoff
$\alpha_\ell^{\mathrm{adj}}$ can be improved by applying a step-up procedure
(Holm, Hommel, or Hochberg) to the siblings at each depth, testing the family
at level $k \cdot \alpha_\ell^{\mathrm{adj}}$ rather than testing each node at
$\alpha_\ell^{\mathrm{adj}}$ individually. For two-level trees ($L = 2$) this
is valid and reduces to the standard bottom-up Hommel procedure. For trees with
$L \geq 3$, however, within-level step-up procedures inflate the FWER.

The mechanism is straightforward. Hommel (or Holm) at family level
$k \cdot \alpha_\ell$ can reject a specific null node with probability up to
$k \cdot \alpha_\ell$ --- $k$ times the per-node Bonferroni level --- when a
non-null sibling has a very small $p$-value. In a two-level tree this causes no
harm because there is only one depth at which false rejections can occur. In
deeper trees, a false rejection at depth~$\ell$ is an independent event from
false rejections at deeper depths, and the inflated per-node probabilities
compound across levels. A binary tree with $L = 3$ and a single non-null path
gives $\mathrm{FWER} \geq 1.5\alpha - \alpha^2/2$; in general, the worst-case
FWER converges to $\alpha \cdot k/(k-1)$ as $L$ grows. Seven alternative
strategies --- family-level accounting, closed testing within sibling groups
\autocite{marcus1976closed}, graphical approaches \autocite{bretz2009graphical},
depth-by-depth alpha trading, Simes-based depth gating, partition-based methods,
and weighted Bonferroni --- all fail for the same reason:
the adversary places nulls where the procedure is most generous, and the
cross-depth accumulation absorbs any within-depth gain.

In fact, $\alpha_\ell = \alpha / k^{\ell-1}$ is the unique
minimax-optimal per-node cutoff among depth-uniform gating procedures on
regular $k$-ary trees. The per-depth constraint (all nodes at depth~$\ell$
null) and the cross-depth constraint (single non-null path through the tree)
bind simultaneously at the adaptive $\alpha$-schedule, leaving no room for
improvement. The Bonferroni structure at each depth is not a limitation of the
proof technique but the price of the gating benefit: the tree's power comes
from not testing unreachable hypotheses, and gating requires that the
hypotheses that \emph{are} reached be tested conservatively enough for
the per-node bound to hold.
\end{remark}

\begin{remark}[The $\tau$-relaxation]
\label{rem:tau_relaxation}
In practice, one may wish to relax the adaptive adjustment at deep levels where
power is already very low. Algorithm~\ref{alg:adaptive} includes a natural
gating check ($\sum_{\ell=2}^{L} G_\ell \leq 1$) as the first step. When this check passes,
no adjustment is needed at any level. When it fails, one could in principle
relax the adjustment at levels where the \emph{remaining} error load
$\sum_{\ell' \geq \ell} G_{\ell'}$ is small. This is where natural gating
``takes over'' from the adaptive adjustment.

The formal guarantee of Theorem~\ref{thm:fwer_control_adaptive}, however,
requires the full adjustment at every level. Relaxing the adjustment at levels
where power is low but the error load has not yet been spent can inflate FWER.
For example, with $k = 10$, $L = 2$, and a false root with high power, 10 null
children each tested at nominal $\alpha = 0.05$ give $\mathrm{FWER} \approx 1 -
0.95^{10} \approx 0.40$. Users who wish to relax the adjustment beyond the
natural gating check should verify FWER control empirically via simulation.
\end{remark}

\section{Test statistics with power against diverse
alternatives.}\label{sec:test_statistic}

This approach requires, in \cite{rosenbaum2008tho}'s words, a ``first true null
hypothesis'', or an ordering of hypotheses such that non-rejection of the first
hypothesis ought to imply a non-rejection of the subsequent hypotheses. This
requirement raises questions about the test statistic used. A
difference-of-means test statistic has power to detect shifts in the mean.
But such statistics lose power in the presence of outliers
\autocite{andrews1972}, so a test on a smaller subset without outliers could
have more power than a test on the full sample. Mean-difference statistics also
record the sign of effects: if half the blocks have positive effects and half
negative, the overall mean difference may be near zero despite genuine effects
in every block. In either case, the root-level test
could fail to reject even when block-level effects are real, preventing the
procedure from descending into the tree.

\subsection{A Combined Energy Statistic}

We introduce a test statistic that captures not only mean
differences but also distributional differences between
treated and control units, in ways that are not fooled by strong
countervailing effects. The strategy is to transform each unit's outcome into
several representations --- ranks, pairwise distances
\cite{szekely2017energy}, and nonlinear transforms --- and then combine them
into a single test. Specifically, we compute the following six scores for each
unit:

\begin{enumerate}
  \tightlist

  \item The raw outcome

  \item The rank of the outcome

  \item The mean Euclidean distance between a unit's outcome and the outcomes
      of all other units in the block, following the energy-distance framework
        of \textcite{szekely2013energy, szekely2017energy, rizzo2010}.

  \item The mean Euclidean distance between the unit's outcome rank and
    the ranks of all other units.

  \item The maximum Euclidean distance between the unit and all other
    units.

  \item The hyperbolic tangent (tanh) transform of the outcome, which
    serves as an alternative to the log transform when the variable has
    many zeros \autocite{mebane2004robust}.

\end{enumerate}

Each unit thus receives six scores, each capturing a different aspect of the
outcome distribution. We combine these scores into a single test statistic and
compare the observed value to its randomization distribution, using either
large-sample chi-square
approximations \autocite{hansen2008cbs, hotelling1931gss, strasser1999asymptotic}
or exact permutation tests when the sample size is small.  

Why should this somewhat arbitrary combined test statistic have large-sample
approximations? Under the null hypothesis of no effects and given the known
randomization (within blocks in this case), a linear combination of mean
differences follows a chi-square distribution. The required variances and
covariances are functions of the randomization distribution, which the null
and the design determine \cite{hansen2008cbs, strasser1999asymptotic}.
The \verb+independence_test+ function in the \verb+coin+ R package
\cite{hothornetal:coin:2006} implements this combination. 

We do not claim that this omnibus test is optimal in all circumstances. But a
tree-structured testing procedure needs a root-level test that is sensitive to
diverse departures from the null, and the combined energy statistic meets this
requirement without excessive computational cost.

\subsection{The Combined Stephenson Rank Test}

\textcite{kim2025acic} present a test that targets the distribution of
treatment effects across units in a block-randomized experiment.
Specifically, it tests hypotheses about the additive effect $c$ at each
quantile of the experimental pool, answering questions such as ``what
proportion of units benefited from the intervention?'' If that proportion
exceeds zero, at least one unit benefited --- which counts as a detected
effect in our framework. We apply this test via \ifblind a separate R package\else the \texttt{CMRSS} R package\fi
of \textcite{kim2025acic}.

\section{MDRC Application: Full Results}\label{sec:mdrc_full}

\begin{correctionbox}
Three results in this appendix are affected. The caption of
Table~\ref{tab:mdrc_application_full} reports every study's error load as
below 0.051 and reads that as natural gating; corrected loads exceed 1 in all
25 studies at the planning effect size $d = 0.20$, so no study is in that
regime (Correction notice, item~1). At the individual-block level this
appendix sets \MDRCTopDownSingles{} top-down detections against
\MDRCBUHommelBlocks{} for bottom-up Hommel and \MDRCInhSingles{} for the
inheritance procedure. The top-down count comes from the unadjusted
procedure, and a verified bound leaves at most 33 of the
\MDRCTopDownSingles{} under a corrected $\alpha$-schedule, below
\MDRCBUHommelBlocks{}, so the block-level ordering may reverse (item~3). The
``adaptive $+$ pruning'' row of Table~\ref{tab:mdrc_inheritance} ran the
schedule of Theorem~\ref{thm:fwer_budget_pruning} with the switching of
Corollary~\ref{cor:switching}, both withdrawn (item~2). The
outcome-distribution table (Table~\ref{tab:mdrc_credits}) is not affected.
The appendix is otherwise retained unchanged pending the full revision.
\end{correctionbox}

Table~\ref{tab:mdrc_credits} reports the distribution of the outcome variable
(credits earned in the first main session after the intervention) across all 25
studies.

\begin{table}[tb]
\centering
\begingroup\small
\begin{tabular}{cccccc}
  \toprule
  Study & Blocks & min & med & mean & max \\ \midrule
CUNY Start &  21 &   0 & 0.00 & 2.28 &  27 \\ 
  ASAP Ohio &   9 &   0 & 10.00 & 9.19 &  32 \\ 
  OD PBS + Advising &  11 &   0 & 6.00 & 5.97 &  27 \\ 
  ASAP CUNY &   5 &   0 & 11.75 & 10.36 &  31 \\ 
  EASE &  26 &   0 & 0.00 & 1.42 &  21 \\ 
  LC Career &  28 &   0 & 12.00 & 11.24 &  26 \\ 
  OD LC &   4 &   0 & 12.00 & 10.96 &  26 \\ 
  PBS OH &  11 &   0 & 8.00 & 8.09 &  29 \\ 
  OD Success &   2 &   0 & 3.00 & 3.92 &  17 \\ 
  ModMath &   4 &   0 & 6.00 & 6.01 &  26 \\ 
  DPP &  44 &   0 & 0.00 & 4.33 &  16 \\ 
  PBS + Supports &  15 &   0 & 9.00 & 8.92 &  43 \\ 
  LC English &   4 &   0 & 6.00 & 6.18 &  30 \\ 
  PBS + Math &   6 &   0 & 8.00 & 8.07 &  34 \\ 
  iPASS MCCC &  42 &   0 & 6.00 & 5.52 &  28 \\ 
  DCMP &  31 &   0 & 6.00 & 6.55 &  27 \\ 
  OD Advising + Incentive &   8 &   0 & 6.00 & 5.80 &  30 \\ 
  iPASS UNCC &  35 &   0 & 14.00 & 13.14 &  30 \\ 
  PBS Variations &   6 &   0 & 0.00 & 4.02 &  29 \\ 
  iPASS Fresno State &   8 &   0 & 12.00 & 11.84 &  27 \\ 
  AtD Success Course &   8 &   0 & 7.00 & 6.80 &  25 \\ 
  LC Reading &   7 &   0 & 7.00 & 7.28 &  29 \\ 
  PBS + Advising &   2 &   0 & 14.00 & 12.92 &  21 \\ 
  OD Success (Enhanced) &   2 &   0 & 1.00 & 3.08 &  17 \\ 
  LC English + Success &   4 &   0 & 6.50 & 6.13 &  23 \\ 
   \bottomrule
\end{tabular}
\endgroup
\caption{Distribution of the outcome variable (credits earned in the first main session after the intervention) across all 25 MDRC studies. The outcome is skewed toward 0: in most studies the median student earns few credits. See Table~\ref{tab:mdrc_application} for the testing results.} 
\label{tab:mdrc_credits}
\end{table}

Table~\ref{tab:mdrc_application_full} reports the structured testing results for
all 25 MDRC studies, including the 13 studies where neither overall test
rejected at $\alpha = 0.05$. Table~\ref{tab:mdrc_application} in the main text
presents the 12 studies with at least one significant overall test.

\begin{table}[tb]
\centering
\caption{Structured (top-down) vs.\ bottom-up testing in all 25 MDRC studies. See Table~\ref{tab:mdrc_application} in the main text for the subset with at least one significant overall test. The top-down column (`Unadj') reports the unadjusted procedure, which tests each reached block at nominal $\alpha$. Each study has its own error load $\sum_\ell G_\ell$; all 25 are below 0.051, far under the natural-gating threshold of one, so no per-level $\alpha$-adjustment is needed.\corrflag[items]{1 and 3} `Ratio' is top-down node detections divided by bottom-up Hommel detections.} 
\label{tab:mdrc_application_full}
\begingroup\footnotesize\setlength{\tabcolsep}{3pt}
\resizebox{\linewidth}{!}{\begin{tabular}{ccccc|cc|cc|cc}
  \toprule
  &&& \multicolumn{2}{c}{Overall Tests} &
    \multicolumn{2}{|c}{Nodes detected} &
    \multicolumn{2}{|c}{Bottom-up} &
    \multicolumn{2}{|c}{Single blocks}
    \\
    &&&&& \multicolumn{2}{|c}{Top-down} & \multicolumn{2}{|c}{blocks} & \multicolumn{2}{|c}{Top-down}\\ Study &
  Blocks &
  $\widehat{ITT}$ &
  wilcox &
  t-test &
  Unadj &
  Ratio &
  Hommel &
  BH &
  Unadj &
  Ratio
  \\ \midrule
CUNY Start &  21 & -6.05 & 0.00 & 0.00 &  35 & 1.7$\times$ &  21 &  21 &  21 & 1.0$\times$ \\ 
  ASAP Ohio &   9 & 2.23 & 0.00 & 0.00 &  11 & 2.8$\times$ &   4 &   5 &   5 & 1.2$\times$ \\ 
  OD PBS + Advising &  11 & 1.75 & 0.00 & 0.00 &   7 & 3.5$\times$ &   2 &   3 &   4 & 2.0$\times$ \\ 
  ASAP CUNY &   5 & 2.08 & 0.00 & 0.00 &   7 & 1.8$\times$ &   4 &   4 &   4 & 1.0$\times$ \\ 
  EASE &  26 & 0.21 & 0.00 & 0.00 &   6 & 6.0$\times$ &   1 &   1 &   2 & 2.0$\times$ \\ 
  LC Career &  28 & 0.89 & 0.03 & 0.04 &   4 & $\infty$ &   0 &   0 &   2 & $\infty$ \\ 
  OD LC &   4 & 1.28 & 0.00 & 0.00 &   3 & 3.0$\times$ &   1 &   2 &   2 & 2.0$\times$ \\ 
  PBS OH &  11 & 0.71 & 0.00 & 0.00 &   3 & 3.0$\times$ &   1 &   1 &   1 & 1.0$\times$ \\ 
  OD Success &   2 & -0.67 & 0.02 & 0.02 &   2 & $\infty$ &   0 &   0 &   1 & $\infty$ \\ 
  ModMath &   4 & 0.61 & 0.00 & 0.02 &   1 & $\infty$ &   0 &   0 &   0 & --- \\ 
  DPP &  44 & 0.57 & 0.04 & 0.06 &   1 & $\infty$ &   0 &   0 &   0 & --- \\ 
  PBS + Supports &  15 & 0.57 & 0.06 & 0.13 &   0 & --- &   0 &   0 &   0 & --- \\ 
  LC English &   4 & 0.63 & 0.06 & 0.04 &   0 & --- &   0 &   0 &   0 & --- \\ 
  PBS + Math &   6 & 0.48 & 0.07 & 0.17 &   0 & --- &   0 &   0 &   0 & --- \\ 
  iPASS MCCC &  42 & -0.29 & 0.07 & 0.08 &   0 & --- &   0 &   0 &   0 & --- \\ 
  DCMP &  31 & 0.39 & 0.12 & 0.20 &   0 & --- &   0 &   0 &   0 & --- \\ 
  OD Advising + Incentive &   8 & 0.30 & 0.15 & 0.19 &   0 & --- &   0 &   0 &   0 & --- \\ 
  iPASS UNCC &  35 & 0.20 & 0.15 & 0.18 &   0 & --- &   0 &   0 &   0 & --- \\ 
  PBS Variations &   6 & 0.28 & 0.16 & 0.42 &   0 & --- &   0 &   0 &   0 & --- \\ 
  iPASS Fresno State &   8 & 0.36 & 0.38 & 0.11 &   0 & --- &   0 &   0 &   0 & --- \\ 
  AtD Success Course &   8 & -0.36 & 0.43 & 0.36 &   0 & --- &   0 &   0 &   0 & --- \\ 
  LC Reading &   7 & 0.09 & 0.57 & 0.80 &   0 & --- &   0 &   0 &   0 & --- \\ 
  PBS + Advising &   2 & -0.03 & 0.62 & 0.89 &   0 & --- &   0 &   0 &   0 & --- \\ 
  OD Success (Enhanced) &   2 & -0.14 & 0.88 & 0.70 &   0 & --- &   0 &   0 &   0 & --- \\ 
  LC English + Success &   4 & -0.05 & 0.88 & 0.87 &   0 & --- &   0 &   0 &   0 & --- \\ 
   \bottomrule
\end{tabular}}
\endgroup
\end{table}

\subsection{Comparison with the Goeman--Finos inheritance procedure}
\label{sec:mdrc_inheritance_supp}

The same hierarchical procedure we apply to the simulated designs, the
Goeman--Finos inheritance procedure \autocite{goeman2012inheritance}, applies to the
25 MDRC studies. It assigns a $p$-value to every node and controls the FWER across
the whole tree, as top-down testing does. Because these are real data with no known
truth, we compare detection counts, not error rates: the number of nodes and of
individual blocks each method rejects, the same quantities
Table~\ref{tab:mdrc_application} reports for the top-down and bottom-up methods.
Table~\ref{tab:mdrc_inheritance} totals these across the 25 studies.

The inheritance procedure detects $\MDRCInhNodes$ nodes and $\MDRCInhSingles$
individual blocks, against $\MDRCTopDownNodes$ and $\MDRCTopDownSingles$ for the
top-down procedure and $\MDRCBUHommelBlocks$ blocks for bottom-up Hommel. Two
comparisons follow. At the node level the ordering is top-down ($\MDRCTopDownNodes$)
above inheritance ($\MDRCInhNodes$) above bottom-up, which tests only individual
blocks and poses no node-level hypothesis. At the individual-block level the
inheritance procedure ($\MDRCInhSingles$) falls just below bottom-up Hommel
($\MDRCBUHommelBlocks$) and well below the top-down procedure ($\MDRCTopDownSingles$).

On these real trees the inheritance procedure is far more competitive than on the
wide simulated Detroit Promise tree of Table~\ref{tab:dpp_inheritance}, where its
leaf cutoff --- $\alpha$ divided by the number of blocks below a node --- shrank so
much that it detected fewer blocks than even bottom-up Hommel. The MDRC trees are
narrower, so that cutoff is less severe: at the individual-block level the
inheritance procedure matches bottom-up Hommel, and it adds node-level detections
bottom-up cannot make. The top-down procedure, testing each reached node at the
nominal level rather than dividing $\alpha$ across the leaves, detects the most at
both levels.\corrflag[items]{1 and 3}

\begin{table}[tb]
\centering
\caption{The Goeman--Finos inheritance procedure \autocite{goeman2012inheritance} applied to all 25 MDRC studies, totaling the number of nodes and of individual blocks each method rejects across the studies ($\alpha = 0.05$). These are real data with no known truth, so the entries are detection counts, not error rates. `Nodes detected' counts every rejected node (the overall test, intermediate groups of blocks, and individual blocks); `Single blocks' counts only individual-block rejections. The bottom-up methods test only the individual blocks, so they pose no node-level hypothesis. At the node level the ordering is top-down $>$ inheritance $>$ bottom-up; at the individual-block level the inheritance procedure matches bottom-up Hommel while the top-down procedure detects the most. The top-down rows are identical because every study's error load is below one (Table~\ref{tab:mdrc_application_full}),\corrflag{1} so the adaptive adjustment reverts to the nominal level and changes no decision.}
\label{tab:mdrc_inheritance}
\begin{tabular}{lrr}
\toprule
Method & Nodes detected & Single blocks \\
\midrule
Top-down, 2 Conditions & 80 & 42 \\
Top-down, adaptive $+$ pruning & 80 & 42 \\
Goeman--Finos inheritance & 66 & 32 \\
Bottom-up Hommel & --- & 34 \\
Bottom-up BH & --- & 37 \\
\bottomrule
\end{tabular}
\end{table}

\section{A High-Error-Load Application: The National Job Corps Study}\label{sec:njcs}

\begin{correctionbox}
Two results in this section are affected. The opening paragraph restates the
incorrect 25-study error loads, and the paragraph headed ``The diagnostic
flags the design'' compares the Job Corps load against them (Correction
notice, item~1). The Job Corps load came from the same function and is
itself understated: recomputed under Definition~\ref{def:error-load} it is
38.2, not the \NJCSSumG{} printed here, at the same anticipated per-center
effect used throughout this section. Corrected MDRC loads run from 1.12 to
40.9, each at that application's own planning effect size $d = 0.20$. Both
sides of the comparison in this paragraph were therefore wrong, and the
corrected Job Corps load sits inside the corrected MDRC range rather than two
orders of magnitude above it. The diagnostic still flags this design, and
more strongly than the section reports. And the pruned
variant's simulated FWER of \NJCSPrunedFWER{} exceeds $\alpha = 0.05$: the
text below reads this as control held ``near the nominal level,'' but with
Theorem~\ref{thm:fwer_budget_pruning} withdrawn (item~2) we now read it as
the theorem's failure appearing in operation. The adaptive schedule without
pruning (FWER \NJCSAdaptiveFWER) is not affected by the withdrawal.
\end{correctionbox}

Every one of the 25 MDRC trials in Section~\ref{sec:mdrc} is in the
natural-gating regime: the error load is below~$0.051$, so the two conditions of
Section~\ref{sec:topdownbotup} already control the FWER and no adjustment is
needed. A fair question is whether the adaptive $\alpha$-adjustment of
Section~\ref{sec:strong_local_adj} is ever called for by a real design, and
whether it controls the FWER without sacrificing too much power when it is. We
answer both with a multisite trial whose
width places it firmly on the other side of the threshold: the National Job Corps
Study (NJCS), a randomized evaluation of the federal Job Corps training program
conducted at about one hundred centers nationwide \autocite{schochet2008does}.

\paragraph{The design, and what we calibrate against.}
The NJCS randomized eligible applicants to a Job Corps program group or a control
group; each applicant is associated with the center they would attend, so the
centers are the sites at which effects may differ \autocite{hong2025multisite}.
We organize the centers under their administrative regions --- Overall, then
regions, then centers --- and test top-down. The public-use NJCS file (openICPSR
study~113269) suppresses center identity to protect the centers' confidentiality,
so it cannot place applicants in their centers; recovering the center-level
structure a full re-analysis would require needs center-identified data that the
public-use file does not carry, and that \textcite{hong2025multisite}'s
center-level analysis necessarily used. We therefore calibrate a simulation to the
study's \emph{published} design moments rather than re-running the procedure on
center-level data: \NJCSNcenters\ centers across \NJCSNregions\ administrative
regions (the region count is the Job Corps administrative grouping, which
\textcite{hong2025multisite} do not report), per-center sample sizes spanning
\NJCSNmin\ to \NJCSNmax\ (total $\NJCSNtotal$), and a between-center spread of
intent-to-treat (ITT) effects. We anticipate a mean per-center impact of
\$\NJCSMeanDollar\ per week and a between-center standard deviation of
\$\NJCSSdDollar; for comparison, \textcite{hong2025multisite} estimate the average
ITT at about \$13 per week (the weighted estimate in the public-use data is
\$\NJCSDataITTwt) and put the between-center standard deviation at \$21--\$23,
depending on whether the residualized or the raw center effect is used. About a
fifth of centers have negative effects, and the effects are scattered across
regions rather than concentrated in one --- the pattern visible in
\textcite{hong2025multisite}'s Figure~3, and the opposite of the single-college
concentration of the Detroit Promise design (Section~\ref{sec:dpp_example}).

We can check the \emph{outcome}-level moments directly against the public-use
microdata, since the outcome --- self-reported year-4 weekly earnings, the
quantity \textcite{hong2025multisite} analyze --- is not suppressed. Across the
\NJCSDataNctrl\ control-group respondents, year-4 weekly earnings average
\$\NJCSDataCtrlMean\ with a standard deviation of \$\NJCSDataCtrlSd\ (coefficient
of variation \NJCSDataCtrlCV; \NJCSDataFracZeroPct\% report zero earnings and the
upper tail is long, with skewness \NJCSDataSkew). The \$\NJCSOutcomeSd\ outcome
standard deviation and the \$\NJCSMeanDollar\ anticipated mean impact we use below
both fall within the range these data imply. The right-skewed, non-negative shape
is also why we draw control outcomes from a gamma rather than a normal
distribution: a normal would place mass below zero and understate the tail.

\paragraph{The diagnostic flags the design.}
Computed from the design alone --- the per-depth sample sizes and the anticipated
per-center effect (the \$\NJCSMeanDollar\ anticipated mean weekly impact over the
\$\NJCSOutcomeSd\ outcome standard deviation, a Cohen's $d$ of about \NJCSMeanD)
--- the error load
is $\sum_\ell G_\ell = \NJCSSumG$, well above one and two orders of magnitude above
the largest MDRC value\corrflag{1}. The diagnostic calls for adjustment before any outcome is
examined.

\paragraph{What adjustment buys.}
Table~\ref{tab:njcs} reports, across $\NJCSNsim$ simulations, the FWER, the
distribution of affected-center detections (the mean and the proportions of
simulations detecting at least one and at least two of the $\NJCSNcentersNonnull$
affected centers), and the number of affected regions detected, for the top-down
procedure (unadjusted, adaptive, and adaptive with pruning) against the
Goeman--Finos inheritance procedure \autocite{goeman2012inheritance} --- a
hierarchical FWER procedure that, like ours, tests internal nodes --- and the
bottom-up Hommel and Benjamini--Hochberg adjustments. Four findings stand out. First, the unadjusted
procedure now \emph{does} inflate the FWER, to $\NJCSUnadjFWER$: with the effects
scattered, the powerful root and region tests reject and expose true-null centers
under every active region, exactly the exposure the error load anticipates. This
is the case the Detroit Promise design could not provide, where concentrated
effects kept the realized FWER below $\alpha$ despite an error load above one. Second,
the adaptive $\alpha$-schedule restores control: its FWER falls to $\NJCSAdaptiveFWER$, and
branch pruning --- which recovers most of the lost power --- holds the FWER near
the nominal level ($\NJCSPrunedFWER$).\corrflag{2} Pruning spends the error budget
efficiently, so its FWER is just at $\alpha$ rather than far below it, which is
what lets it detect more than the strictly conservative alternatives. Third, the
pruned top-down procedure makes more true center detections than either alternative
--- $\NJCSPrunedLeafTR$ centers, against $\NJCSHommelLeafTR$ for bottom-up Hommel and
$\NJCSInheritanceLeafTR$ for the hierarchical inheritance procedure, both of which
leave much of the error budget unspent (their FWERs, $\NJCSHommelFWER$ and
$\NJCSInheritanceFWER$, fall well below $\alpha$) and so detect less. These means
are small fractions of the $\NJCSNcentersNonnull$ affected centers and can read as
near-failure, but the distribution behind them is less bleak: the pruned procedure
detects at least one affected center in $\NJCSPrunedLeafPGeOne$ of simulations and
two or more in $\NJCSPrunedLeafPGeTwo$, against $\NJCSHommelLeafPGeOne$ and
$\NJCSHommelLeafPGeTwo$ for bottom-up Hommel --- both find some affected center in
most trials, but the pruned procedure finds two or more about twice as often. The
top-down advantage is therefore not an artifact of comparing against a
hierarchy-ignoring procedure: it holds against one that uses the same tree.
Fourth, the top-down procedure identifies affected regions, a question bottom-up
testing never asks. By testing region-level hypotheses on the way down, the
FWER-valid\corrflag{2} pruned procedure rejects at least one affected region in
$\NJCSPrunedRegionPGeOne$ of simulations ($\NJCSPrunedRegionTR$ of the
$\NJCSNregionsNonnull$ regions on average), pointing to which broad groups of
centers carry effects even when the individual centers within them are hard to
identify one at a time. Bottom-up testing tests only centers, so an affected
region is never a hypothesis it evaluates. These results are
robust to the planning effect: re-running the $\alpha$-schedule at a more conservative
planning value --- the anticipated mean plus one between-center standard deviation
($d \approx \NJCSPlanningD$), which the conservative-power requirement of
Remark~\ref{rem:sensitivity} would favor --- leaves the pruned FWER and detection
counts essentially unchanged, because the predictable budget-recovery re-spends to
the same level.\corrflag{2}

\begin{table}[tb]
\centering
\caption{Detection of affected centers and regions across 10,000 simulations of a design calibrated to the National Job Corps Study (99 centers in 9 regions; error load $\sum_\ell G_\ell = 5.5$, far above the natural-gating threshold of one). The unadjusted top-down procedure inflates the FWER; the adaptive and pruned schedules restore control,\corrflag{2} and the pruned procedure detects more affected centers than either the hierarchical Goeman--Finos inheritance procedure or bottom-up Hommel. Benjamini--Hochberg controls the false discovery rate, not the FWER. `Centers' is the mean number of truly affected centers detected; `$\geq 1$' and `$\geq 2$' are the proportions of simulations detecting at least one and at least two affected centers -- the distribution behind a mean that is otherwise easy to read as near-zero. `Regions' is the mean number of truly affected regions (intermediate nodes) detected, a hypothesis bottom-up testing never poses (shown `---').}\label{tab:njcs}
\begin{tabular}{lccccc}
\toprule
Method & FWER & Centers & $\geq 1$ & $\geq 2$ & Regions \\
\midrule
Top-down, unadjusted & 0.362 & 4.97 & 0.965 & 0.914 & 3.00 \\
Top-down, adaptive & 0.007 & 0.59 & 0.457 & 0.112 & 1.25 \\
Top-down, adaptive $+$ pruning & 0.058 & 1.39 & 0.703 & 0.438 & 1.25 \\
Goeman--Finos inheritance & 0.002 & 0.39 & 0.347 & 0.041 & 1.24 \\
Bottom-up Hommel & 0.012 & 0.95 & 0.696 & 0.214 & --- \\
Bottom-up Benjamini--Hochberg & 0.033 & 1.42 & 0.726 & 0.383 & --- \\
\bottomrule
\end{tabular}
\end{table}

\paragraph{Why the depth-sequential weights.}
The pruned $\alpha$-schedule allocates the error budget across depths by a
\emph{depth-sequential} rule: beginning with the full budget, it spends a fixed
fraction of whatever budget remains at each depth and carries the rest forward
(Remark~\ref{rem:weight_choices}). We use this rule rather than equal weights
($w_\ell = 1/(L-1)$) or error-load-proportional weights for two reasons. First,
the pruning guarantee (Theorem~\ref{thm:fwer_budget_pruning}) requires the weights
to be \emph{predictable} --- each depth's share fixed before that depth's tests
are seen, using only the testing history at shallower depths --- and the
depth-sequential rule satisfies this by construction, because the share depends
only on the budget remaining after shallower depths. An allocation that instead
set a depth's weight from the realized pruned tree at that depth would not be
predictable, and the theorem's guarantee would not cover it. Second, the rule
adapts to pruning: when whole branches fail to reject and drop out, the remaining
budget concentrates on the narrower surviving subtree, recovering power that a
fixed allocation leaves unused.

\paragraph{Takeaway.}
The NJCS shows the diagnostic working at the other end of its range. Where the
MDRC trials need no adjustment,\corrflag{1} a wide, well-powered multisite
design is flagged as needing it; the unadjusted procedure indeed inflates; and the
adaptive, pruned $\alpha$-schedule restores control\corrflag{2} while still finding more affected sites
than either a flat or a hierarchical FWER procedure.

\section{Using the Accompanying R Package}\label{sec:software}
\ifblind
The software appendix demonstrating the accompanying R package on its example data is omitted to preserve anonymity and will be restored on publication.
\else

The \texttt{manytestsr} R package implements all four regimes described in
Section~\ref{sec:strong_local_adj}. This section walks through the typical
workflow on the package's built-in example dataset --- preparing data,
assessing the error load, running the tree-structured testing procedure, and
interpreting results. Every block of code below is paired with the actual
console output produced by running it, so readers can verify the exact API and
return values. The script that produced these outputs is shipped with the paper
as \texttt{Analysis/appendix\_e\_examples.R}. Full documentation and
additional vignettes ship with the package at
\texttt{https://github.com/bowers-illinois-edu/manytestsr}. The Lean~4
source files that machine-check the algebraic cores of the FWER-control
theorems ship with the paper's source files under \texttt{Theory/lean/}.

\subsection{Installation and Data Preparation}

Install \texttt{manytestsr} from GitHub and load the example data:

\begin{verbatim}
## Install from GitHub
remotes::install_github("bowers-illinois-edu/manytestsr")
library(manytestsr)
library(data.table)
library(dplyr)

data(example_dat, package = "manytestsr")
idat <- as.data.table(example_dat)
\end{verbatim}

The package expects two data frames: \texttt{idat} (individual-level data with
treatment, outcome, and block membership) and \texttt{bdat} (block-level
summaries with size \texttt{nb}, treated share \texttt{pb}, and harmonic weight
\texttt{hwt}). The block structure encodes the tree hierarchy through a dotted
factor (here \texttt{place\_year\_block}, e.g.\ \texttt{"A.1.B080"}) that
\texttt{splitSpecifiedFactor()} reads to partition the data at each depth.

\begin{verbatim}
> head(idat)
      id  year   trt    Y1    Y2   trtF place_year_block  place blockF
   <int> <int> <int> <num> <num> <fctr>           <char> <char> <fctr>
1:     1     1     0     0     0      0         A.1.B082      A   B082
2:     2     3     0     0    12      0         B.3.B094      B   B094
3:     3     1     0     0     0      0         C.1.B097      C   B097
4:     4     1     0     6     0      0         C.1.B097      C   B097
5:     5     1     0     7    11      0         B.1.B089      B   B089
6:     6     1     1     0     0      1         A.1.B080      A   B080

bdat <- idat %>%
  group_by(blockF) %>%
  summarize(
    nb = n(),
    pb = mean(trt),
    hwt = (nb / nrow(idat)) * (pb * (1 - pb)),
    place = unique(place),
    year = unique(year),
    place_year_block = factor(unique(place_year_block)),
    .groups = "drop"
  ) %>% as.data.table()

> c(n_individuals = nrow(idat), n_blocks = nrow(bdat))
n_individuals      n_blocks
         1268            44
\end{verbatim}

\subsection{Assessing the Error Load}

\begin{correctionbox}
The error-load figures printed in this subsection and the next
(\texttt{sum\_G} of 0.02782891 and 0.04104213) come from the superseded
computation of Correction notice, item~1: the \texttt{nodesize} column holds
harmonic weights rather than headcounts (root 0.22540897 in place of the 1268
individuals), the power is one-tailed, the preliminary
\texttt{find\_blocks()} pass supplies only the reached subtree (5 nodes on a
44-block tree) rather than the full design tree, and the summation is not the
one in Definition~\ref{def:error-load}. As of \texttt{manytestsr}~0.0.4.1008,
\texttt{compute\_error\_load()} stops with an error on weight-scale node
sizes, so this transcript is not reproducible under the current package. The
printed \texttt{needs\_adjustment = FALSE} values, the reading that the
example tree is in the natural-gating regime, and the sentence extending that
conclusion to all 25 MDRC studies are withdrawn. The corrected workflow
appears in the fully corrected version.
\end{correctionbox}

Before running the full procedure, check whether the tree falls in the natural
gating regime. The error load is computed from design quantities alone --- the
per-node sizes and harmonic weights, together with a design-stage effect size ---
and does not use the realized $p$-values, so the adjust-or-not decision is fixed
by the design rather than the observed outcomes.\corrflag{1} We obtain \texttt{node\_dat}
from a preliminary \texttt{find\_blocks()} pass and call
\texttt{compute\_error\_load()} with an estimated effect size:

\begin{verbatim}
## Preliminary pass: build the node structure used by compute_error_load().
prelim <- find_blocks(
  idat = idat, bdat = bdat,
  blockid = "blockF",
  splitfn = splitSpecifiedFactor,
  pfn = pIndepDist,
  fmla = Y1 ~ trtF | blockF,
  splitby = "place_year_block",
  parallel = "no", thealpha = 0.05
)

## Natural-gating check at a small effect size (Cohen's d = 0.10).
el_small <- compute_error_load(node_dat = prelim$node_dat,
                               delta_hat = 0.10 / 2)

> el_small$sum_G
[1] 0.02782891
> el_small$needs_adjustment
[1] FALSE

## Same check at a much larger effect size (Cohen's d = 0.80).
el_big <- compute_error_load(node_dat = prelim$node_dat,
                             delta_hat = 0.80 / 2)

> el_big$sum_G
[1] 0.04104213
> el_big$needs_adjustment
[1] FALSE
\end{verbatim}

Both \texttt{sum\_G} values are far below 1, so this tree is firmly in the
natural-gating regime: no alpha adjustment is needed
(Proposition~\ref{prop:strong_fwer_adj}). All 25 MDRC studies in the main paper
fall in this same regime.\corrflag{1} Trees with \texttt{needs\_adjustment = TRUE} require
the adaptive $\alpha$-schedule shown in the next subsection.

\subsection{Running the Procedure}

The main function is \texttt{find\_blocks()}, which performs the top-down
split-and-test procedure. The basic call uses nominal $\alpha$ at every depth:

\begin{verbatim}
results <- find_blocks(
  idat = idat, bdat = bdat,
  blockid = "blockF",
  splitfn = splitSpecifiedFactor,
  pfn = pIndepDist,
  fmla = Y1 ~ trtF | blockF,
  splitby = "place_year_block",
  parallel = "no", thealpha = 0.05
)

> nrow(results$node_dat)             ## total nodes
[1] 5
> sort(unique(results$node_dat$depth))
[1] 1 2 3
> sum(results$node_dat$testable)     ## nodes that opened (p <= alpha)
[1] 1
> head(results$node_dat[, .(nodenum, parent, depth, nodesize, p, a, testable)])
   nodenum parent depth   nodesize           p     a testable
     <int>  <int> <int>      <num>       <num> <num>   <lgcl>
1:       1      0     1 0.22540897 0.025808348  0.05     TRUE
2:       3      1     2 0.10303381 0.009180512  0.05       NA
3:       2      1     2 0.12237516 0.077261731  0.05       NA
4:       4      3     3 0.04599040 0.134086224  0.05       NA
5:       5      3     3 0.05704341 0.115312234  0.05       NA
\end{verbatim}

For trees outside the natural-gating regime, wrap the depth-indexed adaptive
$\alpha$-schedule with \texttt{alpha\_adaptive\_tree()} and pass the resulting closure
to \texttt{find\_blocks()} as \texttt{alphafn}:

\begin{verbatim}
alphafn_adapt <- alpha_adaptive_tree(
  node_dat = prelim$node_dat,
  delta_hat = 0.80 / 2
)

## Inspect the depth-indexed alpha schedule the closure will use.
> compute_adaptive_alphas_tree(node_dat = prelim$node_dat,
+   delta_hat = 0.80 / 2, thealpha = 0.05)
   1    2    3
0.05 0.05 0.05
attr(,"error_load")$sum_G
[1] 0.04104213
attr(,"error_load")$needs_adjustment
[1] FALSE
\end{verbatim}

Because the example tree satisfies natural gating,\corrflag{1} the $\alpha$-schedule returns
to nominal $\alpha = 0.05$ at every depth. On a tree where
\texttt{needs\_adjustment = TRUE}, this same call returns a strictly smaller
$\alpha_\ell$ at deeper levels.

\subsection{Interpreting Results}

\texttt{report\_detections()} returns the block-level detections in FWER mode.
Multiple blocks under the same parent share that parent's adjusted $p$-value
\texttt{pfinalb}, reflecting that the procedure rejected the joint null at
the parent rather than separating the children:

\begin{verbatim}
detections <- report_detections(results$bdat, fwer = TRUE,
                                alpha = 0.05, blockid = "blockF")

> sum(detections$hit)               ## number of detected blocks
[1] 9
> detections[hit == TRUE, .(blockF, pfinalb, hit)]
   blockF   pfinalb    hit
   <fctr>     <num> <lgcl>
1:   B080 0.1340862   TRUE
2:   B081 0.1340862   TRUE
3:   B082 0.1340862   TRUE
4:   B083 0.1340862   TRUE
5:   B084 0.1153122   TRUE
6:   B085 0.1153122   TRUE
7:   B086 0.1153122   TRUE
8:   B087 0.1153122   TRUE
9:   B088 0.1153122   TRUE
\end{verbatim}

To visualize the testing tree:

\begin{verbatim}
tree <- make_results_tree(results, block_id = "blockF")

> names(tree)
[1] "nodes"        "graph"        "test_summary"
> names(tree$nodes)
 [1] "parent"      "p"           "a"           "group_id"    "node_id"
 [6] "batch"       "testable"    "nodenum"     "depth"       "nodesize"
[11] "alpha3"      "name"        "parent_name" "node_number" "hit"
[16] "node_type"   "blocks"      "node_label"  "num_blocks"  "nonnull"

## Render with:  make_results_ggraph(tree$graph)
\end{verbatim}

The tree visualization shows each node's $p$-value and rejection status,
making the gating structure visible: branches where the parent was not
rejected appear grayed out, and the viewer can trace which paths the
procedure descended.
\fi

\section{Detection-Count Distributions in the Detroit Promise Simulation}\label{sec:dpp_pmf_supp}

The main text summarizes how many of the nine affected blocks the top-down and
bottom-up procedures detect in the Detroit Promise simulation as means and, in
Table~\ref{tab:big_sim_xtab}, as the probabilities of detecting at least one and
at least two blocks. Figure~\ref{fig:dpp_pmf} shows the full distribution behind
those summaries at $d = 0.30$. Both procedures leave most simulations with no
detection, but the shapes differ. The top-down ``2 Conditions'' procedure places
real mass on detecting two to seven blocks at once --- it recovers the HFCC
cluster whenever the root and college tests reject --- and detects at least one
block in $\DPPHighThreeRulesLeafPGeOne$ of simulations. Bottom-up Hommel reaches
at most one block in practice and detects any block in only
$\DPPHighBUHommelLeafPGeOne$ of simulations. The means
($\DPPHighThreeRulesLeafTR$ versus $\DPPHighBUHommelLeafTR$) compress this
difference in shape into a single number that can read as ``neither procedure
works.'' The distribution shows instead that the top-down procedure works in a
specific way: it recovers the affected cluster about a third of the time, while
the bottom-up procedure almost never detects more than an isolated block.

\begin{figure}[tb]
\centering
\includegraphics[width=\linewidth]{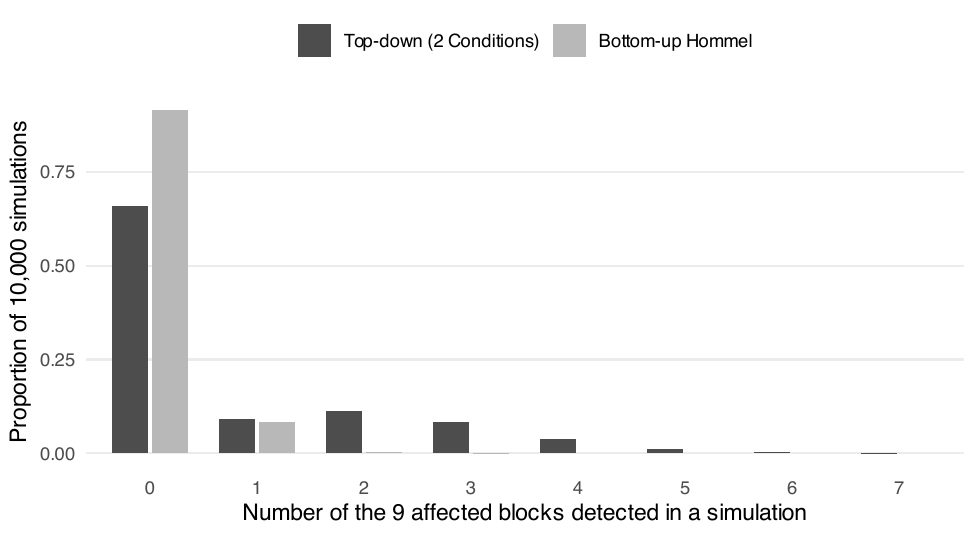}
\caption{Distribution of the number of the nine affected blocks detected per
simulation in the Detroit Promise design at $d = 0.30$ (10,000 simulations), for
the top-down ``2 Conditions'' procedure (unadjusted) and bottom-up Hommel. Both
concentrate mass at zero, but the top-down procedure recovers the affected
cluster --- two to seven blocks at once --- in a sizable fraction of simulations,
while bottom-up detection rarely exceeds a single block. This is the distribution
that the means and the at-least-$k$ probabilities in
Table~\ref{tab:big_sim_xtab} summarize.}
\label{fig:dpp_pmf}
\end{figure}

\subsection{Comparison with the Goeman--Finos inheritance procedure}
\label{sec:dpp_inheritance_supp}

The main text contrasts the top-down procedure with bottom-up Hommel. A third
method is the Goeman--Finos inheritance procedure
\autocite{goeman2012inheritance}, a hierarchical procedure that, like top-down
testing, assigns a $p$-value to every node and controls the FWER across the whole
tree. Table~\ref{tab:dpp_inheritance} runs all four procedures on the simulated
Detroit Promise design at $d = 0.30$, scoring the three tree procedures on the
same materialized trees so the node- and leaf-level numbers are comparable. Two
points stand out. First, the unadjusted ``2 Conditions'' gate controls the FWER
among individual blocks ($\DPPHighTwoCondLeafFWER$) but not among nodes: counting
the rejection of a whole null college as the false rejection it is raises the FWER
to $\DPPHighTwoCondNodeFWER$, above $\alpha$, exactly as the error load of
$\DPPHighErrorLoad$ predicts; the adaptive schedule brings the node FWER back to
$\DPPHighAdaptNodeFWER$. Second, the inheritance procedure controls the FWER but
divides $\alpha$ by the block count at the leaves, so on this wide tree it detects
only $\DPPHighInhLeafTR$ blocks --- fewer than even bottom-up Hommel --- while the
adaptive top-down procedure, at the same node-level FWER, detects several times
more. The inheritance procedure is a general-purpose method for any pre-specified
hierarchy of hypotheses; it was not designed for the block-randomized, data-split
trees here, where the top-down procedure can exploit the design's concentration of
power. The comparison measures the gain from that specialization, not a deficiency
in the inheritance procedure.

\begin{table}[tb]
\centering
\caption{The Goeman--Finos inheritance procedure \autocite{goeman2012inheritance} compared with the top-down procedure and bottom-up Hommel on the simulated Detroit Promise design at Cohen's $d = 0.30$ (10,000 simulations, $\alpha = 0.05$; 9 of 44 blocks non-null, all in one college). `Node FWER' counts a false rejection at any node, including the rejection of a whole null college; `Leaf FWER' counts only false rejections of individual blocks. `Nodes detected' and `Blocks detected' are the mean numbers of truly non-null nodes (at any level) and blocks rejected per simulation. The three tree methods are scored on the same materialized trees; bottom-up Hommel tests only the 44 blocks and so has no node-level entries. At this error load (8.85) the unadjusted procedure's node FWER is inflated above $\alpha$ -- rejecting a null college is a false rejection that no leaf records -- while the adaptive schedule restores control; the inheritance procedure controls the FWER but detects fewer blocks than even bottom-up Hommel.}
\label{tab:dpp_inheritance}
\begin{tabular}{lrrrr}
\toprule
Method & Node FWER & Leaf FWER & Nodes detected & Blocks detected \\
\midrule
Top-down, 2 Conditions & 0.077 & 0.038 & 2.13 & 0.80 \\
Top-down, adaptive $\alpha$ & 0.019 & 0.003 & 1.34 & 0.22 \\
Goeman--Finos inheritance & 0.019 & 0.000 & 1.09 & 0.04 \\
Bottom-up Hommel & --- & 0.024 & --- & 0.09 \\
\bottomrule
\end{tabular}
\end{table}

\printbibliography[title=References]

\end{document}